%% file: main_preprint.tex
\documentclass[11pt]{article}
\usepackage[dvips,letterpaper,margin=1.0in]{geometry}

\usepackage{authblk}
\usepackage{graphicx}
\usepackage{amsmath}
\usepackage{amssymb}
\usepackage{sty}
\usepackage{lmodern}
\usepackage[utf8]{inputenc} 
\usepackage[T1]{fontenc}    
\usepackage{url}            
\usepackage{booktabs}       
\usepackage{amsfonts}       
\usepackage{nicefrac}       
\usepackage{microtype}      
\usepackage{xcolor}         
\usepackage{macros}         
\usepackage{algorithm}
\usepackage{algorithmic}
\usepackage{eqparbox}
\usepackage{comment}
\usepackage{setspace}
\usepackage{caption}
\usepackage{tikz}
\usepackage{float}
\usetikzlibrary{positioning, arrows.meta, shapes.geometric, decorations.text, decorations.pathreplacing}
\usepackage[
backend=biber,
backref=true,
style=alphabetic,
citestyle=alphabetic,
maxnames=99,
maxalphanames=4,
minalphanames=4,
date=year]{biblatex}
\DefineBibliographyStrings{english}{%
  backrefpage = {page},
  backrefpages = {pages},
}
\AtEveryBibitem{
    \clearname{editor}
    \clearlist{publisher}
    \clearlist{location}
    \clearfield{url}
    \clearfield{urldate}
    \clearfield{review}
    \clearfield{doi}
    \clearfield{series}
    \clearfield{pages}
    \clearfield{issn}
    \clearfield{isbn}
}
\title{Unforgeable Watermarks for Language Models
\\via Robust Signatures}

\author[a]{Huijia Lin}
\author[a]{Kameron Shahabi}
\author[b]{Min Jae Song}
\affil[a]{Paul G.~Allen School of Computer Science \& Engineering,
University of Washington}
\affil[b]{Data Science Institute,
University of Chicago}
\affil[ ]{\texttt {\{rachel,kshahabi\}@cs.washington.edu, minjae.song@uchicago.edu}}

\begin{document}

\maketitle

\begin{abstract}
Language models now routinely produce text that is difficult to distinguish from human writing, raising the need for robust tools to verify content provenance. Watermarking has emerged as a promising countermeasure, with existing work largely focused on model quality preservation and robust detection. However,  current schemes provide limited protection against false attribution.

We strengthen the notion of soundness by introducing two novel guarantees: unforgeability and recoverability. Unforgeability prevents adversaries from crafting false positives, texts that are far from any output from the watermarked model but are nonetheless flagged as watermarked.
Recoverability provides an additional layer of protection: whenever a watermark is detected, the detector identifies the source text from which the flagged content was derived. Together, these properties strengthen content ownership by linking content exclusively to its generating model, enabling secure attribution and fine-grained traceability.

We construct the first undetectable watermarking scheme that is robust, unforgeable, and recoverable with respect to substitutions (i.e., perturbations in Hamming metric). The key technical ingredient is a new cryptographic primitive called robust (or recoverable) digital signatures,
which allow verification of messages that are close to signed ones, while preventing forgery of messages that are far from all previously signed messages. We show that any standard digital signature scheme can be boosted to a robust one using property-preserving hash functions (Boyle, LaVigne, and Vaikuntanathan, ITCS 2019).
\end{abstract}

\newpage
\tableofcontents
\newpage

\input{sections/intro-new}

\input{sections/technical-overview}
\input{sections/watermarking-primitive}
\input{sections/preliminaries}

\input{sections/rds-formal}
\input{sections/watermarking-formal}
\newpage
\printbibliography
\newpage

\appendix
\input{sections/appendix}

\end{document}

%% file: sections/intro-new.tex
\section{Introduction}

The rapid rise of generative AI has made reliable content provenance---maintaining verifiable records of the origin and history of generated content---an urgent priority. The remarkable quality of content output by GenAI models, including large language models (LLMs) like GPT-4~\cite{achiam2023gpt} and diffusion models such as Stable Diffusion~\cite{rombach2022high}, increasingly blurs the line between authentic and synthetic content. Without reliable provenance, trust in digital ecosystems erodes, enabling large-scale misinformation, disputed ownership, and other serious risks \cite{barrett2023identifying}. This paper focuses on two avenues of content provenance: 
\begin{itemize}
    \item \textit{Tracing content to its source}: Attributing generated content to the model, model owner, or prompter responsible for generating it. 
    \item \textit{Tracking edits}: Recovering the modifications or edits made to generated content by computationally bounded parties.
    \end{itemize}
A recent line of work~\cite{aaronson2022my,kirchenbauer2023watermark,gunn2024undetectable} has introduced \emph{watermarking} as a tool for provenance. However, existing notions of watermarking typically target a narrow objective and offer guarantees that fall short of reliable source tracing. Moreover, prior approaches do not consider the challenge of tracking edits. In this work, we address these limitations in the context of language models.

\paragraph{Existing watermarking guarantees.} A watermarking scheme for a language model is a protocol for embedding hidden, verifiable watermarks into generated content. It consists of three algorithms, $(\Gen, \Wat, \Ver)$. The key-generation algorithm jointly samples watermarking and verification keys $(\wk, \vk) \gets \Gen(1^\lambda)$. As we will discuss shortly, we may either consider a public or secret verification key. The watermarking algorithm $\Wat(\wk, \cdot)$ acts as the interface to the language model, receiving prompts $x \in \{0,1\}^*$ and producing watermarked responses $y \in \{0,1\}^*$. The verification algorithm $\Ver(\vk, \cdot)$ decides whether a candidate string $\zeta \in \{0, 1\}^*$ is watermarked. Existing constructions~\cite{christ2024pseudorandom,golowich2024edit,fairoze2023publicly} satisfy the following security properties:
\begin{itemize}
    \item \textit{Robustness:} No efficient attacker with oracle access to $\Wat_{\wk}$\footnote{In the public verification key setting, the attacker also receives the verification key $\vk$; in the secret key setting they instead have oracle access to $\Ver_{\vk}$.} can produce text that is \textit{close}, under a specified closeness predicate, to one of the responses $y \leftarrow \Wat_{\wk}(x)$ that they observe, but for which verification fails. In other words, it is hard to remove the watermark through mild modifications.
    \item \textit{Soundness:} No attacker acting \textit{independently} of the watermarking and verification keys $(\wk, \vk)$ can generate text that is incorrectly flagged as watermarked.
\end{itemize}
Preserving model quality is also a central goal. The strongest notion, \emph{undetectability}~\cite[Definition 9]{christ2024undetectable}, requires that without the watermarking and verification keys, no efficient observer can distinguish the oracle $\Wat_{\wk}$ from an oracle that produces standard (unwatermarked) outputs from the underlying model. In this work, we adopt undetectability as a core requirement for watermarking schemes.

At a high level, robustness and soundness address complementary failures: robustness protects against adversaries attempting to trigger false negatives, while soundness defends against false positives. However, the threat model for soundness is weaker than that for robustness. Soundness adversaries must act \emph{without knowledge} of the watermarking and verification keys. This restriction is information-theoretic, implying that the attacker has no access to watermarked content (which depends on the keys) and cannot test verification decisions on candidate text. In contrast, robustness adversaries may adaptively query watermarked outputs and are provided access to the verification procedure \footnote{Because soundness is defined under a weaker threat model, existing watermarking schemes typically achieve soundness against computationally unbounded adversaries, whereas robustness has been achieved either against computationally bounded \cite{christ2024pseudorandom,alrabiah2024ideal} or unbounded \cite{fairoze2023publicly} adversaries.}.

Despite this gap, these definitions do well in capturing the original use cases of language watermarking: distinguishing AI-generated content from honest, human-written text. Prior work has focused on two settings: \emph{public verification} \cite{fairoze2023publicly}, where the verification key is released so that anyone can distinguish content, and \emph{private verification} \cite{christ2024pseudorandom,golowich2024edit}, where a model owner keeps their verification key secret and uses it to test their own model's content---for example, to avoid retraining their model on its own prior outputs while scraping the web for training data \cite{shumailov2024ai}. The robustness adversary then receives appropriate access to the verification key---full access to $\vk$ in the public key setting and oracle access to $\Ver_{\vk}$ in the secret key setting---as well as oracle access to $\Wat_{\sk}$, reflecting the capabilities of a motivated attacker who is actively trying to remove the watermark. The soundness adversary, on the other hand, gets no access to the keys, which is justified under this traditional setup, as honest humans are unlikely to produce text correlated with randomly chosen watermarking and verification keys, and thus would rarely trigger false positives. Moreover, some applications can naturally tolerate a small number of false positives. For instance, in the example of filtering a model’s own outputs from its training data, it may be acceptable if a few human written texts are mistakenly identified as watermarked and filtered out of the training data.

\paragraph{Tracing the source: Unforgeability.} Unfortunately, existing soundness guarantees are insufficient for reliably tracing content back to its source. This is an important functionality for applications that require accountability, where it is crucial to identify parties responsible for generating harmful content, or important to assign credit to deserving parties that produced content.

One might hope that verifying text under a particular verification key, together with soundness, would serve as proof that the content originated from the owner of the corresponding watermarking key. However, enabling such fine-grained attribution motivates a new class of attacks. Adversaries may observe previously generated watermarked content, detection outcomes (in the secret key setting), or the verification key itself (in the public key setting) to produce content that differs significantly from any genuine model output, yet is still flagged as watermarked. We refer to such attacks as \emph{forgeries}. For example, an adversary may want to forge a model owner's watermark onto harmful content, to frame the model owner for producing it. Soundness alone, which only protects against key-independent adversaries, provides no protection against these key-dependent forgeries. 

To address this, we propose a strengthened notion of soundness, (robust) unforgeability, which allows for dependence on keys:
\begin{itemize}
\item \textit{Unforgeability:} No efficient attacker with a watermarking oracle $\Wat_{\wk}$ can produce content that is \textit{far} (under some closeness predicate) from all the responses $y \leftarrow \Wat_{\wk}(x)$ that they observe, yet is still flagged as watermarked. 
\end{itemize}
The attacker is also given access to the verification key $\vk$ (in the public key setting) or to a verification oracle $\Ver_{\vk}$ (in the secret key setting). Figure~\ref{fig:guarantees} illustrates the distinction between soundness and unforgeability. With an unforgeable watermarking scheme, content that verifies under a particular verification key $\vk$ can be reliably attributed to the owner of the corresponding watermarking key $\wk$, since it must \emph{closely resemble} content generated using that key.

\begin{figure}[htbp]
    \centering
    \includegraphics[width=0.8\textwidth]{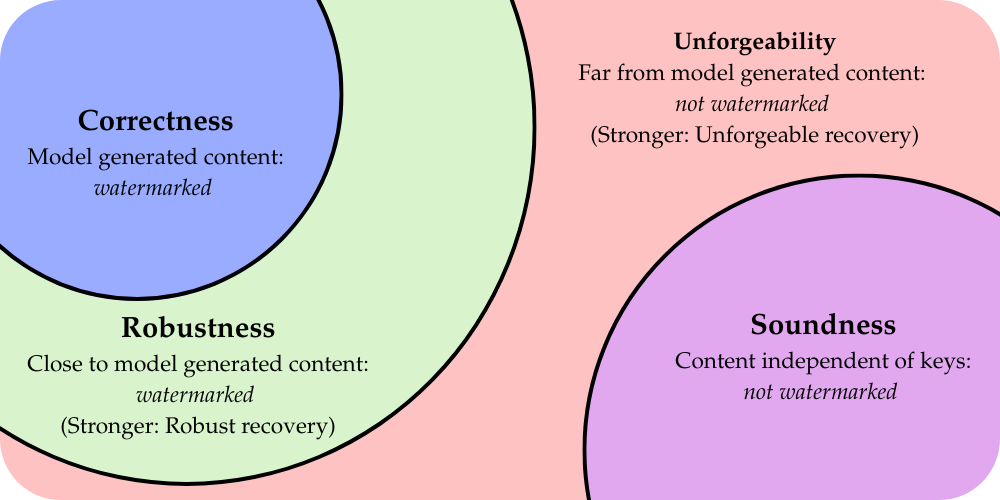} 
    \caption{Illustration of watermarking guarantees. Correctness and soundness (blue and purple) are statistical guarantees, whereas robustness and unforgeability (green and red) can be computational. Consequently, the green and red regions may \textit{contain} unwatermarked and watermarked content respectively, but no efficient adversary can find them.}
    \label{fig:guarantees} 
\end{figure}
Our notion of unforgeability is inspired by digital signatures but differs in that it only prevents forgery of content that is significantly \textit{far} from the model’s outputs, rather than content that is merely different. This design ensures that unforgeability complements robustness. A strict digital signature–style definition would conflict with robustness, since even mildly modified text that still verifies could be treated as a forgery. Prior work on watermark unforgeability~\cite{christ2024pseudorandom,fairoze2023publicly} illustrates this trade-off: these approaches leverage digital signatures but do not maintain robustness.

\paragraph{Tracking edits: Recoverability.} Prior works have yet to consider tracking edits to generated text. We introduce a formal notion, (robust) {recoverability}, which allows doing so. Recoverable watermarks include an additional polynomial time algorithm $\Rec$ which outputs a list, with the property:
\begin{itemize}
    \item \textit{Robust recoverability:} If an efficient adversary produces text $\zeta$ that is \emph{close} (under a closeness predicate) to some response $y \leftarrow \Wat_{\wk}(x)$, then the list $\Rec(\vk, \zeta)$ contains a substring of $y$ whose length is nearly the length of $\zeta$\footnote{We allow a small loss in length between the recovered substring from $y$ and the adversary's output $\zeta$.}.
\end{itemize}
A key benefit of recoverability is that it reduces ambiguity in defining \emph{closeness} using a boolean predicate $\Phi : \{0, 1\}^* \times \{0, 1\}^* \rightarrow \{0, 1\}$. In an unforgeable watermarking scheme, adversarially generated text verifies only if it is close, according to $\Phi$, to one of the adversary's observed model-generated responses. However, text may be close under \emph{that} predicate yet still not warrant attribution to the watermarker, as it seems unlikely that any single predicate can capture diverse attribution decisions. Recoverability addresses this by enabling reconstruction of the original output, allowing the verifier to make a more informed judgement about whether to attribute content to the watermarker or to a user who modified that content.

Recoverability can be viewed as a strengthening of robustness, protecting against false negatives by ensuring that we \emph{recover}, rather than merely verify, when text is close to a watermarked output. It does \emph{not}, however, protect against false positives in the recovered list. Nonetheless, some applications of recoverability (e.g., aiding in attribution decisions) may require the guarantee that the recovered substrings genuinely originate from the source model. We formalize this with \emph{unforgeable recoverability}: if an efficient adversary produces text $\zeta$ such that the recovery algorithm outputs a non-empty list ($\Rec(\vk, \zeta) \neq \perp$), then each recovered $\xi \in \Rec(\vk, \zeta)$ must be a substring of some response $y \leftarrow \Wat_{\wk}(x)$ that the adversary approximately copied.
\paragraph{Applications of unforgeability and recoverability.}
Unforgeability enables reliably tracing content back to a specific \emph{watermarking key}. This supports the following applications:

\begin{itemize}
\item \textit{Strengthened claims of ownership and accountability.} Unforgeable watermarking guarantees that content attributed to a particular watermarker (e.g., a model owner or a prompter) genuinely originates from them. This enables reliable authorship claims (e.g., for AI-generated art) and supports accountability (e.g., determining whether harmful content was actually produced by a certain watermarker).
\item \textit{Improved soundness for honest users.} A user may interact with a watermarked model honestly---for instance, by observing and substantially modifying watermarked outputs to avoid direct copying. Soundness offers no protection against false positives in such cases, since the resulting text may still be correlated with the watermarking key. Moreover, as language models and watermarked outputs become increasingly widespread, even fully human-written text may inadvertently be correlated with watermarking keys \cite{yakura2024empirical}, undermining the independence assumption that soundness relies on. Our threat model protects such honest users from false positives.
\item \textit{Tracking memorization or distillation.} Consider a scenario in which a new language model, trained on content scraped from the internet, may have memorized outputs from a different, watermarked model. Unforgeability guarantees that if the new language model's output verifies as watermarked under that watermarking key, then it is genuinely close to an output produced by the watermarked model. This provides strong evidence for claims of model distillation (as alleged in DeepSeek's use of GPT-4).
\end{itemize}
Recoverability, in turn, allows for tracking the modifications made to generated content produced by a certain watermarking key. This allows:
\begin{itemize}
    \item \textit{Finer-grained attribution.} For example, in the public verification key setting, an instructor may recover AI-generated text that a student approximately copied, in order to determine whether that student used AI in a permissible manner. Unforgeable recoverability ensures that the student truly observed edit the recovered AI-generated text. 
    \item \textit{Tracing edit histories.} A verifier can reconstruct content that has been modified or built upon by multiple watermarkers by following the recovery chain using their respective verification keys. 
\end{itemize}

\subsection{Our results}
Our contributions fall under two categories: conceptual and technical. Conceptually, we introduce new public and secret key notions of unforgeability and recoverability for watermarking schemes (Section \ref{sec:watermarking-primitive}). Towards attaining these properties, we define new cryptographic tools---robust and recoverable digital signatures (Section \ref{sec:rds})---which may be of independent interest. We show these can be constructed from property-preserving hashes (PPHs) \cite{boyle2019adversarially}, a primitive that, although previously studied, has largely been treated as a standalone tool. Our work demonstrates a novel use of PPH as a building block in cryptographic constructions, uncovering an unexplored connection to AI content provenance.

On the technical side, we develop a general framework for constructing robust and recoverable digital signatures, as well as recoverable, robust, and unforgeable watermarking schemes for language model. Our framework accomodates any general predicate defining closeness. We demonstrate its applicability by obtaining both public and secret key watermarking schemes for the Hamming predicate (Section~\ref{sec:watermark-instantiation}). Along the way, we identify and formalize a useful property of PPHs, which we call \emph{difference recovery} (Section~\ref{sec:difference-recovery}), and show that every PPH for the Hamming predicate must satisfy this property.

\paragraph{Defining robust, unforgeable, and recoverable watermarking.} In Section \ref{sec:watermarking-primitive}, we present our new notions of security for watermarking schemes. We model adversaries as probabilistic polynomial-time (PPT) algorithms that may interact with oracles. For an adversary $\cA$ with oracle access to $\M$, 
the \emph{transcript} of their interaction, denoted $\Trans(\cA^\M)$, is the set 
of all input--output pairs observed by $\cA$. For example, if $\M$ is a watermarking oracle $\Wat_{\sk}$, the transcript is the set of all the prompt-response pairs $(x, y)$ such that $\cA$ queried $\Wat_{\sk}(x)$ and received response $y$.

A robust and unforgeable watermarking scheme ensures that adversarially generated text verifies if and only if it is close to one of the adversary's observed responses. This can be decomposed into two conditions: the ``if'' direction corresponds to robustness, while the ``only if'' direction corresponds to unforgeability. The meaning of closeness is determined by a two-input predicate $\Phi$ that outputs 1 iff its inputs are deemed close.
\begin{definition}[Robust and unforgeable watermarking, informal version of Definitions \ref{def:robustness-general} and \ref{def:unforgeability-general}]
    For any closeness predicate $\Phi$, a watermarking scheme $(\Gen, \Wat, \Ver)$ for a language model is public key $\Phi$-robust and unforgeable if for any PPT adversary $\mathcal{A}$, 
    \[\Big(\Ver(\vk, \zeta)  = 1\Big) \quad \iff \quad \Big(\exists (x, y) \in \Trans(\cA^{\Wat_{\sk}}) \text{ such that } \Phi(\zeta, y) = 1\Big)\;,\]
    with probability at least $1-\negl(\lambda)$ over $(\sk, \vk) \leftarrow \Gen(1^\lambda)$ and $\zeta \leftarrow \mathcal{A}^{\Wat_{\sk}}(\vk)$. 
    
    Secret key robustness and unforgeability are the same, except $\cA$ gets oracle access to $\Ver_{\vk}$, rather than full access to $\vk$.
\end{definition}
Note that the adversary's output $\zeta \in \{0,1\}^*$ and watermarked outputs $y \in \{0,1\}^*$ may have different lengths. With this in mind, one can think of $\Phi(\zeta, y) = 1$ semantically as 
\[\text{``$\zeta$ is approximately contained in $y$,''}\]
meaning that there exists a substring of $y$ to which $\zeta$ is similar to. In other words, $\Phi(\zeta, y) = 1$ indicates that the adversary approximately copied a portion of the response $y$.

We also define \emph{recoverable} watermarking. Robust recovery states that a verifier can recover an original substring from \emph{every} watermarked response that an adversary approximately copies from. 
\begin{definition}[Robust recovery, informal version of Definition \ref{def:recoverability-general}]
     For any closeness predicate $\Phi$, a \emph{recoverable} watermarking scheme $(\Gen, \Wat, \Rec)$ for a language model is public key $\Phi$-robust if for any PPT adversary $\cA$, 
    \begin{align*}
&\forall (x, y) \in \Trans(\mathcal{A}^{\Wat_{\sk}}) \text{ such that } \Phi(\zeta, y) = 1, \\
& \Rec(\vk, \zeta) \text{ contains a substring of } y \text{ that is the same size as $\zeta$, up to some loss,}
\end{align*}
with probability at least $1-\negl(\lambda)$ over $(\sk, \vk) \leftarrow \Gen(1^\lambda)$ and $\zeta \leftarrow \mathcal{A}^{\Wat_{\sk}}(\vk)$.
\end{definition}
We require that any recovered substring be sufficiently long relative to the adversary’s output, meaning that $\Rec(\vk,\zeta)$ contains a substring of size at least $(|\zeta| - \mathrm{loss})$ for a prescribed loss parameter. 

Robust recoverability does not protect against false positives in the recovered list; this restriction is enforced separately by the following unforgeability requirement, which ensures that every recovered string corresponds to a genuine model output that was approximately copied by the adversary.
\begin{definition}[Unforgeable recovery, informal version of Definition \ref{def:unforgeable-recoverability-general}]
     For any closeness predicate $\Phi$, a \emph{recoverable} watermarking scheme $(\Gen, \Wat, \Rec)$ for a language model is public key $\Phi$-unforgeable if for any PPT adversary $\cA$, 
    \begin{align*}
    &\text{for every recovery } \xi \in \Rec(\vk, \zeta)\\
   & \exists (x, y) \in \Trans(\cA^\Wat) \text{ such that } \xi \text{ is a substring of } y \text{ and }\Phi(\zeta, y) = 1,
\end{align*}
with probability at least $1-\negl(\lambda)$ over $(\sk, \vk) \leftarrow \Gen(1^\lambda)$ and $\zeta \leftarrow \mathcal{A}^{\Wat_{\sk}}(\vk)$.
\end{definition}
Recoverability can thus be viewed as a \emph{witness} variant of robustness: rather than merely deciding whether a given text $\zeta$ is watermarked, $\Rec$ outputs a certificate justifying that decision---namely a list of model-generated texts whose similarity to $\zeta$ can be verified. Any recoverable watermarking scheme that satisfies robustness and unforgeability immediately yields a standard robust and unforgeable watermarking scheme: one can simply verify content whenever the recovery algorithm recovers any text, $\Ver(\vk, \zeta) := \big(\Rec(\vk, \zeta) \neq \bot \big).$

\paragraph{\textbf{Realizing our definitions.}} 
Our main technical result is a general framework for constructing watermarking schemes that are recoverable, robust, and unforgeable. This framework relies on two building blocks. The first is \emph{robust steganography}, an abstraction which we observe is implicitly achieved by existing watermarking schemes. The second is \emph{robust and recoverable digital signatures}, a new notion that we introduce and construct.

\itulineparagraph{Robust block steganography.} 

Robust block steganography schemes (Definition~\ref{def:block-steg}) are closely related to watermarking schemes, but additionally allow embedding messages into generated content while preserving output quality. The term \emph{block} refers to their fixed-size interface: for a given security parameter, the message size $k$ and output size $n$ are fixed, and the scheme embeds messages $m \in \{0,1\}^k$ into \emph{blocks} $y \in \{0,1\}^n$.

More formally, for message and block sizes $k,n : \mathbb{N} \to \mathbb{N}$, a block steganography scheme for a language model consists of algorithms $(\Gen, \Embed, \Dec)$. The key-generation algorithm samples secret and decoding keys jointly as $(\sk,\dk) \leftarrow \Gen(1^\lambda)$. The embedding algorithm $\Embed(\sk,\pi,m) \in \{0,1\}^{n(\lambda)}$ takes a context $\pi \in \{0,1\}^*$ and a message $m \in \{0,1\}^{k(\lambda)}$ and outputs a block $y \in \{0, 1\}^{n(\lambda)}$. The decoding algorithm $\Dec(\dk,\zeta) \in \{0,1\}^{k(\lambda)} \cup \{\perp\}$ outputs either a decoded message or a failure symbol.

Robustness with respect to a closeness predicate $\Phi$ requires that if a block $\zeta \in \{0, 1\}^{n(\lambda)}$ (chosen by an efficient adversary) is close to a valid embedding $y$ of a message $m$, then $\zeta$ will decode to that message,
\[\Big(\exists \big((\pi, m), y\big) \in \Trans(\cA^{\Embed_{\sk}}) \text{ such that } \Phi(\zeta, y) = 1\Big) \implies \Big(\Dec(\dk, \zeta) = m\Big)\;,\]
Block steganography should also preserve the quality of generated text. We enforce this by requiring they satisfy undetectability \cite{christ2024undetectable} with respect to the underlying language model\footnote{Informally, this ensures that no efficient adversary, without the watermarking and verification keys, can distinguish an oracle producing embedded outputs $\Embed_{\sk}$ from one producing unembedded outputs according to the underlying language model.}. 

We emphasize that robustness and undetectability are the only properties we require. In particular, block steganography schemes need not prevent false positives nor support recovery, and therefore do not provide unforgeability or recoverability. This matches existing constructions~\cite{christ2024pseudorandom,alrabiah2024ideal,fairoze2023publicly}, which achieve neither\footnote{These constructions do satisfy soundness, but our results hold even when using an unsound block steganography scheme.}. Nonetheless, we show that block steganography can be generically composed with robust and recoverable digital signatures to obtain watermarking schemes that are simultaneously robust, recoverable, and unforgeable.

\itulineparagraph{Robust signatures.} We introduce new notions of robust and recoverable signatures (Section \ref{sec:rds}). A robust signature scheme verifies on an adversarially chosen message-signature pair $(\zeta, \sigma)$ if and only if $\zeta$ is \emph{close} to a message $y$ that was previously signed with $\sigma$. A signature scheme is further said to be \emph{recoverable} if, from any adversarially chosen message-signature pair $(\zeta, \sigma)$, the recovery algorithm recovers a message $y$ if and only if (1) $\zeta$ is close to $y$ and (2) $y$ is a message that was previously signed with $\sigma$.
\begin{definition}[Robust/Recoverable signatures, informal version of Definitions \ref{def:rds-formal} and \ref{def:rds-recoverability-formal}]
    For any closeness predicate $\Phi$, a digital signature scheme $(\Gen, \Sign, \Ver)$ is $\Phi$-robust if for any PPT adversary $\mathcal{A}$, 
    \[\Big(\Ver(\pk, \zeta, \sigma) = 1\Big)  \iff  \Big(\exists y \text{ such that }(y, \sigma) \in \Trans(\cA^{\Sign_{\sk}}) \text{ and } \Phi(\zeta, y) = 1\Big) \;.\]
 A digital signature scheme $(\Gen, \Sign, \Rec)$ is further said to be $\Phi$-recoverable if
 \[\Big(\Rec(\pk, \zeta, \sigma) \text{ outputs }y\Big)   \iff \Big((y, \sigma) \in \Trans(\cA^{\Sign_{\sk}}) \text{ and } \Phi(\zeta, y) = 1\Big)\;.\]
  Here, $(\sk, \pk) \gets \Gen(1^\lambda)$ and $(\zeta, \sigma) \gets \mathcal{A}^{\Sign_{\sk}}(\pk)$,
 \end{definition}
We refer to the reverse direction of the above statements—--namely, that  $(\zeta,\sigma)$ verifies (or recovers) only if $\zeta$ is close to a message $y$ signed with $\sigma$—--as \emph{strong} unforgeability. Strongly unforgeable robust signatures generalize standard strong UF-CMA digital signatures \cite{bellare2007two}, which are recovered as the special case where $\Phi$ is the equality predicate. The term \emph{strong} reflects the fact that this notion guarantees not only that signatures cannot be forged on unsigned messages, but also that no \emph{new} signatures can be produced even for messages that were previously signed under a different signature.
\itulineparagraph{Informal guarantees of our framework.} Using these two tools, we design a general approach for constructing robust and unforgeable watermarking schemes. To understand its guarantees, it will be helpful to think of a string $y \in \{0,1\}^*$ as being broken into consecutive, non-overlapping $n$-bit \emph{blocks}
$(y_1, \ldots, y_{\lfloor |y|/n \rfloor})$, where $y_t$ denotes the $t$-th contiguous, non-overlapping $n$-bit substring of $y$. We make this notation precise in Section~\ref{sec:notation}.

At a high level, for a security parameter $\lambda$, our construction produces watermarked text $y \leftarrow \Wat_{\sk}(x)$ by repeatedly generating blocks $y_t$ of size $n = n(\lambda)$ using a block steganography scheme. Each block embeds the robust signature $\sigma_{t-1}$ of the preceding block $y_{t-1}$. Verification on candidate text $\zeta = (\zeta_1, \ldots, \zeta_{\floor{|\zeta|/n}})$ proceeds by attempting to decode a signature from each block $\zeta_t$ and verifying it against the preceding block $\zeta_{t-1}$\footnote{Section~\ref{sec:wm-overview} contains a more detailed overview.}. Intuitively, this means that $\zeta$ will verify if and only if its blocks are close to a \emph{sequence} of back-to-back blocks from a watermarked response $y$.

We capture this intuition using the \emph{every-block-close} lifting of a predicate $\Phi$, $\EBC{\Phi}{n}$. This predicate deems that ``$\zeta$ is approximately contained in $y$'' if it is blockwise $\Phi$-close to a consecutive sequence of blocks in $y$,
 \[
\EBC{\Phi}{n}(\zeta, y) :\;
\left(
\begin{aligned}
&\text{Parse } \zeta \text{ and } y \text{ into } n\text{-bit blocks } (\zeta_1, \ldots, \zeta_{\floor{|\zeta|/n}}) \text{ and } (y_1, \ldots, y_{\floor{|y|/n}}).\\
&\text{Output 1 iff } \exists t \in \NN \text{ such that } \forall j \in [\floor{|\zeta|/n}], \;\Phi(\zeta_j, y_{t+j-1}) = 1.
\end{aligned}
\right).
\]
Figure \ref{fig:ebcfig} captures the predicate in action. 
\begin{figure}[htbp]
    \centering
    \includegraphics[width=\textwidth]{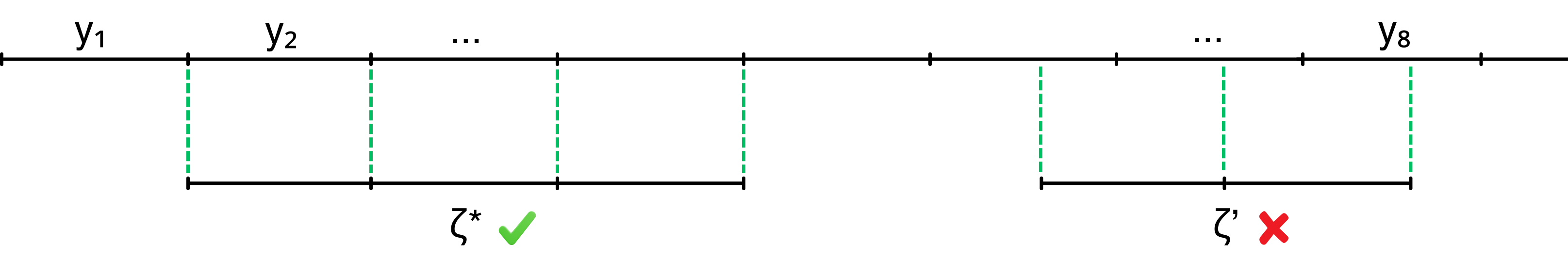} 
    \caption{Illustration of $\EBC{\Phi}{n}$-closeness. Green dotted lines indicate $\Phi$-closeness. The string $\zeta^*$ is $\EBC{\Phi}{n}$-close to $y$ because its blocks are $\Phi$-close to three consecutive blocks of $y$, whereas $\zeta'$ may be far: although $\zeta'$ is close to a substring of $y$ of length $2n$, this substring is not aligned with two \emph{blocks} of $y$.
    \label{fig:ebcfig}}
\end{figure}

By combining a $\Phi$-robust block watermarking scheme, and a $\Phi$-robust and recoverable digital signature scheme, we obtain a watermarking scheme that is recoverable, robust, and unforgeable with respect to the every-block-close lifting of  $\Phi$. Whether the watermark is public or secret key depends on the underlying block steganography scheme.
\begin{theorem}[General watermarking framework, informal version of Theorem \ref{thm:watermark-main}]
\label{thm:watermark-informal}
Let $Q$ be any language model, $\Phi$ be any closeness predicate, and $n : \mathbb{N} \rightarrow \mathbb{N}$ be any polynomial block size. Assume there exists both
\begin{itemize}
\item A (secret or public key) steganography scheme for $Q$ that is robust to $\Phi$ and has block size $n$.
\item A digital signature scheme that is robust to $\Phi$.
\end{itemize}
Then there exists a watermarking scheme for $Q$ that is (secret or public key) robust and unforgeable\footnote{Against adversaries outputting sufficiently long strings $\zeta$.} with respect to the closeness predicate $\EBC{\Phi}{n}$. If the digital signature scheme is additionally recoverable, then there exists a recoverable watermarking scheme for the same predicate.
\end{theorem}

\begin{remark}[Subtlety between robustness and unforgeability]
\label{rem:gap}
While the informal theorem suggests that an input $\zeta$ verifies if and only if it is $\EBC{\Phi}{n}$-close to a watermarked response, there is a subtle distinction in the formal guarantees:
\begin{itemize}
\item For \emph{robustness}, verification succeeds if the \emph{entire} string 
\((\zeta_1, \ldots, \zeta_{\lfloor |\zeta|/n \rfloor})\) is $\EBC{\Phi}{n}$-close to a watermarked response $y$.
\item For \emph{unforgeability}, verification succeeds only if the substring consisting of all but the last block, 
\((\zeta_1, \ldots, \zeta_{\lfloor |\zeta|/n \rfloor - 1})\), is $\EBC{\Phi}{n}$-close to a watermarked response $y$.
\end{itemize}
This gap is made explicit in our full Theorems (Theorems~\ref{thm:watermark-main} and~\ref{thm:r-block-watermark-full}).
\end{remark}
\itulineparagraph{Instantiating the framework.} We instantiate our framework to obtain a Hamming robust and unforgeable watermarking scheme. Prior work has already constructed both secret and public key robust block steganography schemes for any  language model with sufficiently high entropy, when the predicate $\Phi$ is the Hamming predicate, which indicates if its inputs are close in Hamming distance. In particular, assuming the existence of ideal pseudorandom codes that are robust to a constant fraction of Hamming errors---shown by Alrabiah et al.~\cite{alrabiah2024ideal} to follow from the subexponential hardness of exact learning parity with noise~\cite{jain2012commitments}---the watermarking construction of Christ and Gunn \cite{christ2024pseudorandom} yields a \emph{secret key} steganography scheme that is robust to a constant fraction of Hamming errors (Theorem \ref{thm:secret-key-block-wm}). Separately, in the random oracle model, Fairoze et al.~\cite{fairoze2023publicly} construct a watermarking scheme for high-entropy language models, which can be viewed as a \emph{public-key} steganography scheme that is robust to a subconstant fraction of Hamming errors (with a caveat, see the predicate $\Phi_\delta$ defined in Theorem~\ref{thm:public-key-block-wm}). 

Finally, in this work we show that robust and recoverable digital signatures for a constant fraction of Hamming errors can be constructed under standard assumptions, namely the existence of collision-resistant hash functions and strong UF-CMA digital signatures (Corollary~\ref{cor:recoverable-signature-concrete}).

Using these results, we obtain both secret and public key recoverable, robust, and unforgeable watermarking schemes with respect to the every-block-close lifting of the Hamming predicate (Corollaries \ref{cor:sk-watermark-concrete} and \ref{cor:pk-watermark-concrete}).
\paragraph{Constructing robust and recoverable signatures.} 
In \refSc{sec:rds-construction}, we construct robust and recoverable digital signatures with respect to a constant fraction of Hamming errors. We do so by instantiating the following general framework, which leverages property-preserving hashing, introduced by Boyle et al. \cite{boyle2019adversarially}.

\itulineparagraph{Property-preserving hashing} For any predicate $\Phi$ and hash size $k : \mathbb{N} \rightarrow \mathbb{N}$, a $\Phi$-property preserving hash family (PPH) is a hashing scheme $(\Sample, \Eval)$, that allows computing the predicate value of two inputs using only their hashes of length $k(\lambda)$. That is,
\[\Eval(h, h(y), h(y')) = \Phi(y, y')\]
with probability at least $1 - \negl(\lambda)$ over $h \leftarrow \Sample(1^\lambda)$ and $(y, y') \leftarrow \cA(h)$. Importantly, $\Eval$ is adaptively secure.

We also observe that, although not explicitly stated, current PPH constructions for the Hamming predicate $\Ham_\delta$ \cite{fleischhacker2022property,holmgren2022nearly} also allow recovering the bitwise difference $\Delta = x \,\oplus\, x'$ from a hash $h(x)$ and an input $x'$ such that $\Phi(x, x') = 1$\footnote{In fact, these constructions satisfy a stronger property, allowing recovery of the bitwise difference from just the hashes $h(x)$ and $h(x')$. However, we show that the weaker property, where the recovery algorithm gets $x'$ in the clear, suffices for watermarking.}. We formalize this as \emph{difference recovery} and show that this is no coincidence: every PPH for the Hamming predicate must satisfy it (Proposition~\ref{prop:pph-recover-app}).

\begin{definition}[Difference recovery, informal version of Definition \ref{def:PPH-recoverability}]
A PPH $(\Sample, \Eval)$ for a predicate $\Phi$ is said to satisfy \emph{difference recovery} if there exists a PPT recovery algorithm $\Rec$ such that
\begin{align*}
    \Big(\Phi(x, x') = 1\Big) \quad \implies \quad \Big(\Rec(h, h(x), x') = x' \oplus x\Big)\;,
\end{align*}
holds with probability at least $1 - \negl(\lambda)$ over $h \leftarrow \Sample(1^\lambda)$ and $(x, x') \leftarrow \cA(h)$.  
\end{definition}
Difference recovery is an unexplored aspect of PPHs, and our work illustrates its utility through a compelling application: recoverable watermarking. 
\itulineparagraph{General framework for robust signatures.} We show that any standard strongly unforgeable digital signature \cite{bellare2007two} can be supplemented with a $\Phi$ property preserving hash (with difference recovery) in order to build a $\Phi$ robust (recoverable) digital signature.
\begin{theorem}[General robust signature framework, informal statement of Theorem \ref{thm:rds-appdx-main}]
Let $\Phi$ be any closeness predicate and let $k, \ell : \mathbb{N} \rightarrow \mathbb{N}$ be any polynomials.
Assume there exists 
\begin{itemize}
    \item A $\Phi$-PPH with hash size $k$.
    \item A strongly unforgeable digital signature scheme with signature size $\ell$.
\end{itemize}
Then there exists a $\Phi$-robust digital signature scheme with signature size $k + \ell$. Moreover, if the PPH is difference recovering, then the signature scheme is $\Phi$-recoverable.
\end{theorem}

Instantiating this framework with the difference-recovering Hamming PPH of Holmgren et al.~\cite{holmgren2022nearly}, which relies on a special homomorphic hash function, yields a signature scheme that is robust and recoverable against a constant fraction of Hamming errors, with signature size essentially a constant fraction of the message length (Corollary~\ref{cor:recoverable-signature-concrete}). For our purposes, however, we observe that it suffices to use a standard collision resistance hash function rather than the homomorphic one, giving us a Hamming recoverable digital signature scheme from CRHFs and strong UF-CMA signatures.

\subsection{Related work}
\label{sec:related-work}

\paragraph{Watermarking schemes for generative models.} Watermarking text has a history that predates the current era of large language models, with early work~\cite{topkara2005natural,atallah2001natural,atallah2002natural} focused on embedding signals into \emph{fixed} objects in a way that is hard to remove, while preserving their perceived quality. More relevant to our setting, which views language models as \emph{randomized sampling algorithms}, is the recent line of work~\cite{aaronson2022my,kirchenbauer2023watermark,kirchenbauer2023reliability,kuditipudi2023robust,christ2024undetectable,zhao2023provable} that develops watermarking techniques from this perspective, typically by embedding correlations with a watermarking key at the token level during the sampling process. 

While many existing methods are distortion-free at the level of a single response~\cite{aaronson2022my,kuditipudi2023robust}, they may introduce noticeable distortions when the watermarked model is queried repeatedly---across multiple rounds of interaction with arbitrary seed prompts. The strongest notion of model quality preservation that guards against such cumulative distortions, known as watermark \emph{undetectability}, was introduced by Christ et al.~\cite{christ2024undetectable}. 

Among undetectable watermarking schemes, only a select few guarantee unforgeability~\cite{fairoze2023publicly,christ2024pseudorandom}. These methods embed digital signatures into text in a computationally undetectable manner, thereby achieving both undetectability and unforgeability. However, these watermarks lack robustness. Our robust digital signatures can be integrated into either of these schemes. Recently, undetectable watermarking schemes for \emph{image} models have been proposed by Gunn et al.~\cite{gunn2024undetectable}, representing a notable step toward achieving undetectability in generative models beyond discrete autoregressive sampling. For a broader overview of the field, we refer the reader to the surveys~\cite{liu2024survey,zhao2024sok}.

\paragraph{Steganography.} Steganography, the practice of hiding messages within content drawn from some distribution (referred to as a channel), has a long history of scientific study. Early work by Cachin \cite{cachin1998information} considered a noiseless setting, where a message only needed to be hidden within unaltered content sampled from the channel. This study was motivated by Simmons’ Prisoners’ Problem \cite{simmons1984prisoners}, in which two prisoners in separate cells wish to communicate hidden messages by embedding them into letters visible to the warden. Hopper et al. \cite{hopper2002provably} further studied this problem and introduced robust steganography. However, in these initial works, the decoder was often assumed to know the distribution from which the content was drawn.

In contrast, hiding messages in AI-generated content presents additional challenges. The decoder, who sees content later down the line, does not know the prompt used to generate it and therefore cannot infer the distribution it came from. Nonetheless, a recent line of work has shown that (robust) steganography is still possible in this more demanding setting. Zamir \cite{zamir2024excuse} demonstrated that, even without knowledge of the content distribution, messages can be hidden in unaltered content. Follow-up works \cite{christ2024pseudorandom,fairoze2023publicly} extend this result to the noisy setting, constructing schemes in which the message can be decoded even in the presence of adversarial errors.

\paragraph{Pseudorandom error-correcting codes.} Pseudorandom codes (PRC), introduced by Christ and Gunn~\cite{christ2024pseudorandom}, are error-correcting codes whose codewords are computationally indistinguishable from uniformly random strings to any observer lacking the decoding key. As objects that satisfy both pseudorandomness and error correctability, PRCs provide a natural foundation for constructing watermarking schemes that achieve both undetectability and robustness~\cite{christ2024pseudorandom,gunn2024undetectable,cohen2024watermarking,golowich2024edit}.

PRCs with varying robustness properties have been constructed under standard cryptographic assumptions, such as the subexponential hardness of learning parity with noise (LPN)~\cite{christ2024pseudorandom,ghentiyala2024new}. In particular, Golowich and Moitra~\cite{golowich2024edit} construct PRCs that are robust to a constant fraction of edits, including substitutions, insertions, and deletions, at the cost of requiring a larger alphabet. This yields watermarking schemes that are robust under edit distance.

However, the notion of robustness satisfied by previous PRC constructions is \emph{weak} in a different sense. The decoder's robustness guarantee holds only against corruptions that are generated \emph{independently} of the PRC keys. In particular, this requirement assumes that the corruption adversary has not observed any prior codewords. This lack of \emph{adaptive} robustness was highlighted by Cohen et al.~\cite{cohen2024watermarking}, who ask whether PRC messages remain robustly decodable even after the perturbation adversary (or error channel) has observed multiple codewords. Ideal PRCs, introduced by Alrabiah et al.~\cite{alrabiah2024ideal}, address this limitation by constructing PRCs that remain robust even against computationally bounded adversaries with black-box access to both encoding and decoding oracles. In this work, we use them to instantiate secret key robust block steganography schemes.  

\paragraph{Property-preserving hash functions.} We build our watermarking scheme on top of property-preserving hash functions (PPH), introduced by Boyle et al.~\cite{boyle2019adversarially}. Several PPH constructions have been proposed for the exact Hamming predicate under standard cryptographic assumptions~\cite{fleischhacker2022property,holmgren2022nearly,zhao2025privacy}. Particularly relevant to our work is the elegant construction by Holmgren et al.~\cite{holmgren2022nearly}, which achieves nearly optimal hash size. Their key idea of combining syndrome list-decoding of linear error-correcting codes with collision-resistant hash functions has prior roots in the study of error-correcting codes against computationally bounded noise~\cite{micali2010optimal,dodis2008fuzzy}.

Recent work by Bogdanov et al.~\cite{bogdanov2025adaptive} proposed and examined a new computational assumption that implies a PPH for the Gap-$\ell_2$ predicate over Euclidean space. This enables evaluation of whether the $\ell_2$ distance between a pair of real vectors is small or large based solely on their compressed digests. As such, these PPHs serve as \emph{adversarially} robust locality-sensitive hash functions for the Euclidean distance: compressive hash functions for which it is computationally hard to find input pairs whose true pairwise distances differ substantially from those inferred from their hashes.

PPHs are also related to the cryptographic notion of \emph{secure sketches}~\cite{dodis2008fuzzy}, which are practically relevant in biometric applications where the secret key (e.g., one’s biometric data) is inherently noisy and varies slightly with each use. Secure sketches enable recovery of a value $x$ given a \emph{sketch} $g(x)$ and a nearby input $x'$. This can be viewed as a relaxation of PPH, as it assumes access to one noisy input in the clear.

%% file: sections/technical-overview.tex
\textbf{}\section{Technical Overview}
We highlight the techniques we use to construct robust signatures (\refSc{sec:rds-overview}) and robust, unforgeable, and recoverable watermarks (\refSc{sec:wm-overview}).
\subsection{Robust and recoverable signatures from PPH} 
\label{sec:rds-overview}

Our construction produces a $\Phi$-recoverable signature scheme $(\Gen, \Sign, \Recover)$ from a $\Phi$-difference-recovering PPH, $\PPH[\Sample, \Eval, \Recover]$, and a strongly unforgeable digital signature scheme, $\DSS[\Gen, \Sign, \Ver]$~\cite{bellare2007two}. An analogous construction yields a robust (but not recoverable) digital signature from a standard PPH without difference recovery. 

To sign a message $y \in \{0,1\}^n$, we compute its property-preserving hash $h(y)$, then generate a signature of that hash $\tau \leftarrow \DSS.\Sign(\DSS.\sk, h(y))$. The robust signature is defined as $\sigma = (h(y), \tau)$. The public and secret keys are given by $\pk = (h, \DSS.\pk)$ and $\sk = \DSS.\sk$.

To verify this signature $\sigma = (h(y), \tau)$ against a candidate message $\zeta \in \{0,1\}^n$, we compute
\begin{align*}
    \Ver(\pk, \zeta, \sigma) &= \DSS.\Ver(\DSS.\pk, h(y), \tau) \wedge \PPH.\Eval(h,h(y),h(\zeta))\;.
\end{align*}
We formally describe our construction in Figure~\ref{alg:rds-construct}. We give the intuition behind our security proofs, which can be found in \refSc{sec:rds-construction}.
\paragraph{Robustness.} Let $\sigma = (z, \tau)$ be the robust signature of some message $y \in \{0, 1\}^n$. If an adversarially chosen message $\zeta$ is $\Phi$-close to $y$, then with overwhelming probability $\PPH.\Eval(h,z,h(\zeta)) = 1$, since $z = h(y)$, from PPH correctness, and $\DSS.\Ver(z, \tau) = 1$ from $\DSS$ correctness.
\paragraph{Unforgeability.} Suppose an adversary $\mathcal{A}$ forges a pair $(\zeta, \sigma')$ with inverse-polynomial probability, meaning $\Ver(\pk, \zeta, \sigma') = 1$, but $\zeta$ is $\Phi$-far from all messages $y \in \{0, 1\}^n$ associated with the signature $\sigma'$. We consider two cases and show that each leads to a contradiction.
\begin{enumerate}
    \item The forged $\sigma'$ is not one of the signatures observed by $\mathcal{A}$ through its queries to the robust signature oracle $\Sign_{\sk}$. Notice that responses from $\Sign_{\sk}$ are precisely the message-signature pairs $(z, \DSS.\Sign(\DSS.\sk, z))$ that $\cA$ observes from the underlying standard strongly unforgeable digital signature $\DSS$. If $\mathcal{A}$ outputs $\sigma' = (z', \tau')$ that verifies despite not being one of these pairs, then $(z', \sigma')$ breaks the strong unforgeability of $\DSS$.
    \item The forged signature $\sigma' = (h(y), \tau)$ was observed by $\cA$, meaning it is the signature of some message $y$ that is $\Phi$-far from the candidate message $\zeta$. In this case, $\PPH.\Eval(h, h(y), h(\zeta))$ detects, with overwhelming probability, that $y$ and $\zeta$ are $\Phi$-far. If it did not, the pair $(y, \zeta)$ would violate the security of the PPH. Consequently, the verifier rejects, $\Ver(\pk, \zeta, \sigma') = 0$, except with negligible probability, leading to a contradiction.
\end{enumerate}
\paragraph{Recoverability.}
For recoverability, we first (robustly) check $\Ver(\pk, \zeta, \sigma')$ to ensure $\sigma'$ is a signature of some $y$ close to $\zeta$. If the signature verifies, then output
\begin{align*}
    \Recover(\pk, \zeta, \sigma') &= \PPH.\Recover(h, h(y), \zeta) \xor \zeta\;.
\end{align*}
Otherwise, $\Recover(\pk, \zeta, \sigma') = \bot$. Notice that when an adversarially chosen $\zeta$ is $\Phi$-close to a previously signed message $y \in \{0, 1\}^n$, difference recovery of $\PPH$ ensures that $\PPH.\Recover(h, h(y),  \zeta)$ outputs the bitwise difference $y \oplus \zeta$ with overwhelming probability, so $\Recover(\pk, \zeta, \sigma')$ correctly outputs $y$.
\subsection{Watermarking using robust signatures and block steganography}
\label{sec:wm-overview}
\paragraph{Language models and blocks.} Following prior work \cite{gunn2024undetectable,christ2024pseudorandom,fairoze2023publicly}, we model a language model $Q : \{0, 1\}^* \rightarrow [0, 1]$ as a mapping that takes in a token sequence and outputs the probability that the next token is $1$\footnote{Without loss of generality, we consider a binary token alphabet. This assumption is standard in the watermarking literature; see \cite[Section~7.1]{christ2024pseudorandom} for a detailed discussion.}. Given a prompt $x \in \{0,1\}^*$, we obtain an $n$-bit \emph{response} $u = (u[1],u[2],\ldots, u[n]) \leftarrow \Response_n(x)$ by autoregressive sampling initialized with $x$. 

Throughout this section, we refer to the $n$-bit \emph{blocks} $(y_1, \ldots, y_{\floor{|y|/n}})$ of a string $y \in \{0,1\}^*$. Informally, $y_t$ denotes the $t$-th non-overlapping, $n$-bit substring of $y$.

\paragraph{A first attempt at unforgeable watermarking.} 
A natural approach to constructing unforgeable watermarks is to sign text as it is generated using a standard digital signature scheme and then embed that signature into subsequently generated text via a block steganography scheme. Later, the signature can be decoded and verified along with the text. This is the approach taken by existing constructions that consider unforgeability \cite{christ2024pseudorandom,fairoze2023publicly}, but it comes at the cost of robustness.

Specifically, consider a language model $Q$. Prior work \cite{christ2024pseudorandom,alrabiah2024ideal,fairoze2023publicly} demonstrates that block steganography schemes can be constructed for any model with sufficiently high entropy. For the purposes of this overview, we therefore assume that $Q$ has high entropy. We will attempt to watermark $Q$ using the approach outlined above.

Let $\lambda$ be the security parameter and $n = \poly(\lambda)$ be a block size. Let $\BWM[\Gen, \Embed, \Dec]$ be a $\Phi$-robust block steganography scheme for $Q$ with block size $n$. $\BWM$ can be either public or secret key. For the purposes of this overview, we will assume it's public key. Let $\DSS[\Gen, \Sign, \Ver]$ be a \emph{standard} digital signature scheme with message size $n$.

To generate verification and watermarking keys $(\vk, \sk)$, sample 
\[\DSS.\pk,\DSS.\sk) \gets\DSS.\Gen(1^\lambda), \quad  (\BWM.\dk, \BWM.\sk) \gets \BWM.\Gen(1^\lambda)\]
and let $\vk =$ $(\DSS.\pk,$ $\BWM.\dk)$, $\sk = (\DSS.\sk, \BWM.\sk)$. We can generate $2n$ bits of watermarked text as follows. On input prompt $x \in \{0,1\}^*$, the watermarking algorithm $\Wat(\sk, x)$ first samples $n$ bits normally from the underlying language model, $y_1 \gets \Response_n(x)$. Next, it signs this block and embeds the signature into the following block\footnote{Because robust block steganography schemes must be expanding in the message, the signature length must be smaller than the resulting block size, i.e., $|\sigma| \ll n$.}:
\[
\sigma \gets \DSS.\Sign(\DSS.\sk, y_1), \quad y_2 \gets \BWM.\Embed(\BWM.\sk, \pi = (x, y_1), \sigma)\;,
\]
and outputs the concatenation $(y_1, y_2)$. Undetectability of $\Wat$ follows from the undetectability of $\BWM$.

To verify a candidate text $\zeta \in \{0,1\}^{2n}$, the verification algorithm $\Ver(\vk, \zeta)$ first parses it into two blocks $\zeta_1, \zeta_2 \in \{0, 1\}^n$. It then tries to decode a signature from the second block and check if it verifies with the first,
\[\Ver(\vk, \zeta) = \DSS.\Ver(\DSS.\pk, \zeta_1, \sigma'), \text{ where } \sigma' \leftarrow \BWM.\Dec(\BWM.\dk, \zeta_2)\;.\]
Notice that this approach is unforgeable in that it only verifies $\zeta = (\zeta_1, \zeta_2) \leftarrow \cA^{\Wat_{\sk}}(\vk)$ when the first block $\zeta_1$ \emph{exactly} matches one of the blocks $y_1$ for which $\cA$ saw a signature for; otherwise, $\cA$ could break the unforgeability of $\DSS$ with the message $\zeta_1$. However, the scheme provides no robustness. Concretely, suppose an adversary $\cA^{\Wat_{\sk}}(\vk)$ queries its watermarking oracle and receives a response $(y_1, y_2)$. If $\cA$ modifies even a single bit of $y_1$ to produce $\zeta_1$, the string $(\zeta_1, y_2)$ can no longer verify, despite being nearly identical to an observed response. 
\paragraph{Robust signatures for both robustness and unforgeability.} To overcome the limitation outlined above, we instead embed a \emph{robust} signature into the generated text. Now, even if $\cA$ modifies $y_1$ to produce $\zeta_1$, so long as $\Phi(\zeta_1, y_1) = 1$, the signature---and therefore the watermark---can still verify.

Formally, let $\RDS[\Gen, \Sign, \Ver]$ be a \emph{robust} signature scheme for the same predicate $\Phi$ that $\BWM$ is robust to, and generate keys $(\vk, \sk)$ as in the standard digital signature approach, but using $\RDS.\Gen(1^\lambda)$ instead of $\DSS.\Gen(1^\lambda)$. We generalize the approach above to generate $rn$ bits for any $r \ge 2$. For the first block, $t = 1$, the algorithm $\Wat(\sk, x)$ invokes $y_1 \leftarrow \Response_n(x)$. For subsequent blocks $t \ge 2$, it updates the context to be the prompt and the text generated up to this point, $\pi = (x, y_1, \ldots, y_{t-1})$, and embeds the robust signature of the previous block:
\[
\sigma_{t-1} \gets \RDS.\Sign(\RDS.\sk, y_{t-1}), \quad y_t \gets \BWM.\Embed(\BWM.\sk, \pi, \sigma_{t-1})\;,
\]
The final output is the concatenation of blocks $(y_1, \ldots, y_r)$.
\paragraph{Watermark verification.}
We define watermark verification to operate on $rn$-bit strings for some $r \ge 2$. Given a candidate string $\zeta^* = (\zeta_1^*, \ldots, \zeta_r^*)$ composed of $r$ blocks of $n$-bits, the verification algorithm $\Ver_r(\vk, \zeta^*)$ checks whether the \emph{entire} chain of message-signature pairs verifies. Concretely, this means verifying that
\[\forall i \in [r-1], \;\RDS.\Ver(\RDS.\pk, \zeta_i^*, \sigma_{i}') = 1, \text{ where } \sigma_i' \leftarrow \BWM.\Dec(\BWM.\dk, \zeta_{i+1}^*)\;.\]
If $\RDS$ is additionally recoverable, we can define a corresponding recovery procedure $\Rec_r(\vk, \zeta^*)$ as follows.
\[\forall i \in [r-1],\;\xi_i\leftarrow \;\RDS.\Rec(\RDS.\pk, \zeta_i^*, \sigma_{i}'), \text{ where } \sigma_i' \leftarrow \BWM.\Dec(\BWM.\dk, \zeta_{i+1}^*)\;.\]
The procedure outputs a single string $(\xi_1, \ldots, \xi_{r-1})$ provided that none of the recovery attempts return $\bot$, else it outputs $\bot$.

\paragraph{Guarantees of our construction when $r=2$.}
Consider the simple case $r = 2$. To see that verification is robust, let $\zeta^* = (\zeta_1^*, \zeta_2^*)$ be a candidate text produced by an efficient adversary $\cA^{\Wat_{\sk}}(\vk)$. Say that $\zeta^*$ is approximately contained in some response $y \leftarrow \Wat_{\sk}(x)$ observed by $\cA$, meaning that $y$ contains two consecutive $n$-bit blocks $(y_t, y_{t+1})$ such that both $\Phi(\zeta_1^*, y_t) = 1$ and $\Phi(\zeta_2^*, y_{t+1}) = 1$. By the $\Phi$-robustness of $\BWM$, we can decode the signature $\sigma_t$ of $y_t$ from $\zeta_2^*$. Moreover, by the $\Phi$-robustness of $\RDS$, the pair $(\zeta_1^*, \sigma_t)$ will pass verification. Similarly, in this setting, the recovery algorithm outputs the first response block $y_t$.

Moreover, to see the watermark is unforgeable, notice that unforgeability of the digital signature says that if $\Ver_2(\vk, (\zeta_1^*, \zeta_2^*)) = 1 \text{ (or } \Rec_2(\vk, (\zeta_1^*, \zeta_2^*))= y_t)$ 
then $\zeta_1^*$ must be $\Phi$-close to one of the signed blocks $y_t$ observed by $\cA$. We formalize our robustness and unforgeability guarantees in Theorem \ref{thm:watermark-main}, with recoverability stated in Corollary \ref{cor:watermark-recoverable}.

Our construction does, however, exhibit a subtle gap. We can guarantee that verification succeeds when \textit{both blocks} $(\zeta_1^*, \zeta_2^*)$ are close to back-to-back response blocks $(y_t, y_{t+1})$, whereas we can only ensure that successful verification implies the \textit{first block} $\zeta_1^*$ is close to a single response block $y_t$. This asymmetry arises because we embed signatures directly into subsequent blocks: to decode the signature and ensure verification succeeds, both blocks must be close. Successful RDS verification, on the other hand, only certifies the closeness of the underlying message (the first block). We are optimistic that future strategies for embedding messages into AI-generated content could help narrow this gap.
\paragraph{Guarantees when $r > 2$.} In Theorem \ref{thm:r-block-watermark-full}, we show that $\Ver_r$ satisfies a stronger unforgeability guarantee when $r > 2$. Specifically, we show that if $\Ver_r\big(\vk, (\zeta_1^*, \ldots, \zeta_r^*)\big) = 1,$ then not only is each block $\zeta_i^*$ for $i \in [r-1]$ $\Phi$-close to \emph{some} response block, but moreover the blocks $(\zeta_1^*, \zeta_2^*, \ldots, \zeta_{r-1}^*)$ are each $\Phi$-close to the aligned blocks of a \emph{consecutive chain} $(y_t,y_{t+1}, \ldots, y_{t+r-2})$ present in a \emph{single response} $y$. 

This notion of closeness is captured by the every-block-close lifting of $\Phi$ (Definition~\ref{def:every-block-close}). In particular, we obtain an unforgeability guarantee of the following form: $(\zeta_1^*, \ldots, \zeta_r^*)$ verifies only if for there is some watermarked response $y$ such that $\EBC{\Phi}{n}\big((\zeta_1^*, \ldots, \zeta_{r-1}^*), y\big)=1$. We obtain a matching robustness guarantee: verification succeeds whenever there exists a response $y$ such that $\EBC{\Phi}{n}\big((\zeta_1^*, \ldots, \zeta_r^*), y\big)=1$.
\paragraph{Generalizing verification to arbitrary-sized strings.} In the discussion above, the algorithms $\Ver_r$ and $\Rec_r$ operate on inputs of fixed length $\zeta^* \in \{0, 1\}^{rn}$. Semantically, however, it is more natural for watermark verification to operate on inputs of arbitrary length. In particular, given $\zeta \in \{0,1\}^*$, verification should succeed if and only if some sufficiently long \emph{substring} $\zeta^*$ of $\zeta$ is approximately contained in a watermarked response. For example, if a student copies only a portion of an essay from a watermarked response, the verification algorithm should still verify the watermark when given the full essay.

To capture this more semantically meaningful notion, our formal security definitions (Definitions \ref{def:robustness-general} and \ref{def:unforgeability-general}) differ slightly from the informal ones given in this overview. As a concrete example, for a closeness predicate $\Phi$ and a substring-length $\ell : \mathbb{N} \to \mathbb{N}$, we define public-key $(\Phi,\ell)$-unforgeability as follows:

\[
\Big(\Ver(\vk, \zeta) = 1\Big) 
\;\; \implies \;\;
\Biggl(
\begin{array}{l}
\exists \zeta^* \preceq \zeta \text{ of size } \ell(\lambda), \\
\exists (x, y) \in \Trans(\cA^{\Wat_{\sk}}) \text{ such that } \Phi(\zeta^*, y) = 1
\end{array}
\Biggr)\;,
\]
 where $\zeta^* \preceq \zeta$ denotes that $\zeta^*$ is a substring of $\zeta$. 

Defining robustness and unforgeability in this way also has the technical advantage of addressing the gap noted in Remark~\ref{rem:gap}, as it allows each security property to be stated over substrings of different lengths. In particular, the construction outlined above can be easily adapted to achieve $(\EBC{\Phi}{n}, rn)$-robustness and $(\EBC{\Phi}{n}, (r-1)n)$-unforgeability. This can be done by defining a verifier $\Ver(\vk, \zeta)$ that takes as input an arbitrary-length string $\zeta \in \{0,1\}^*$ and performs $r$-block verification $\Ver_r(\vk, \zeta^*)$ on \emph{every} substring $\zeta^* \preceq \zeta$ of length $rn$, accepting if \emph{any} such substring verifies. We formalize this algorithm in Figure \ref{alg:chain-ver}.

\subsection{Discussion}
\label{sec:discussion}
We view our work as an initial step toward a broader line of investigation. Our framework treats robustness and undetectability as the traditional goals of watermarking, while focusing on introducing and achieving the new properties of unforgeability and recoverability without sacrificing the former. We demonstrate its effectiveness by instantiating language watermarking schemes for the every-block-close lifting of the Hamming predicate.

The resulting notion of robust signatures is versatile. Our general construction paradigm suggests avenues for improvement and indicates that our ideas can be extended to other modalities of generative AI. We address some concerns, as well as potential applications of our framework to other settings.

\paragraph{Why Ideal PRCs do not offer unforgeability?}
Ideal pseudorandom codes (PRCs) \cite{alrabiah2024ideal} do provide sharp decoding guarantees: any (adversarially chosen) $\zeta$ that is sufficiently far from all codewords decodes to $\perp$. At first glance, this suggests that ideal PRCs alone might yield secret key unforgeable watermarks. Indeed, a PRC can be seen as a \emph{secret key} robust and unforgeable watermarking scheme for the \emph{uniform} language model, where each bit is equally likely, as text generated by the model is indistinguishable from PRC codewords.

For general language models, however, embedding PRC codewords into text introduces errors that break these guarantees. When the embedding error is large, a string $\zeta$ that is close to some watermarked response $y$ may be far from the PRC codeword embedded in $y$, and a string $\zeta$ that is far from all valid responses $y$ may still be close to a codeword. Therefore, even if $\zeta$ is far from all responses $y$, if the embedding errors are large it can still verify as watermarked. Typically, this error is bounded by making an assumption on the entropy of the language models: a lower bound on entropy yields an upper bound on the errors incurred when embedding PRC codewords into text.

In contrast, our unforgeability guarantees are independent of the errors introduced by embedding, enabling protection against false positives without imposing entropy assumptions. Moreover, our framework is compatible with \emph{any} block steganography scheme, regardless of its guarantees against false positives, which allows us to achieve a watermarking scheme that is both public key robust and public key unforgeable. Finally, our approach supports recoverability without requiring access to the underlying language model or the original generation prompt.

\paragraph{About our Hamming guarantees.} Our final guarantee is a watermarking scheme for the every-block-close lifting of the Hamming predicate. We acknowledge that Hamming distance between binary token strings may not capture a wide class of watermark removal attacks. However, our general construction paradigm means any new improvement to block steganography and PPHs, e.g. improved error rates or support for different types of errors, will translate to improved unforgeable and recoverable watermarking. For example, we believe that robust signatures for a predicate we call \textit{GapEdit}, which outputs 1 when the edit distance between two strings is small, 0 when it is very large, and is undefined for intermediate distances, are within reach. We expect such signatures to be attainable by instantiating our construction with a secure sketch for edit-distance \cite{dodis2008fuzzy}.

\paragraph{Post-hoc vs.\ ad-hoc steganography.} 
We leverage ad-hoc steganography, where the message is embedded into the text during generation, to encode robust signatures directly into model outputs. In contrast, post-hoc steganography schemes, such as Google's SynthID image watermark \cite{gowal2025synthid}, embed the message after generation. That is, the embedder first generates $y \gets \Response(x)$ normally and then post-processes it to obtain $y' \approx y$, where a given message is embedded into $y'$. Post-hoc steganography schemes cannot guarantee undetectability, but often maintain, at least empirically, that the modified content $y'$ is ``perceptually'' similar to the original response $y$.

Robust signatures are a natural choice for the message to be embedded using a post-hoc steganography scheme. One can generate $y \gets \Response(x)$, compute its robust signature $\sigma \gets \Sign(\sk, y)$, and embed $\sigma$ into $y' \approx y$ using a post-hoc watermarking scheme. Later, the signature can be decoded from $\zeta \approx y'$ and verified via robust signature verification $\Ver(\pk, \zeta, \sigma)$. Note, however, that some robustness may be lost due to the modifications introduced during post-hoc steganography, since $\sigma$ is the robust signature of the original response $y$, and the message $\zeta$ is close to the post-processed output $y' \approx y$. In this paper, we instead focus on ad-hoc steganography due to its stronger theoretical guarantees. 

\paragraph{Robust and unforgeable image watermarking. } 
Fairoze et al. \cite{fairoze2025difficulty} argue that constructing post-hoc watermarking schemes for \emph{image models} that are both robust and unforgeable may be difficult. While our work focuses on language models, our results in fact suggest the opposite conclusion under the right abstraction.

In their work, they assume the existence of an embedding function $g : \mathbb{R}^n \rightarrow \mathbb{R}^m$ with $m < n$ that compresses images and preserves visual similarity: namely, that two images $y, y' \in \mathbb{R}^n$ are visually similar if and only if $g(y)$ and $g(y')$ are close under a chosen distance metric (they consider cosine distance). They show that by combining such an embedding with standard digital signatures and a post-hoc steganography scheme for images, they can construct a robust and unforgeable post-hoc watermarking scheme for image models. At a high level, their unforgeability guarantee is as follows: if a watermark is detected in an adversarially-chosen image $\zeta$, then the embedding $g(\zeta)$ must be close to the embedding of some watermarked image $g(y)$ observed by the adversary. 

However, they subsequently show that the assumed embedding function~$g$ may be overly idealized. In particular, they evaluate several candidate embedding functions (e.g., the forward pass of a trained deep network) and observe that these embeddings fail to satisfy the desired property. Specifically, it is easy to efficiently find images $y$ and $y'$ that are visually dissimilar, yet whose embeddings $g(y)$ and $g(y')$ are close in cosine distance. As a consequence, the unforgeability guarantee becomes vacuous: even if the underlying watermarking scheme is unforgeable, an adversary can produce an image that is visually far from all watermarked images, yet still verifies as watermarked because its embedding is close to the embedding of a previously observed image.

In this work, we overcome this issue for \emph{language models} by leveraging PPHs. Rather than compressing text via a heuristic, similarity-preserving embedding~$g$, we hash it using a PPH, which enables evaluating the similarity of the original texts directly (even if the hashes themselves may not be close under any distance metric). The adaptive security of PPH rules out attacks of the form outlined above. 

In principle, this approach also extends to post-hoc \emph{image} watermarking by using a PPH instead of the embedding function~$g$ in the construction of~\cite[Section~5.3]{fairoze2025difficulty}. However, Hamming PPHs are not especially compelling in this setting: one could hash the binary representations of pixels of the images, but this would effectively deem two images visually similar if and only if the Hamming distance between the binary representation of their pixels is small. PPHs for $\ell_2$-distance do exist under the recently proposed contracting hypergrid vector assumption~\cite{bogdanov2025adaptive}, when the input space consists of bounded integer vectors. If this bound is set to $256$, then one could hash raw (unnormalized) pixel values so that two images are deemed similar iff their pixels are close in $\ell_2$-distance. Likewise, one could hash the latent vector used to generate an image during diffusion, which is often interpreted as capturing the image’s semantic meaning. Unfortunately, such PPH constructions may not be practical, particularly when the coordinate bound is as large as~$256$, highlighting the need for further improvements.

%% file: sections/watermarking-primitive.tex
\section{Watermarking definitions}
\label{sec:notation}
\paragraph{Notation.}  
For strings $x, y \in \{0, 1\}^*$, we write $x \preceq y$ to denote that $x$ is a \emph{substring} of $y$.  
We write $(x, y)$ for their concatenation and $|y|$ for the length of $y$. We write $y[i]$ for the $i$th bit of $y$, and $y[i:j]$ for $i < j$ to denote the substring $(y[i], y[i+1], \ldots, y[j-1])$ (excluding the $j$th bit). \\\\
For $q \in \NN$, we denote $[q] = \{1, 2, \ldots, q\}$. For a set $X$, we write $x \leftarrow X$ to denote that $x$ is a random variable sampled uniformly from $X$.  For a (probabilistic) algorithm $\cA$ and a (possibly random) input $Y$, we write $x \leftarrow \cA(Y)$ to denote that $x$ is a random variable distributed according to the output of $\cA$ when run on input $Y$.\\\\
Let $\lambda \in \NN$ be a security parameter. We write $n = \poly(\lambda)$ if $n : \NN \to \NN$ is a polynomial function. When $\lambda$ is clear from context, we simply write $n$ for $n(\lambda)$. For a function $f : \NN \to \RR$, we write $f(\lambda) \leq \negl(\lambda)$ to indicate that $f$ is bounded above by some negligible function.\\\\
For $p \in [0,1]$, we write $\Ber(p)$ for the Bernoulli distribution on $\{0,1\}$ with expectation $p$. 
\paragraph{Language models.} A language model $Q : \{0, 1\}^* \rightarrow [0, 1]$ maps a token sequence to the probability that the next token is $1$. Given a prompt $x \in \{0,1\}^*$, we obtain a \emph{response} $u = (u[1],u[2],\ldots) \leftarrow \Response(x)$ by autoregressive sampling initialized with $x$. Formally, each bit $u[j]$ of the response is sampled as
\begin{align*}
    u[j] \leftarrow \Ber\big(Q(x, u[1:j])\big)\;.
\end{align*}
$(x,u)$ is a single prompt-response pair. For $\ell \in \NN$ we write $\Response_\ell(x)$ for the distribution obtained from autoregressively sampling the next $\ell$ bits.
\paragraph{Transcripts.} We consider adversaries that interact with oracles $\M$ by submitting inputs $x \in \{0, 1\}^*$ and receiving responses $y \in \{0, 1\}^*$. For oracles that accept multiple inputs, $x$ may be viewed as a concatenation of these inputs. For a (randomized) adversary $\mathcal{A}$, we denote by  $\Trans(\mathcal{A}^\M)$ the \emph{transcript} of this interaction, i.e., the random variable describing the set of all input-output pairs observed by $\mathcal{A}$.
\paragraph{Blocks.}
For a string $y \in \{0,1\}^*$ and a block size $n \in \NN$, we define the $n$-bit block parsing of $y$ as
\[
  \Blocks(y) \;=\; (y_1, y_2, \ldots, y_{\lfloor |y|/n \rfloor})\,,
\]
where each $y_t \in \{0,1\}^n$ is the substring $y[(t-1)n + 1:tn+1]$. Any trailing bits that do not form a full block are discarded. For convenience, we sometimes write ``an $n$-sized block $y_t$ of $y$'', or just $y_t$ when $n$ is clear from context.  In such cases, we mean: let $y_t$ denote the $t$-th entry in the sequence given by
$\Blocks(y)$.

\paragraph{Predicates.} We use predicates to formalize ``closeness'' between strings. A \emph{predicate} is a Boolean-valued function $\Phi : \cX \to \{0,1\}$ defined on some domain $\cX$. Throughout this paper, when we quantify over predicates, we implicitly restrict to predicates that are efficiently computable. A canonical example is the Hamming predicate, which checks whether two strings differ in at most a $\delta$-fraction of positions. 
\begin{definition}[Hamming]
\label{def:hamming-predicate}
    For any $\delta \in [0, 1]$, we define the Hamming predicate $\Ham_\delta$ as follows. For any strings $y,y' \in \{0,1\}^*$,
    \begin{align*}
        \Ham_\delta(y,y') = (|y| = |y'| \; \wedge \;\|y-y'\|_0 \le \delta |y|)\;.
    \end{align*}
\end{definition}
For a predicate $\Phi$ with input domain $\cX = \{0, 1\}^* \times \{0, 1\}^*$ and strings $y, y' \in \{0, 1\}^*$, we say that $y$ is {$\Phi$-close} to $y'$ if $\Phi(y, y') = 1$, and {$\Phi$-far} from $y'$ if $\Phi(y, y') = 0$.

\paragraph{Every block close predicate.}
To state the guarantees of our watermarking construction, it will be useful to have in mind a notion of \emph{every-block-closeness} for a predicate $\Phi$.  
Intuitively, this notion \emph{lifts} $\Phi$, which originally may only compare two equal-sized strings of size $n$, to a predicate on sequences of blocks: it checks whether the first input is composed of consecutive $n$-bit blocks that are $\Phi$-close to a sequence of consecutive $n$-bit blocks contained in the second input. 
\begin{definition}[Every-block close]
\label{def:every-block-close}
    For a predicate $\Phi$ and block size $n \in \NN$, we define the predicate $\EBC{\Phi}{n}$ (every-block-close) as
    \[
\EBC{\Phi}{n}(\zeta^*, y) :\;
\left(
\begin{aligned}
&\text{Check that } \exists r \in \NN \text{ such that } |\zeta^*| = rn \text{. Output 0 if not.}\\
&\text{Parse }  (\zeta^*_1, \ldots, \zeta_r^*) =\Blocks(\zeta^*),\\
&(y_1, \ldots, y_{\floor{|y|/n}}) = \Blocks(y)\text{ into }  n\text{-bit blocks.}\\
&\text{Output 1 iff } \exists t \in \NN \text{ such that } \forall j \in [r], \;\Phi(\zeta^*_j, y_{t+j-1}) = 1.
\end{aligned}
\right)
\]
\end{definition}
 Recall Figure \ref{fig:ebcfig}, which illustrates $\EBC{\Phi}{n}$. The predicate $\EBC{\Phi}{n}$ is generally {asymmetric}: the first input $\zeta^*$ is treated as a sequence of $r$ consecutive $n$-bit blocks, while the second input $y$ may be longer. The predicate searches for a sequence of $r$ consecutive blocks \emph{within} $y$ that are blockwise $\Phi$-close to $\zeta^*$.  When $\zeta^*$ and $y$ have the same length, the predicate becomes symmetric (assuming $\Phi$ is symmetric) and reduces to a simple blockwise comparison, checking that their aligned blocks are $\Phi$-close.

\subsection{Watermarking}
\label{sec:watermarking-primitive}
We start by introducing the syntax of a \textit{watermarking scheme} for a language model $Q$, a tuple of algorithms that act as the interface for generating and verifying watermarked outputs. We require that watermarking schemes satisfy a strong notion of quality preservation: undetectability \cite{gunn2024undetectable}. This guarantees that watermarked outputs from $\Wat$ are indistinguishable from standard autoregressive samples from $Q$ to any efficient adversary that can choose input prompts adaptively.
\begin{definition}[Watermarking]
\label{def:watermarking-syntax}
    Let $\lambda \in \mathbb{N}$ be a security parameter. A watermarking scheme for a language model $Q$ and token set $\{0, 1\}$ is a tuple of polynomial time algorithms $(\Gen, \Wat, \Ver)$ satisfying,
\begin{itemize}
    \item $\Gen(1^\lambda):$ takes in a security parameter $\lambda$ and outputs verification and secret keys $(\vk, \sk)$.
    \item $\Wat(\sk, x):$ takes in the secret key $\sk$ and a prompt $x \in \{0, 1\}^*$ and outputs a response $y \in \{0, 1\}^*$.
    \item $\Ver(\vk, \zeta):$ takes in the verification key $\vk$ and input text $\zeta \in \{0, 1\}^*$  and outputs a verification bit $b \in \{0, 1\}$.
\end{itemize}
A recoverable watermarking scheme additional requires the following PPT algorithm $\Rec$,
\begin{itemize}
    \item $\Rec(\vk,\zeta):$ takes in the verification key $\vk$ and input text $\zeta \in \{0, 1\}^*$ and outputs a list $\mathcal{L} \subseteq \{0,1\}^*$ of recoveries or a failure symbol $\bot$.
\end{itemize}
\textbf{Undetectability} For every PPT distinguisher $\mathcal{A}$,
\[\left|\arraycolsep=1.4pt
  \Pr_{(\vk, \sk) \leftarrow \Gen(1^\lambda)}\left[
      \begin{array}{rl}
\mathcal{A}^{Q, \Wat_{\sk}}(1^\lambda) = 1
     \end{array}
  \right] - 
  \arraycolsep=1.4pt
  \Pr\left[
      \begin{array}{rl}
\mathcal{A}^{Q, \Response}(1^\lambda) = 1      \end{array}
  \right]\right|\leq \negl(\lambda)\;.\]
\end{definition}

We now define our security notions for watermarking schemes.  
Throughout this section, we present the \emph{secret key} variants, where the adversary has oracle access to a verification oracle $\Ver_{\vk}$.  
In Remark \ref{remark:public-key}, we describe the corresponding public key variants, obtained by simply giving the adversary explicit access to~$\vk$.

Robustness requires that it is hard for an efficient adversary---even with access to both the watermarking oracle and the verification oracle---to produce a string $\zeta$ that contains a $\ellrb$-sized substring $\zeta^* \preceq \zeta$ which is $\Phi$-\emph{close} to one of the watermarked responses $y \leftarrow \Wat(\sk, x)$ it observed, yet for which $\zeta$ fails to verify. 
\begin{definition}[Robustness]
\label{def:robustness-general}
For any closeness predicate $\Phi$ and any substring size $\ellrb = \poly(\lambda)$, a watermarking scheme $(\Gen, \Wat, \Ver)$ is secret key \emph{$(\Phi, \ellrb)$-robust} if for any PPT algorithm $\cA$ outputting $\zeta \in \{0, 1\}^*$,
\begin{align*}
    \Pr\left[
    \begin{array}{rl}
        (\vk, \sk) &\leftarrow \Gen(1^\lambda)\\
        \zeta &\leftarrow\cA^{\Wat_{\sk}, \Ver_{\vk}}(1^\lambda)
    \end{array} 
    :
    \begin{array}{rl}
    \Big(\exists \zeta^* \preceq \zeta, (x, y) \in \Pi \text{
    such that } 
    \\\Phi(\zeta^*, y) = 1
    \wedge \; |\zeta^*| = \ellrb \Big)\\
    \wedge \; \Ver(\vk, \zeta) = 0
    \end{array}
    \right] \le \negl(\lambda)\;,
\end{align*}
where $\Pi$ denotes the transcript $\Trans(\cA^{\Wat_{\sk}})$. 
\end{definition}

We include the substring requirement, $|\zeta^*| = \ellrb$, for two reasons. First, we cannot guarantee that the watermark is verifiable in arbitrarily small $\zeta$, so this requirement imposes a lower bound on the size of $\zeta$. Second, it allows us to capture adversaries who approximately copy a portion of their output, $\zeta^*$, from the watermarked model, then insert arbitrary text around it. The substring size specifies exactly how much the text must be copied for the watermark to verify, regardless of the surrounding content.

Unforgeability complements robustness by requiring that it is hard for an efficient adversary to produce a string $\zeta$ that verifies as watermarked, but for which $\zeta$ contains \emph{no} $\elluf$-sized substring $\zeta^* \preceq \zeta$ that is $\Phi$-close to any watermarked response $y$ observed by the adversary. Another interpretation of unforgeability is that if all $\elluf$-sized substrings $\zeta^* \preceq \zeta$ are \emph{far} from all of the watermarked responses $y$ seen by the adversary, then verification must fail.
\begin{definition}[Unforgeability]
\label{def:unforgeability-general}
For any closeness predicate $\Phi$ and any substring size $\elluf = \poly(\lambda)$, a watermarking scheme $(\Gen, \Wat, \Ver)$ is secret key \emph{$(\Phi, \elluf)$-unforgeable} if for any PPT algorithm $\cA$ outputting $\zeta \in \{0, 1\}^*$,
\begin{align*}
    \Pr\left[
    \begin{array}{rl}
        (\vk, \sk) &\leftarrow \Gen(1^\lambda)\\
        \zeta &\leftarrow\cA^{\Wat_{\sk}, \Ver_{\vk}}(1^\lambda)
    \end{array} 
    :
    \begin{array}{rl}
    \Ver(\vk, \zeta) = 1 \\
    \wedge \; \neg \Big(\exists \zeta^* \preceq \zeta, (x, y) \in \Pi \text{ such that } \\
    \Phi(\zeta^*, y) = 1\;
    \wedge \; |\zeta^*| = \elluf \Big)
    \end{array}
    \right] \le \negl(\lambda)\;,
\end{align*}
where $\Pi$ denotes the transcript $\Trans(\cA^{\Wat_{\sk}})$. 
\end{definition}

These definitions provide security for binary watermark verification. However, we may wish to go further and track the edits applied to content by recovering the original model outputs. \emph{Recoverable watermarking schemes}
provide exactly this capability. Robustness for a recoverable watermarking scheme requires that no efficient adversary can produce a string
$\zeta$ containing a sufficiently long substring $\zeta^* \preceq \zeta$ that is
$\Phi$-close to some watermarked response $y$, yet for which the recovery algorithm does not recover a substring of that response $\xi \preceq y$. This recovered substring $\xi$ must also be similar to the copied $\zeta^*$ in size (up to an allowed loss).
\begin{definition}[Robust recoverability]
\label{def:recoverability-general}
For any closeness predicate $\Phi$, substring size $\ellrec = \poly(\lambda)$, and loss size $\ellloss = \poly(\lambda)$, a \emph{recoverable} watermarking scheme $(\Gen, \Wat, \Rec)$ is secret key \emph{$(\Phi, \ellrec, \ellloss)$-robust} if for any PPT algorithm $\cA$ outputting $\zeta \in \{0, 1\}^*$,
\begin{align*}
    \Pr\left[
    \begin{array}{rl}
        (\vk, \sk) &\leftarrow \Gen(1^\lambda)\\
        \zeta &\leftarrow\cA^{\Wat_{\sk}, \Rec_{\vk}}(1^\lambda)\\
        \mathcal{L} &\leftarrow \Rec(\vk, \zeta)
    \end{array} 
    :
    \begin{array}{rl}
   \Big(\exists \zeta^* \preceq \zeta, (x, y) \in \Pi\text{ such that }\\
    \Phi(\zeta^*, y) = 1 \; \wedge \;  |\zeta^*| \ge \ellrec \Big) \; \wedge\\
    \mathcal{L} \text{ does not contain a} \\ (|\zeta^*| - \ellloss)\text{-sized substring of } y
    \end{array}
    \right] \le \negl(\lambda)\;,
\end{align*}
where $\Pi$ denotes the transcript $\Trans(\cA^{\Wat_\sk})$. 
\end{definition}

To ensure that recoveries are reliable, it is natural to additionally require that any recovered text truly originates from the model. This requirement is especially important in applications where recovered substrings are used for fine-grained attribution. One could rely on standard unforgeability, which guarantees that whenever the recovered list $\Rec(\vk, \zeta)$ contains a string $\xi$, there exists \emph{some} observed response $y$ that is close to a substring $\zeta^* \preceq \zeta$. However, this does not ensure that the recovery $\xi$ itself was genuinely observed by the adversary. Instead, we require \emph{unforgeable recoverability}: no efficient adversary can produce text $\zeta$ 
such that $\Rec(\vk, \zeta)$ recovers a string $\xi$ which is either too small, or is \emph{not} a substring of any response $y$ that was approximately copied by the adversary.
\begin{definition}[Unforgeable recoverability]
\label{def:unforgeable-recoverability-general}
For any closeness predicate $\Phi$ and any substring size $\ellrec = \poly(\lambda)$, a \emph{recoverable} watermarking scheme $(\Gen, \Wat, \Rec)$ is secret key \emph{$(\Phi, \ellrec)$-unforgeable} if for any PPT algorithm $\cA$ outputting $\zeta \in \{0, 1\}^*$,
\begin{align*}
    \Pr\left[
    \begin{array}{rl}
        (\vk, \sk) &\leftarrow \Gen(1^\lambda)\\
        \zeta &\leftarrow\cA^{\Wat_{\sk}, \Rec_{\vk}}(1^\lambda)
    \end{array} 
    :
    \begin{array}{rl}
    \exists \xi \in \Rec(\vk, \zeta) \text{ such that } \Big(|\xi| < \ellrec\big)\\
   \vee \; \neg \Big(\exists \zeta^* \preceq \zeta, (x, y) \in \Pi \text{ with} \\ \Phi(\zeta^*, y) = 1 \wedge \; |\zeta^*| = |\xi| \; \wedge \; \xi \preceq y\Big)\\
    \end{array}
    \right] \le \negl(\lambda)\;,
\end{align*}
where $\Pi$ denotes the transcript $\Trans(\cA^{\Wat_{\sk}})$.
\end{definition}
\begin{remark}
    \label{remark:public-key}
    We presented the secret key variants of all our definitions, where the attacker receives oracle access to both $\Wat_{\sk}$ and $\Ver_{\vk}$. We also consider the corresponding public key variants, which are identical except that the attacker additionally receives the full verification key $\vk$ (rather than only black-box access to the verification oracle).
\end{remark}

\begin{remark}
    We do not include a loss term $\ellloss$ in the definition of unforgeable recoverability, because our construction guarantees the existence of a substring $\zeta^* \preceq \zeta$ that has \emph{exactly} the same length as the recovered string.  
\end{remark}

\begin{remark}
We define standard robustness to consider fixed-size substrings $|\zeta^*| = \ellrb$, but allow variable-size substrings $|\zeta^*| \geq \ellrec$ for
recoverable robustness. One could additionally define a dynamic version of standard robustness, in which verification succeeds whenever there exists a close substring $\zeta^*$ with $|\zeta^*| \ge \ellrb$. Our
construction in fact satisfies this property. However, we include it
explicitly only in the definition of recoverability for two reasons: 
\begin{enumerate}
    \item It is meaningful there, as it allows us to capture that the size of the recovered output $\xi \in \mathcal{L}$ scales with the size of the approximate copy $\zeta^* \preceq \zeta$.  
    \item Standard robustness typically implies this behavior implicitly. In particular, consider
$(\Phi,\ell_{\mathrm{rb}})$-robustness for a predicate $\Phi$ satisfying the following: for any $\zeta \in \{0, 1\}^*$ with $|\zeta| \geq \ellrb$,
\[\Phi(\zeta, y) =1 \implies \exists \zeta^* \preceq \zeta \text{ with } | \zeta^*| = \ellrb \text{ and } \Phi(\zeta^*, y) = 1\;,\]
Under this condition, standard robustness implies the dynamic notion of robustness. If $\zeta$ contains a substring $\zeta'$ with $|\zeta'|>\ellrb$ that is $\Phi$-close to a response $y$, then by the above property there exists a further substring
$\zeta^* \preceq \zeta' \preceq \zeta$ of size exactly $\ellrb$ that is also $\Phi$-close to $y$.
$(\Phi, \ellrb)$-robustness then guarantees that verification succeeds on $\zeta$.
\end{enumerate}
\end{remark}

%% file: sections/preliminaries.tex
\section{Preliminaries}
\paragraph{Digital signatures.} We recall the definition of (non-robust) digital signature schemes. Digital signature schemes allow a signer to associate a message with their publicly-verifiable signature. Parties without the signing key cannot forge valid signatures on messages that the signer has not signed. Here, we define digital signatures with \textit{strong unforgeability} \cite{bellare2007two}. While standard unforgeability only prevents a forger from creating a valid signature for a new message, strong unforgeability goes further: it prevents the forger from producing a \textit{new} valid signature, even for a message whose signature it has already seen.
\begin{definition}[Digital signature scheme] Let $\lambda \in \NN$ be the security parameter, $n = \poly(\lambda)$ be the message size, and $k = \poly(\lambda)$ be the signature size. A digital signature scheme is a triple $(\Gen, \Sign, \Vrfy)$ of PPT algorithms satisfying,
\begin{itemize}
    \item $\Gen(1^\lambda):$ outputs a pair of public and secret keys $(\pk, \sk)$.
    \item $\Sign(\sk, y):$ takes in the secret key $\sk$ and a message $y \in \{0, 1\}^n$ and outputs a signature $\sigma \in \{0, 1\}^k$.
    \item $\Vrfy(\pk, \zeta, \sigma):$ takes in the public key $\pk$, a string $\zeta \in \{0, 1\}^n$, and a signature $\sigma \in \{0, 1\}^k$ and outputs a binary decision $b \in \{0, 1\}$.
\end{itemize}
\emph{\textbf{Correctness:}} For any $y \in \{0, 1\}^n$,
\begin{align*}
    \Pr\left[
    \begin{array}{rl}
        (\pk, \sk) &\leftarrow \Gen(1^\lambda)  \\
        \sigma &\leftarrow \Sign(\sk, y)
    \end{array}
    :
    \begin{array}{rl}
        \Vrfy(\pk, y, \sigma) = 1
    \end{array}
    \right] = 1\;,
\end{align*}
\noindent\emph{\textbf{Strong unforgeability:}} For any PPT adversary $\mathcal{A}$,
\begin{align*}
    \Pr\left[
    \begin{array}{rl}
        (\pk, \sk) &\leftarrow \Gen(1^\lambda) \\ 
        (\zeta, \sigma) &\leftarrow \cA^{\Sign_{\sk}}(\pk)
    \end{array} 
    :
    \begin{array}{rl}
       \Big(\Vrfy(\pk, \zeta, \sigma) = 1\Big) \; \wedge \; \Big((\zeta, \sigma) \notin \Pi\Big)
    \end{array}
    \right] \leq \negl(\lambda)\;,
\end{align*}
where $\Pi$ is the transcript $\Trans(\cA^{\Sign_{\sk}})$.
\end{definition}
\paragraph{Collision resistant hash.} The sharp sketch construction of Dodis et al. \cite{dodis2008fuzzy} and property-preserving hash construction of Holmgren et al. \cite{holmgren2022nearly} rely on the existence of collision resistant hash functions (CRHFs). We recall their definition here.
\begin{definition}[Collision resistant hash]
    \label{def:crhf}
    Let $\lambda \in \NN$ be the security parameter and let $n=\poly(\lambda)$ and $k=\poly(\lambda)$ be such that $k \le n$. A collision resistant hash family $(\Sample, \{\mathcal{H}_\lambda\})$ consists of an efficient algorithm $h \leftarrow \Sample(1^\lambda)$ that samples a hash function $h \in \mathcal{H}_\lambda$ with $h : \{0, 1\}^n \rightarrow \{0, 1\}^k$. We require,

\medskip
\noindent
\emph{\textbf{Efficiency:}} Given $x \in \{0, 1\}^n$ and $h \leftarrow \Sample(1^\lambda)$, the function $h(x)$ can be computed in time $\poly(\lambda)$.

\medskip
\noindent
\emph{\textbf{Collision-resistance:}} For any PPT adversary $\cA$,
\begin{align*}
    \Pr\left[
        \begin{array}{rl}
        h &\gets \Sample(1^\lambda) \\ 
        (x, x') &\gets \cA(1^\lambda, h)
        \end{array}
        :
         \begin{array}{rl}
       h(x) = h(x') \\ \wedge \; (x \neq x')
       \end{array}\right] \leq \negl(\lambda)\;.
\end{align*}
\end{definition}

\paragraph{Pseudorandom codes.} Pseudorandom codes (PRCs) are error-correcting codes whose codewords are indistinguishable from random binary vectors~\cite{christ2024pseudorandom}. We use ideal security (defined by Alrabiah et al.~\cite{alrabiah2024ideal}), which means that even an adversary with \emph{oracle} access to the encoder and decoder cannot distinguish between a PRC scheme and an ``ideal'' random code scheme, whose codewords are uniformly random and decoder simply responds with exact proximity queries to the underlying transcript.  Ideal security thus captures both the pseudorandomness and robustness properties that PRCs are designed to achieve. 

\itulineparagraph{PRC encoding and decoding keys.} In the PRC literature, the PRC encoding and decoding keys are often referred to as ``public'' and ``secret'' keys, respectively. However, in the watermarking setting, these roles are reversed: encoding is performed by the model provider, who holds the \emph{secret} watermark key, while decoding is performed by users, who hold the \emph{public} key and wish to verify the watermark. To avoid confusion, we refer to the PRC keys simply as the encoding and decoding keys, denoted $\ek$ and $\dk$, respectively.

\begin{definition}[Pseudorandom codes] Let $\lambda \in \mathbb{N}$ be a security parameter, $k=poly(\lambda)$ be the message size, and $n=\poly(\lambda)$ be the block size. A PRC is a triple $(\Gen, \Enc, \Dec)$ of PPT algorithms satisfying,
\begin{itemize}
    \item $\Gen(1^\lambda):$ outputs encoding and decoding keys $(\ek, \dk)$.
    \item $\Encode(\ek, \sigma):$ takes in the encoding key $\ek$ and a message $\sigma \in \{0, 1\}^{k}$ and outputs a codeword $y \in \{0, 1\}^{n}$.
    \item $\Decode(\dk, \xi):$ takes in the decoding key $\dk$ and a string $\xi \in \{0, 1\}^{n}$ and outputs message or decoding failure $\sigma \in \{0, 1\}^{k} \cup \{\perp\}$.
\end{itemize}
The information rate of the PRC is $\frac{k}{n}$.
\end{definition}
Ideal security captures the core properties that PRCs are designed to satisfy. It asserts that no PPT adversary can distinguish between a real-world interaction with the PRC and an ideal-world interaction. In the real world, the adversary interacts with the PRC encoder and decoder via oracle access. In the ideal world, the adversary receives a fresh uniformly random string for each encoding query, and the decoder maintains a complete transcript of all past encodings, decoding only those strings that are sufficiently close to previously generated ones. This notion thus captures both pseudorandomness and robustness within a single definition. It is formalized through two security games, $\cG^\real$ and $\cG^\ideal$.\\\\ \underline{$\cG^{\real}_{\lambda}(\cA)$}:
\begin{enumerate}
    \label{game:ideal}
    \item The challenger samples $(\ek, \dk) \leftarrow \PRC.\Gen(1^\lambda)$.
    \item The adversary makes encoding and decoding queries.
    \begin{itemize}
        \item For each encoding query $\sigma$, the challenger responds with $y \leftarrow$  $\PRC.\Enc(\ek,$ $\sigma)$.
        \item For each decoding query $\xi$, the challenger responds with $\sigma' \leftarrow$ $\PRC.\Decode(\dk,$ $\xi)$.
    \end{itemize}
    \item The adversary returns a decision bit $b \in \{0, 1\}$.
\end{enumerate}
\underline{$\cG^{\ideal}_{\delta, \lambda}(\cA)$}:
\begin{enumerate}
    \item The challenger initializes a transcript $\Pi = \emptyset$.
    \item The adversary makes encoding and decoding queries.
    \begin{itemize}
        \item For each encoding query $\sigma$, the challenger responds with a fresh random codeword $y \leftarrow \{0, 1\}^n$ and adds the entry $(\sigma, y)$ to the transcript $\Pi$.
        \item For each decoding query $\xi$, the challenger responds with a random message $\sigma$ satisfying $(\sigma, y) \in \Pi$ for some $y$ such that $\Ham_\delta(y, \xi) = 1$. If no such $\sigma$ exists, then the challenger responds with $\perp$.
    \end{itemize}
    \item The adversary returns a decision bit $b \in \{0, 1\}$.
\end{enumerate}
\begin{definition}[Ideal security]
\label{def:ideal-security}
Let $\lambda \in \mathbb{N}$ be the security parameter and $\delta \in (0, \frac{1}{2})$ be the error-rate. 
A pseudorandom code $(\Gen, \Enc, \Dec)$ is \emph{ideal} with error-rate $\delta$ if, for any PPT adversary $\mathcal{A}$, 
\[
    \left|\Pr\big[\cG^{\real}_{\lambda}(\cA) = 1\big] - \Pr\big[\cG^{\ideal}_{\delta, \lambda}(\cA) = 1\big]\right| \leq \negl(\lambda)\;.
\]
\end{definition}
Alrabiah et al.~\cite{alrabiah2024ideal} construct ideal pseudorandom codes with constant error rate $\delta \in (0, \frac{1}{4})$ while achieving a constant information rate $\rho \in (0, 1)$, assuming the subexponential hardness of the exact learning parity with noise problem (LPN)~\cite{jain2012commitments}.
\begin{theorem}[{\cite[Corollary 5]{alrabiah2024ideal}}]
    Assume the subexponential hardness of exact LPN \cite[Assumption 1]{alrabiah2024ideal}. For any constant $\delta \in (0, \frac{1}{4})$, there exists a constant $\rho \in (0, 1)$ such that there exists an ideal PRC with error-rate $\delta$ and information rate $\rho + o(1)$.
\end{theorem}

\subsection{Property-preserving hash} 
Property preserving hash functions (PPHs) \cite{boyle2019adversarially} allow efficient computation of a predicate $\Phi$ on two strings using only their compressed hashes, even when those inputs are chosen adversarially with full knowledge of the hash. 
\begin{definition}[Property-preserving hash]
    \label{def:pph}
    Let $\lambda \in \NN$ be the security parameter, let $\Phi$ be any predicate, and let $n=\poly(\lambda)$ and $k=\poly(\lambda)$ with $k \le n$ be the input and hash size respectively. A PPH scheme for $\Phi$ is a triple of PPT algorithms $(\Sample, \Hash, \Eval)$ satisfying
\begin{itemize}
    \item $\Sample(1^\lambda):$ samples a hash function $h$ from the family of functions $\cH_\lambda$.
    \item $\Hash(h, x):$ takes in the hash description $h$ and a string $x \in \{0, 1\}^n$ and outputs a hash $z \in \{0, 1\}^k$.
    \item $\Eval(h, z, z'):$ takes in the hash description $h$ and strings $z, z' \in \{0, 1\}^k$ and outputs an evaluation $b \in \{0, 1\}$.
\end{itemize}
\emph{\textbf{Direct-access robustness:}} For any PPT adversary $\cA$,
\begin{align*}
    \Pr\left[
        \begin{array}{rl}
        h &\gets \Sample(1^\lambda) \\ 
        (x, x') &\gets \cA(1^\lambda, h)
        \end{array}
        :
        \Eval(h, h(x), h(x')) \neq \pred(x, x') \right] \leq \negl(\lambda)\;.
\end{align*}
We use the notation $h(x) := \Hash(h, x)$.
\end{definition}
Prior works \cite{fleischhacker2021robust,fleischhacker2022property,holmgren2022nearly} have constructed PPHs for the Hamming predicate $\Ham_\delta$. In particular, assuming the existence of a special type of homomorphic CRHF, the construction of Holmgren et al. \cite{holmgren2022nearly} achieves constant-factor compression of the input when the error rate is constant.\begin{theorem}[{\cite[Corollary 6.7]{holmgren2022nearly}}]
\label{thm:pph-hamming}
Let $\lambda \in \NN$ be the security parameter, and let $n = \poly(\lambda)$ denote the input size. Suppose there exists a homomorphic CRHF with output size $\ell = \ell(\lambda)$ that is collision-resistant over $\{-1,0,1\}^n$~\cite[Definition 6.1]{holmgren2022nearly}. Then, for any constants $\delta \in (0,1/2)$ and $\rho > H_3^{\mathsf{BZ}}(\delta) \cdot \log_2 3$, where $H_3^{\mathsf{BZ}}(\cdot)$ refers to the Blokh-Zyablov bound~\cite[Fact 2.12]{holmgren2022nearly}, there exists a PPH for $\Ham_{\delta}$ with hash size $\rho \cdot n + \ell$.
\end{theorem}
\begin{remark}[Blokh-Zyablov]
\label{rem:blokh-zyablov}
    The Blokh-Zyablov bound~\cite{blokh1982linear} is what best known constructions of efficiently list-decodable codes achieve~\cite{guruswami2008better}. It is a longstanding open problem in coding theory to match the Hamming bound $H(\delta)$ with efficient codes. Note that 
    $\log \binom{n}{\delta n} = n H(\delta) +o(n)$.
\end{remark}
\begin{remark}
Theorem~\ref{thm:pph-hamming} relies on a specialized form of homomorphic CRHF, which can be constructed under either the shortest integer solution (SIS) assumption or the discrete logarithm assumption. Later, we instantiate a slightly weaker hash family that relies on standard CRHFs, but still suffices for our purposes (Theorem \ref{thm:sharp-sketch-hamming}).
\end{remark}

\subsection{Block steganography} 
We use \emph{steganography} as a tool in our watermarking constructions. Steganography is similar to watermarking but differs in three key respects. First, a steganography scheme enables messages to be robustly embedded in and decoded from generated text, rather than allowing only binary verification.
Second, we only require correctness and robustness from steganography schemes, with no guarantees on soundness, unforgeability, or recoverability (i.e., no bound on false positives). Third, we consider \emph{block steganography}, which produces a single, size $n$ block per invocation, whereas a full watermarking scheme generates an entire response, which could be arbitrarily long. This distinction is not significant: as shown below, a block-steganography scheme can be used to autoregressively generate a full response. The block interface is simply more convenient for later constructions.
\begin{definition}[Block steganography]
    \label{def:block-steg-syntax}
    Let $\lambda \in \mathbb{N}$ be the security parameter, $k=poly(\lambda)$ be the message size, and $n=\poly(\lambda)$ be the block size. A block steganography  scheme for a language model $Q$ and token set $\{0, 1\}$ is a tuple of polynomial time algorithms $(\Gen, \Embed, \Dec)$ satisfying,
\begin{itemize}
    \item $\Gen(1^\lambda):$ takes in a security parameter $\lambda$ and outputs verification and secret keys $(\vk, \sk)$.
    \item $\Embed(\sk, \pi, m):$ takes in the secret key $\sk$, a context $\pi \in \{0, 1\}^*$, and a message $m \in \{0, 1\}^k$ and outputs a block $y \in \{0, 1\}^n$.
    \item $\Dec(\vk, \zeta):$ takes in the verification key $\vk$ and input text $\zeta \in \{0, 1\}^n$  and outputs a decoded message or a failure symbol $m \in \{0, 1\}^k \cup \{\perp\}$.
\end{itemize}
\textbf{Undetectability} For every PPT distinguisher $\mathcal{A}$,
\[\left|\arraycolsep=1.4pt
  \Pr_{(\vk, \sk) \leftarrow \Gen(1^\lambda)}\left[
      \begin{array}{rl}
\mathcal{A}^{Q, \Embed_{\sk}}(1^\lambda) = 1
     \end{array}
  \right] - 
  \arraycolsep=1.4pt
  \Pr\left[
      \begin{array}{rl}
\mathcal{A}^{Q, \Response_n
}(1^\lambda) = 1      \end{array}
  \right]\right|\leq \negl(\lambda)\;,\]
where both $\Embed_{\sk}$ and $\Response_n$ take in contexts $\pi \in \{0, 1\}^*$ and message $m \in \{0, 1\}^k$ as input, but $\Response_n$ ignores the message: $\Response_n(\pi, m) := \Response_n(\pi)$.
\end{definition}
Robustness for block steganography is similar to robustness for watermarking, but allows decoding a message $m$ from an input $\zeta$ that is close to an embedding $y$ of that message.
\begin{definition}[Block Robustness]
\label{def:block-steg}
For any closeness predicate $\Phi$, a block steganography scheme $(\Gen,\Embed,\Dec)$ with block size $n = \poly(\lambda)$ is \emph{secret-key $\Phi$-robust} if for any PPT algorithm $\cA$ outputting $\zeta  \in \{0,1\}^{n}$,
\begin{align*}
    \Pr\left[
    \begin{array}{rl}
        (\vk, \sk) &\leftarrow \Gen(1^\lambda)\\
        \zeta &\leftarrow \cA^{\Embed_\sk, \Dec_\vk}(1^\lambda)
    \end{array} 
    :
    \begin{array}{rl}
    \Big(\exists (\pi, m, y) \in \Pi \text{ such that}\\ \Phi(\zeta, y) = 1 
    \wedge \; \Dec(\vk, \zeta) \neq m\Big)
    \end{array}
    \right] \le \negl(\lambda)\;,
\end{align*}
where $\Pi$ denotes the transcript $\Trans(\cA^{\Embed_{\sk}})$.
\end{definition}
\begin{remark}
Public-key $\Phi$-robustness for block steganography is defined identically to the secret key variant, except the adversary receives the verification key $\vk$.
\end{remark}
\paragraph{Robust watermarking from block steganography.}
Many existing watermarking schemes are explicitly or implicitly based on block steganography~\cite{christ2024pseudorandom,golowich2024edit,fairoze2023publicly}. Below, we describe a simple watermarking scheme $(\Gen, \Watfull, \Verfull)$ that uses a steganography scheme $(\Gen, \Embed, \Dec)$ in an autoregressive manner, embedding arbitrary messages (e.g., the message $1$) until a special $\done$ token signals that the model should output the string. While this scheme guarantees robustness, it does not provide unforgeability or recoverability, which are the focus of our construction in Section~\ref{sec:watermark-main}.
\begin{algorithm}[H]
\caption{$\Watfull$}
\label{alg:watfull}
    \begin{algorithmic}[1]
        \STATE  \textbf{input}: $(1^\lambda, \sk, x)$
        \STATE $t \leftarrow 1$
        \STATE $y_1 \leftarrow \Embed(\sk, x, 1)$
        \WHILE{$\done \not \preceq y_t$}
            \STATE $t \leftarrow t + 1$
            \STATE $y_t \leftarrow \Embed(\sk, (x, y_1, \ldots, y_{t-1}), 1)$
        \ENDWHILE
        \STATE return $y = (y_1, \ldots, y_{t - 1}, y_t[< \done])$
        \end{algorithmic}
\end{algorithm}
\noindent Here, we use the shorthand $y_t[< \done]$ to denote the substring of $y_t$ preceding the $\done$ token. The verification algorithm then examines every substring of size $n$ in its input $\zeta$, and accepts if any of them decode to a message other than $\perp$.
\begin{algorithm}[H]
\caption{$\Verfull$}
\label{alg:verfull}
    \begin{algorithmic}[1]
        \STATE  \textbf{input}: $(1^\lambda, \vk, \zeta)$
        \STATE $n \leftarrow n(\lambda)$
        \FOR{$i \in [|\zeta| - n + 1]$} 
            \STATE $m \leftarrow \Dec(\vk, \zeta[i:i+n])$
            \IF{$m \neq \perp$}
                \STATE return $1$
            \ENDIF 
        \ENDFOR
        \item return $0$
    \end{algorithmic}
\end{algorithm}
In this construction, each block $y_t$ of a response
$y \leftarrow \Watfull(\sk, x)$ is generated as an embedding $y_t \leftarrow \Embed(\sk, \pi, 1)$.  
Consequently, in $\Verfull(\vk, \zeta)$, as long as the input $\zeta$ contains an $n$-bit substring $\zeta^* \preceq \zeta$ that is $\Phi$-close to some block $y_t$ of a response $y$, the decoder $\Dec(\vk, \zeta^*)$ outputs $1$ except with negligible probability.  
This, in turn, ensures that $\Verfull(\vk, \zeta)$ also outputs $1$.  Recall that this condition is captured exactly by the predicate $\EBC{\Phi}{n}(\zeta^*, y)$ (Definition~\ref{def:every-block-close}), which, given $\zeta^*$ of size $n$, outputs 1 iff it is $\Phi$-close to a block of $y$. We formalize this guarantee in Fact~\ref{fact:full-wm-from-block}.
\begin{fact}
\label{fact:full-wm-from-block}
    Let $\lambda \in \NN$ be the security parameter, $n = \poly(\lambda)$ be any block size, $\Phi$ be any closeness predicate, and $Q$ be any language model. If $(\Gen,\Embed,\Dec)$ is a block size $n$ steganography scheme for $Q$ with (secret or public key) $\Phi$-robustness, then the watermarking scheme $(\Gen, \Watfull, \Verfull)$ for $Q$ satisfies (secret or public key) $(\EBC{\Phi}{n}, n)$-robustness.
\end{fact}
\begin{proof}
Undetectability follows immediately from the undetectability of the underlying
steganography scheme. We now prove robustness for the secret key case; the
public key version is analogous.

Fix a PPT adversary $\mathcal{A}$ and a security parameter $\lambda \in \NN$.
Let $(\vk, \sk) \gets \Gen(1^\lambda)$, let 
$\zeta \gets \mathcal{A}^{\Watfull_{\sk}, \Verfull_{\vk}}(1^\lambda)$, and let $\Pi$ denote the transcript $\Trans(\cA^{\Watfull_{\sk}})$.
Assume, toward contradiction, that there exists a polynomial 
$q : \NN \rightarrow \NN$ such that
\begin{align*}
    \Pr\left[
    \begin{array}{rl}
    \big(\exists \zeta^* \preceq \zeta, (x, y) \in \Pi \text{
    such that } 
    \EBC{\Phi}{n}(\zeta^*, y) = 1
    \wedge \; |\zeta^*| = n \big)\\
    \wedge \; \Verfull(\vk, \zeta) = 0
    \end{array}
    \right] \geq 1/q(\lambda)\;.
\end{align*}
Whenever the above event occurs, the predicate 
$\EBC{\Phi}{n}(\zeta^*, y)=1$, with $|\zeta^*| = n$, guarantees the existence of a block index 
$t \in \NN$ for which
\[
\Phi(\zeta^*, y_t) = 1\;.
\]
Because each block $y_t$ is produced by $\Embed_{\sk}$ during $\mathcal{A}$'s
interaction with $\Watfull_{\sk}$, we have
\[
(x, y) \in \Trans(\mathcal{A}^{\Watfull_{\sk}})
\;\Rightarrow\;
\Big((\pi = (x, y_1, \ldots, y_{t-1}),\, m = 1),\, y_t\Big)
\in \Trans(\mathcal{A}^{\Embed_{\sk}})\;.
\]
Yet in the same event, we also have $\Verfull(\vk, \zeta)=0$, which implies 
$\Dec(\vk, \zeta^*) = \perp$, since failed verification guarantees that 
\emph{every} $n$-bit substring of $\zeta$ decodes to $\perp$. 
Consequently, there exists an adversary $\mathcal{B}$ that simulates $\mathcal{A}$
and outputs the substring $\zeta^*$, which is $\Phi$-close to an embedding block 
$y_t$ but fails to decode with probability at least $1/q(\lambda)$. This violates the 
$\Phi$-block robustness of the underlying steganography scheme, yielding a contradiction.
\end{proof}

\paragraph{Block steganography instantiations.}
We present existing constructions for both secret and public key robust block steganography schemes for language models with \emph{sufficiently high entropy}. 
\itulineparagraph{Secret key block steganography from PRCs.}
For the secret key setting, recall PRCs, which closely resemble block steganography schemes for the uniform language model (e.g., $Q(\pi) = 1/2$ for all $\pi \in \{0, 1\}^*$). Christ and Gunn~\cite{christ2024pseudorandom} showed that PRCs can be used to build a robust watermarking scheme. However, their construction leverages standard PRCs, which achieve a weaker form of robustness, where the robustness adversary is modeled as an error channel acting independently of the encoding and decoding keys. By plugging \emph{ideal} PRCs \cite{alrabiah2024ideal} (as in Definition~\ref{def:ideal-security}) into their construction, we essentially obtain the adaptive guarantee that we require.

A subtlety is that ideal PRCs guarantee that an input 
$\xi$ decodes to a \emph{random} message $m$ whose encoding $y$ is close to $\xi$, rather than the \emph{unique} message $m$ as required in block steganography. We show that this distinction is harmless: due to pseudorandomness, no efficient adversary can find any $\xi$ that is simultaneously close to the encodings of two distinct messages.

\begin{theorem}[Secret-key block steganography from ideal PRC \cite{christ2024pseudorandom,alrabiah2024ideal}]
    \label{thm:secret-key-block-wm}
    Let $\lambda \in \NN$ be the security parameter, let $n = \poly(\lambda)$ with $n \geq \lambda$ be any block size, and let $Q$ be any language model with min-entropy rate at least a constant. That is, there exists a constant $\epsilon \in (0, \frac{1}{4})$ such that 
    \[\forall \pi\in \{0, 1\}^*, y \leftarrow \Response_\ell(\pi) \text{ satisfies } H_\infty(y) \geq 4\Big(\sqrt{\frac{1}{2}-\epsilon}\Big)n\]
    If there exists an ideal pseudorandom code with constant information rate $\rho \in (0, 1)$ and constant error rate $\delta \in (0, \frac{1}{4})$, then there exists a block size $n$ \emph{secret key} $\Ham_{\delta-\epsilon}$-robust steganography scheme for $Q$ with message size $k = \rho n$.
\end{theorem}

\begin{proof}
Let $\PRC = (\Gen, \Enc, \Dec)$ be a block size $n$ ideal PRC, assumed to exist in the theorem. We define $\BWM = (\Gen, \Embed, \Dec)$, which mirrors the watermarking scheme of \cite{christ2024pseudorandom} but allows embedding messages and only outputs $n$-sized blocks.
\begin{itemize}
    \item $\BWM.\Gen(1^\lambda): $ samples encoding and decoding keys $(\ek, \dk) \leftarrow \PRC.\Gen(1^\lambda)$ and outputs $(\vk = \dk, \sk = \ek)$.
\item $\BWM.\Embed(\sk, \pi, m)$: encodes the message 
    $c \gets \PRC.\Encode(\sk, m)$ and then samples each bit of the response block 
    $y \in \{0,1\}^n$ as
    \[
    y[i] \gets \Ber\Big(p_i - (-1)^{c[i]} \cdot \min\{p_i, 1-p_i\}\Big), \quad \forall i \in [n],
    \]
    where $p_i  = Q(\pi, y[1:i])$ is the probability that the $i$-th bit should be $1$ under $Q$. Then outputs $y$.
        \item $\BWM.\Dec(\vk, \zeta \in \{0, 1\}^n)$ outputs $m' \leftarrow \PRC.\Decode(\vk, \zeta)$.
\end{itemize}

Undetectability follows from pseudorandomness of $\PRC$ \cite[Lemma 19]{christ2024pseudorandom}. For robustness fix $\lambda \in \NN$ and a PPT adversary $\cA$, and let $(\PRC.\ek, \PRC.\dk) \leftarrow \PRC.\Gen(1^\lambda)$ and $\zeta \leftarrow \mathcal{A}^{\Embed_{\sk},\Dec_{\vk}}(1^\lambda)$. We will show that $\cA$ cannot break the $\Phi$-block robustness of $\BWM$.

We move to a hybrid in which the oracles $\Embed_{\sk}$ and $\Dec_{\vk}$ are implemented using
$\PRC.\Enc^{\ideal}$ and $\PRC.\Dec^{\ideal}$ respectively (the oracles given in the ideal PRC game \ref{game:ideal}, $\mathcal{G}^{\ideal}_{\lambda, \delta}(\cA)$).
By the ideal security of $\PRC$, this hybrid is computationally 
indistinguishable to~$\mathcal{A}$.\\\\
Now consider the event that $\zeta$ is $\Ham_{\delta - \epsilon}$-close to some block $y \in \{0, 1\}^n$ for which 
\[
(\pi,m,y) \in \Trans(\mathcal{A}^{\BWM.\Embed_{\sk}})\;.
\]
We will show that $\zeta$ decodes to $m$ with overwhelming probability. Since $Q$ has entropy rate at least a constant, every response $y$ has high empirical
entropy as in \cite[Definition 7]{christ2024pseudorandom}. From~\cite[Lemma 21]{christ2024pseudorandom}, $y$ must therefore be $\Ham_\epsilon$-close to the PRC 
codeword $c$ used to sample it with all but negligible probability. That is, with probability at least $1 - \negl(\lambda)$, there exists a codeword $c \in \{0, 1\}^n$ such that
\begin{equation}
\label{event:close-codeword}
(m,c) \in \Trans(\mathcal{A}^{\PRC.\Enc^{\ideal}}) 
\text{ and }
\Ham_\epsilon(y,c)=1\;.
\end{equation}
By the triangle inequality, $\zeta$ is 
$\Ham_\delta$-close to $c$. We want to argue that the ideal 
decoder will therefore decode $\zeta$ to the corresponding message $m$. 
However, the ideal decoder only guarantees decoding to \emph{a random} message 
satisfying~\ref{event:close-codeword}. We now show that $m$ is, with overwhelming probability, the unique such 
message, which will conclude the argument.

Since $\PRC.\Enc^{\ideal}$ generates a fresh random codeword for each input, all 
codewords in the transcript are, with overwhelming probability, pairwise at a 
Hamming distance greater than $2\delta n$. Formally, let
\[
C \;=\; \{\, c : \exists m' \text{ with } (m',c) \in \Trans(\mathcal{A}^{\Enc^{\ideal}}) \,\}.
\]
By Lemma~\ref{lem:hamming-far} and the fact that codewords from the ideal encoder are uniformly random strings, the probability that any two distinct 
$c, c' \in C$ satisfy
\[
\Ham_{2\delta}(c,c') = 1
\]
is negligible in $n \geq \lambda$. Since $|C|$ is polynomial in $\lambda$, a union bound over all pairs implies that, with negligible probability, no pair of distinct codewords in $C$ are $\Ham_{2\delta}$-close. Consequently, if $\zeta$ is $\Ham_\delta$-close to some codeword 
$c \in C$, then it cannot be $\Ham_\delta$-close to any other codeword in $C$.

Combining this fact with~\ref{event:close-codeword}, we have that whenever 
$\zeta$ is $\Ham_{\delta - \epsilon}$-close to some block $y$ encoding a 
message $m$, the ideal decoder outputs $m$ with all but negligible probability. 
Since the ideal decoder is indistinguishable from 
$\PRC.\Dec_{\vk}$ to $\mathcal{A}$, it follows that 
$\BWM$ satisfies $\Ham_{\delta-\epsilon}$-block robustness.
\end{proof}

The proof above relied on a simple fact, that random strings are far in Hamming distance with overwhelming probability.
\begin{lemma}
\label{lem:hamming-far}
Let $n \in \mathbb{N}$ be any string length and let $\delta \in (0,\tfrac12)$ be a fixed constant. Then
\[
\Pr_{c,c' \leftarrow \{0,1\}^n}\big[\Ham_\delta(c,c') = 1\big] \le \negl(n).
\]
\end{lemma}

\begin{proof}
For independently and uniformly random $c,c' \leftarrow \{0,1\}^n$, each coordinate differs with probability $1/2$, so
\[
\mathbb{E}\big[\|c - c'\|_0\big] = \frac{n}{2}.
\]
Applying a Chernoff bound with $t = 1 - 2\delta \in (0, 1)$, we get
\[
\Pr\big[\|c - c'\|_0 \le \delta n\big]
\leq \exp\!\left(-\frac{n (1 - 2\delta)^2}{4}\right).
\]
Since $\delta$ is a constant strictly less than $1/2$, the exponent is a negative constant times $n$, and therefore the probability is negligible in $n$.
\end{proof}
\itulineparagraph{Public key block steganography in random oracle model.}
Fairoze et al. \cite{fairoze2023publicly} leverage random oracles to build a \emph{pubic key} watermarking scheme. We observe that their scheme is implicitly built on top of a public key block steganography scheme, which we formalize in the following theorem.
\begin{theorem}[Public key block steganography from random oracles \cite{fairoze2023publicly}]
\label{thm:public-key-block-wm}
    Let $\lambda \in \NN$ be the security parameter, $n = \poly(\lambda)$ be any block size, and $Q$ be any language model satisfying the following entropy condition: 
    for some $\ell = \poly(\lambda)$, size $\ell$ outputs from $Q$ have high min-entropy. That is,
    \[\forall \pi \in \{0, 1\}^*, y\leftarrow \Response_\ell(\pi) \text{ satisfies } H_\infty(y) \geq \omega (\log \lambda) \;.\]
    For simplicity, assume $\ell$ divides $n$, so that $n/\ell$ is an integer.  

    Then there exists constants $\delta \in (0, \frac{1}{2})$, $\rho \in (0, 1)$ where in the random oracle model, there exists a block size $n$ \emph{public key} $\Phi_{\delta}$-robust block steganography scheme with message size $k = \rho(n/\ell - 1)$, 
     where $\Phi_\delta : \{0, 1\}^n \times \{0, 1\}^n \rightarrow \{0, 1\}$ is the predicate
    \[
    \Phi_\delta(y, y') :\;
    \left(
    \begin{aligned}
    &\text{Parse } (y_1, \ldots, y_{n/\ell}) = \Blocks(y),\;
      (y_1', \ldots, y_{n/\ell}') = \Blocks(y') \text{ into } \ell\text{-bit blocks},\\
    &\text{Output } 1 \text{ iff } y_1 = y_1' \text{ and }\text{for a } (1-\delta)\text{-fraction of } i \in \{2, \ldots, n/\ell\},\; y_i = y_i'
    \end{aligned}
    \right).
    \]
\end{theorem}
\begin{proof}
We adapt the construction of \cite{fairoze2023publicly}. Let $(\Enc, \Dec)$ be an information-theoretic
error-correcting code with constant rate $\rho\in (0,1)$, that can correct up to a  $\delta \in (0, \frac{1}{2})$ fraction of errors\footnote{Such codes exists, for example \cite{alon1992construction,naor1990small,ta2017explicit}.}.
 Let $H_1$ and $H_2$ be independent random oracles. We define a block steganography scheme $\BWM = (\Gen, \Embed, \Dec)$:
\begin{itemize}
    \item $\BWM.\Gen(1^\lambda): $ samples $r \leftarrow\{0, 1\}^\lambda$ and outputs $\vk=\sk=r$.
\item $\BWM.\Embed(r, \pi, m):$ first samples a block of $\ell$ bits  $y_1 \leftarrow \Response_\ell(\pi)$, then computes the masked codeword
\[c = \Big(\Enc(m) \oplus H_2\big((r, y_1)\big)\Big) \in \{0, 1\}^{n/\ell - 1} \;.\]
The algorithm then embeds $c$ into the next $n/\ell - 1$ blocks of size $\ell$ by rejection-sampling. That is,
\[y_i \leftarrow \Response_\ell\big((\pi, y_1, \ldots, y_{i-1})\big) \text{  conditioned on }H_1\big((r, y_i)\big) = c[i-1] \;,
\]
for $i \in \{2, \ldots, n/\ell\}$. Then outputs  $y = (y_1, \ldots, y_{n/\ell}) \in \{0, 1\}^n$.

\item $\BWM.\Ver(r, \zeta \in \{0, 1\}^n):$ parses $\zeta = (\zeta_1, \ldots, \zeta_{n/\ell})$ into blocks of size $\ell$, then computes
\[
\pad = H_2\big((r,\zeta_1)\big), \text{ and }
c'[i-1] = H_1\big((r,\zeta_{i})\big) \text{ for } i \in \{2,\ldots,n/\ell\},
\]
and outputs $m = \Dec(c' \oplus \pad)$.
\end{itemize}
Undetectability and the fact that the rejection sampling succeeds for every block with all but negligible probability follows directly from \cite[Lemma 3.2]{fairoze2023publicly}. Note that using their proof requires masking the message with the random oracle value of the first block $H_2\big((r,y_1)\big)$.\\\\
To see robustness, let $\zeta = (\zeta_1, \ldots, \zeta_{n/\ell}) \gets \mathcal{A}^{\Embed_{r}}(r)$, and suppose there exists
$(\pi,m,y) \in \Trans(\cA^{\Embed_{r}})$ with
\[
\Phi_\delta(\zeta,y)=1.
\]
Denote $(y_1,\ldots,y_{n/\ell})$ as the $\ell$-bit blocks of $y$. By definition of $\Phi_\delta$, we have:
\begin{enumerate}
    \item  $y_1 = \zeta_1$, so the verifier $\Ver(r, \zeta)$ recovers the pad, $\pad = H_2\big((r,y_1)\big)$, exactly.
   \item Among the blocks $(\zeta_2, \ldots, \zeta_{n/\ell})$, at most a $\delta$-fraction differ from the corresponding $(y_2,\ldots,y_{n/\ell})$. Thus at most a $\delta$-fraction of the bits of $c'$ computed by $\BWM.\Ver$ differ from the initial encoding \(c = \Enc(m) \oplus \pad\).
\end{enumerate}
Because the error-correcting code can correct up to a $\delta$-fraction of errors, $\Dec(c'\oplus \pad) = m$.
\end{proof}

%% file: sections/rds-formal.tex
\section{Robust and Recoverable Digital Signatures}
\label{sec:rds}
We present our formal definition of robust digital signatures, which allow verification even in the presence of small, adversarial perturbations to the message.
\begin{definition}[Robust digital signatures with strong unforgeability]
\label{def:rds-formal}
Let $\lambda \in \NN$ be the security parameter and let $\Phi$ be any predicate. A digital signature scheme $(\Gen, \Sign, \Ver)$ is robust and strongly unforgeable with respect to $\Phi$ if for any PPT adversary $\mathcal{A}$, the scheme satisfies robustness,
\begin{align*}
    \Pr\left[ 
    \begin{array}{rl}
        (\pk,\sk) &\leftarrow \Gen(1^\lambda) \\
        (\zeta, \sigma) &\leftarrow \cA^{\Sign_{\sk}}(\pk)
    \end{array}
    :\begin{array}{rl}
    \;
    \Big(\exists (y, \sigma) \in \Pi\footnotemark \text{ such that } \Phi(\zeta, y) = 1\Big)
    \\ \wedge \; \Ver(\pk, \zeta,\sigma) = 0
    \end{array}
    \right] \le \negl(\lambda)\;,
\end{align*}
\footnotetext{We slightly overload notation by quantifying over tuples when some components are fixed by context. For example here, the statement $\exists (y, \sigma) \in \Pi$, where $\sigma$ is the signature output by the adversary, means $\exists y$ such that $(y, \sigma) \in \Pi$.}
and strong unforgeability,
\begin{align*}
    \Pr\left[
    \begin{array}{rl}
        (\pk,\sk) &\leftarrow \Gen(1^\lambda) \\
        (\zeta, \sigma) &\leftarrow \cA^{\Sign_{\sk}}(\pk)
    \end{array}
    :\begin{array}{rl}
    \;
    \Ver(\pk, \zeta,\sigma) = 1 \\
    \wedge \; \Big(\forall (y, \sigma) \in \Pi, \; \Phi(\zeta, y) = 0\Big)
    \end{array}\right] \le \negl(\lambda)\;,
\end{align*}
where $\Pi$ is the transcript $\Trans(\cA^{\Sign_{\sk}})$.
\end{definition}
A recoverable digital signature scheme, in turn, requires a PPT recovery algorithm $\Rec$ that can recover the original message from its signature and a nearby, adversarially chosen input. For a digital signature scheme with message size $n=n(\lambda)$, signature size $k = k(\lambda)$, $\Rec$ has the syntax
\begin{itemize}
    \item $\Rec(\pk,\zeta, \sigma):$ takes in the public key $\pk$, a string $\zeta \in \{0, 1\}^n$, and a signature $\sigma \in \{0, 1\}^k$ and outputs a recovered message or a failure symbol $y_{\mathsf{rec}} \in \{0, 1\}^n \cup \{\bot\}$.
\end{itemize}
For $\Rec$ to satisfy strong unforgeability, any recovered message $y_{\mathsf{rec}}$ must indeed be close to the input $\zeta$ and must have been originally signed with the input signature $\sigma$.
\begin{definition}[Recoverable digital signatures]
\label{def:rds-recoverability-formal}
Let $\lambda \in \NN$ be the security parameter and let $\Phi$ be any predicate. A digital signature scheme $(\Gen, \Sign, \Recover)$ is recoverable with respect to $\Phi$ if for any PPT adversary $\mathcal{A}$, the scheme is robustly recoverable, meaning
\begin{align*}
    \Pr\left[ 
    \begin{array}{rl}
        (\pk,\sk) &\leftarrow \Gen(1^\lambda) \\
        (\zeta, \sigma) &\leftarrow \cA^{\Sign_{\sk}}(\pk)\\
        y_{\rec} &\leftarrow \Recover(\pk, \zeta, \sigma)
    \end{array}
    :
    \begin{array}{rl}
    \;
    \Big(\exists (y, \sigma) \in \Pi \text{ such that } \Phi(\zeta, y) = 1\Big)\\
    \wedge \; y_{\rec} \neq y
     \end{array}
     \right] \le \negl(\lambda)\;,
\end{align*}
and unforgeably recoverable, meaning
\begin{align*}
    \Pr\left[
    \begin{array}{rl}
        (\pk,\sk) &\leftarrow \Gen(1^\lambda) \\
        (\zeta, \sigma) &\leftarrow \cA^{\Sign_{\sk}}(\pk)\\
        y_{\rec} &\leftarrow \Recover(\pk, \zeta, \sigma)
    \end{array}
    :
    \begin{array}{rl}
    \;
    y_{\rec} \neq \perp \\ \;
    \wedge \; \Big((y_{\rec}, \sigma) \notin \Pi \; \vee \; \Phi(\zeta, y_{\rec}) = 0\Big)
     \end{array}\right] \le \negl(\lambda)\;,
\end{align*}
where $\Pi$ is the transcript $\Trans(\cA^{\Sign_{\sk}})$.
\end{definition}

Another property we introduce is a notion of \emph{uniqueness} for signatures. While this property is not required for our main result, it will be useful to guarantee a certain chaining property that we present in Section \ref{sec:watermark-chain}. A digital signature has unique signatures if no efficient adversary can produce an output $\zeta$ and two messages $(y, y')$ with the same signature such that $\zeta$ is close to one message and far from another. Intuitively, this prevents an adversary from generating duplicate signatures on messages for which it can produce an intermediate output. Note that this property can be easily achieved by appending $\omega(\log \lambda)$ random bits to each signature, to ensure \emph{all} signatures are almost surely unique. However, we show that our construction inherently satisfies it.
\begin{definition}[Unique signatures]
\label{def:unique-signatures}
Let $\lambda \in \NN$ be the security parameter and let $\Phi$ be any predicate. A digital signature scheme $(\Gen, \Sign, \Ver)$ has unique signatures with respect to $\Phi$, if for any PPT $\cA$,
\begin{align*}
    \Pr\left[ 
    \begin{array}{rl}
        (\pk, \sk) \leftarrow \Gen(1^\lambda) \\
        (\zeta, y, y') \leftarrow \cA^{\Sign_{\sk}}(\pk)
    \end{array}
    :
    \begin{array}{rl}
    \;
    \Big(\Phi(\zeta, y) = 1 \; \wedge \; \Phi(\zeta, y') = 0\Big)\\
    \wedge \;\Big(\exists\sigma \text{ such that } (y, \sigma) \in \Pi \text{ and } (y', \sigma) \in \Pi \Big)
     \end{array}
     \right] \le \negl(\lambda)\;,
\end{align*}
where $\Pi$ is the transcript $\Trans(\cA^{\Sign_{\sk}})$.
\end{definition}

\subsection{Difference-recovering hash family} 
\label{sec:difference-recovery}
The core ingredient of our signature scheme is a property-preserving hashing scheme (PPH) \cite{boyle2019adversarially}. We observe that existing Hamming PPH constructions~\cite{fleischhacker2022property,holmgren2022nearly} support recovery of the bitwise difference $\Delta = x \oplus x'$ between two close strings $x$ and $x'$ using only their hashes $h(x)$ and $h(x')$. We define a slightly weaker property than this, which requires recovering the difference between two inputs when one of the inputs is given in the clear. This property is sufficient for watermarking applications, and we show that it can be achieved by \textit{any} PPH for the Hamming predicate. In Section \ref{sec:homomorphic-pph}, we define the stronger variant, which requires recovering the difference using only the hashes of both inputs.
\begin{definition}[Difference recovering PPH]\label{def:PPH-recoverability}
    For any predicate $\Phi$, a hash family $(\Sample, \{\mathcal{H}_\lambda\})$ is difference-recovering w.r.t. $\Phi$ if there exists a PPT algorithm $\Recover$ such that for any PPT adversary $\cA$,
\begin{align*}
    \Pr\left[
        \begin{array}{rl}
        h &\leftarrow \Sample(1^\lambda) \\ (x, x') &\leftarrow \cA(1^\lambda, h)\\
        \Delta &\leftarrow \Recover(h, h(x), x')
        \end{array}
        :
        \begin{array}{rl}
        \big(\Delta \neq x \oplus x' \wedge \Phi(x, x') = 1\big) \\
        \vee\; \big(\Delta \neq \bot \wedge \Phi(x,x') = 0\big)
        \end{array}
    \right] \leq \negl(\lambda)
\end{align*}
\end{definition}
\begin{proposition}
\label{prop:pph-recover-app} Let $\lambda$ be the security parameter. For any input size $n=\poly(\lambda)$ and threshold $r=\poly(\lambda)$ such that $r < n/2$, every $\Ham_{r/n}$-PPH supports difference recovery w.r.t. $\Ham_{r/n}$.
\end{proposition}
\begin{figure}[htbp]
\begin{minipage}[t]{0.45\textwidth}
\begin{algorithm}[H]
\begin{algorithmic}[1]
    \caption{$\Recover$}
    \STATE  \textbf{input}: $(h, z, x')$
    \IF{If $\Eval(h, z, h(x')) = 0$}
        \STATE \textbf{return} $\perp$
    \ENDIF
    \STATE $\Delta \gets 0^n$
    \STATE $\zeta \leftarrow \GetBoundaryPoint(h, z, x')$
    \FOR{$j \in [n]$}
        \IF{$\Eval(h, z, h(\zeta \oplus e_j)) = 1$}
                \STATE $\Delta \gets \Delta \oplus e_j$
        \ENDIF
    \ENDFOR
    \STATE \textbf{return} $\Delta \oplus x' \oplus \zeta$
    \end{algorithmic}
\end{algorithm}
\end{minipage}
\hfill
\begin{minipage}[t]{0.45\textwidth}
  \begin{algorithm}[H]
    \begin{algorithmic}[1]
        \caption{$\GetBoundaryPoint$}
        \STATE  \textbf{input}: $(h, z, x')$
        \STATE $\zeta \gets x'$
        \FOR{$j \in [n]$}
            \STATE $\zeta \gets \zeta \oplus e_j$
            \IF{$\Eval(h, z, h(\zeta)) = 0$}
                \STATE \textbf{return} $\zeta \oplus e_j$
            \ENDIF
        \ENDFOR
        \STATE \textbf{return} $\bot$
        \end{algorithmic}
    \end{algorithm}
\end{minipage}
\caption{Generic PPH recovery algorithm}
\label{alg:pph-recover} 
\end{figure}
\begin{proof}
Let $h \leftarrow \Sample(1^\lambda)$. We construct a recovery algorithm $\Recover(h, \cdot, \cdot)$ which, given a hash $h(x)$ and an input $x'$ in the clear, returns the difference $x \oplus x'$ if $\Ham_{r/n}(x, x') = 1$, and returns $\bot$ otherwise. The recovery algorithm only needs oracle access to $\Hash(h,\cdot)$ and $\Eval(h,h(x),\cdot)$. For the reader’s convenience, we include the pseudocode for $\Recover$ in Figure~\ref{alg:pph-recover}.
    
We may assume that $\Eval(h, h(x), h(\zeta)) = \Ham_{r/n}(x, \zeta)$ for \emph{all} inputs $\zeta \in \mathbb{F}_2^n$ used by $\Recover$ to query $\Eval(h, h(x), \cdot)$, since $\Recover$ is an efficient deterministic algorithm. Any violation would contradict the direct-access robustness of the underlying Hamming PPH.

If $\Eval(h, h(x), h(x')) = 0$, then $\Recover(h, h(x), x')$ returns $\bot$, which is fine since $ \Ham_{r/n}(x,x') = \Eval(h, h(x), h(x')) = 0$.

If $\Eval(h, h(x), h(x')) = 1$, then $\Recover(h, h(x), x')$ needs to return the difference $x \oplus x'$. To this end, it first finds a \emph{boundary point} relative to $x$. That is, a point $\zeta \in \FF_2^n$ such that $\|x \oplus \zeta\|_0 = r$. Once such a point is found, difference recovery is straightforward: since we know that $\zeta \oplus x$ is $r$-sparse, we can determine its support of the by flipping each bit of $\zeta$ individually. Specifically, for each $j \in [n]$, we test whether $\Eval(h, h(x), h(\zeta \oplus e_j)) = 1$, where $e_j \in \FF_2^n$ denotes the $j$-th standard basis vector (i.e., the vector with a $1$ in the $j$-th position and zeros elsewhere). We conclude that $j$ lies in the support of $x \oplus \zeta$ if and only if the evaluation remains $1$ after the one-bit perturbation.

We now describe an efficient procedure for finding a boundary point $\zeta$ with respect to $x$, given $h(x)$ and $x'$ satisfying $\|x \oplus x'\|_0 \le r$. For any $j \in [n]$, define
\begin{align*}
    u_j &= e_1 \oplus \cdots \oplus e_j\;,\quad \zeta_j = x' \oplus u_j\;.
\end{align*}

Define the potential function $f: \{0,1,\ldots,n\} \to \{0,1,\ldots, n\}$ by
\begin{align*}
    f(j) = \|x \oplus \zeta_j\|_0\;,
\end{align*}
with $f(0) = \|x \oplus x'\|_0$.

Observe that for each $j < n$, we have $|f(j+1) - f(j)| = 1$. In other words, the value of $f(j)$ always changes by either $1$ or $-1$. In addition, it holds that $f(0) \le r$ and $f(n) = n - f(0) > r$. Thus, by the discrete intermediate value property, there exists $s \in \{0,1,\ldots, n\}$ such that $f(s) = r$. The corresponding $\zeta_s$ is the desired boundary point.

Although $\Recover$ cannot explicitly compute $f(j)$, it can determine whether $f(j) \le r$ by querying $\Eval(h, h(x), h(\zeta_j))$. The algorithm increments $j$ from $1$ to $n$, stopping at the first index where $\Eval$ returns $0$. Let $j$ be the smallest such index. Then the boundary point is $\zeta_{j-1}$.
\end{proof}

\paragraph{Difference recovery from standard CRHF.} Proposition~\ref{prop:pph-recover-app} implies that the PPH from homomorphic CRHFs given in Theorem~\ref{thm:pph-hamming} satisfies difference recovery with respect to $\Ham_\delta$. However, the converse does not hold. In particular, one can instantiate a difference recovering hash family for Hamming that does not admit an evaluation algorithm $\Eval$ for computing the predicate on two hashes, as the recovery algorithm $\Rec$ receives one of its inputs in the clear. This instantiation is heavily inspired by the notion and construction of \emph{secure sketches}~\cite{dodis2008fuzzy}. We refer the reader to Section~\ref{sec:sharp-sketch} for additional details.

\begin{theorem}[{Sharp sketch for Hamming, adapted from~\cite[Section 8.3]{dodis2008fuzzy}}]
\label{thm:sharp-sketch-hamming}
Let $\lambda \in \NN$ be the security parameter, and let $n=\poly(\lambda)$ denote the input size. Suppose there exists a CRHF with output size $\ell=\ell(\lambda)$. For any constants $\delta \in (0,1/2)$ and $\rho > H_2^{\mathsf{BZ}}(\delta)$, where $H_2^{\mathsf{BZ}}(\cdot)$ refers to the Blokh-Zyablov bound, there exists a difference recovering hash family $(\Sample, \Rec)$ w.r.t. $\Ham_\delta$ that has hash size $\rho \cdot n + \ell$.
\end{theorem}

\begin{remark}[Comparison with~\refThm{thm:pph-hamming}]
The hash family from Theorem \ref{thm:sharp-sketch-hamming} requires only a standard collision-resistant hash function, not a homomorphic one. Moreover, it uses efficiently list-decodable codes over $\FF_2$, whereas~\refThm{thm:pph-hamming} uses codes over $\FF_3$. As a result, it benefits from the Blokh-Zyablov bound over $\FF_2$ and avoids the multiplicative factor of $\log_2 3$ in the hash size. However, it is not a PPH.
\end{remark}

\subsection{Robust signature from PPH}
\label{sec:rds-construction}
We now prove a generic theorem showing that the construction presented in Section~\ref{sec:rds-overview} is both robust and unforgeable with respect to $\Phi$. In Section \ref{sec:rds-instantiation}, we instantiate the theorem with an existing Hamming PPH \cite{holmgren2022nearly}. We recall the construction in Figure~\ref{alg:rds-construct}.

\begin{figure}[htbp]
\centering
\begin{minipage}[t]{0.3\textwidth}
    
  \begin{algorithm}[H]
    \caption{$\RGen$}\label{alg:rgen-app}
    \begin{spacing}{1.2} 
    \begin{algorithmic}[1]
        \STATE  \textbf{input}: $(1^\lambda)$
        \STATE $h \leftarrow \Sample(1^\lambda)$
        \STATE $(\pk, \sk) \leftarrow \Gen(1^\lambda)$
        \STATE \textbf{return} $((h, \pk), \sk)$
    \end{algorithmic}
    \end{spacing}
\end{algorithm}
\end{minipage}
\hfill
\begin{minipage}[t]{0.3\textwidth}
  \begin{algorithm}[H]
    \caption{$\RSign$}\label{alg:rsign-app}
    \begin{spacing}{1.2} 
    \begin{algorithmic}[1]
        \STATE  \textbf{input}: $((h, \pk), \sk, y)$
        \STATE $z \leftarrow h(y)$
        \STATE $\tau \leftarrow \Sign(\sk, z)$
        \STATE \textbf{return} $(z, \tau)$
    \end{algorithmic}
    \end{spacing}
\end{algorithm}
\end{minipage}
\hfill
\begin{minipage}[t]{0.33\textwidth}
\begin{algorithm}[H]
    \caption{$\RVrfy$}
    \label{alg:rvrfy-app}
    \begin{spacing}{1.2}
    \begin{algorithmic}[1]
        \STATE  \textbf{input}: $((h,\pk), \zeta, \sigma)$
        \STATE $(z, \tau) \leftarrow \sigma$
        \IF{$\Ver(\pk, z, \tau) = 1$} 
        \STATE $b \leftarrow \Eval(h, h(\zeta), z)$
        \STATE \textbf{return} $b$
        \ENDIF
        \STATE \textbf{return} 0
        \end{algorithmic}
        \end{spacing}

\end{algorithm}
\end{minipage}
\caption{Robust digital signature construction}
\label{alg:rds-construct}
\end{figure}

\begin{theorem}[Generic robust signatures]
    \label{thm:rds-appdx-main}
    Let $\lambda \in \NN$ be the security parameter, $n = \poly(\lambda)$ be any message size, and $\Phi$ be any predicate. If $\PPH[\Sample, \Hash, \Eval]$ is a $\Phi$-PPH with hash size $k = \poly(\lambda)$ and $\DSS[\Gen, \Sign, \Ver]$ is a strongly unforgeable digital signature scheme with signature size $\ell = \poly(\lambda)$, then $(\RGen, \RSign, \RVrfy)$, defined in Figure~\ref{alg:rds-construct}, is a robust and unforgeable signature scheme for $\Phi$ with signature size $k + \ell$. Moreover, the scheme has unique signatures with respect to $\Phi$.
\end{theorem}
\begin{proof}
\label{proof:rds-appdx}
Robustness of $\RDS$ follows immediately from the correctness of the underlying $\DSS$ and the direct-access robustness of $\PPH$. We focus on proving strong unforgeability. Assume for contradiction that there exists a PPT adversary $\cA$ that outputs a pair $(\zeta, \sigma)$ such that $\Phi(\zeta, y) = 0$ for all $(y, \sigma) \in \Trans(\cA^{\RSign})$ but the pair passes $\RVrfy$ with non-negligible probability. 

That is, let $(\RDS.\pk, \RDS.\sk) \leftarrow \RGen(1^\lambda)$, $(\zeta, \sigma) \leftarrow \cA^{\RSign_{\RDS.\sk}}(\RDS.\pk)$, and let $\Pi$ denote the transcript $\Trans(\cA^{\RSign})$, where each entry logs a tuple $(y, z, \tau)$, with $z = h(y)$ and $\tau = \DSS.\Sign(\DSS.\pk, z)$. Then there exists a polynomial $q: \NN \to \NN$ such that
\begin{align*}
\Pr\left[
    \begin{array}{rl}
      \RVrfy(\RDS.\pk, \zeta, \sigma) = 1 \\\wedge \; \forall (y, z, \tau) \in \Pi, \Phi(\zeta, y) = 0
    \end{array}
  \right] > 1/q(\lambda)\;.
\end{align*}
Let $E_{\DSS}$ denote the event that $\sigma = (z, \tau)$ constitutes a successful forgery against $\DSS$ , and let $E_{\PPH}$ denote the event that $(\zeta, z)$ violates the direct-access robustness of $\PPH$. More precisely,
\begin{align*}
    E_\DSS &= \DSS.\Ver(\pk, z, \tau) = 1\; \wedge\; (\cdot, z, \tau) \not\in \Pi\;, \\
    E_\PPH &=\Eval(h,h(\zeta),z) = 1\; \wedge\; \bigvee_{(y,z,\tau) \in \Pi} \Phi(\zeta, y) = 0\;,
\end{align*}
where $(\cdot, z, \tau) \not\in \Pi$ is shorthand for $\not \exists y \text{ such that } (y, z, \tau) \in \Pi$. Since \[\RVrfy(\RDS.\pk,\zeta, \sigma) = \DSS.\Ver(\DSS.\pk, z, \tau) \wedge \Eval(h,h(\zeta),z)\;,\]
we observe that
\begin{align*}
    \RVrfy(\RDS.\pk, \zeta, \sigma) = 1\; \wedge\; (\forall (y, \sigma)& \in \Pi, \Phi(\zeta, y) = 0) \subseteq E_\DSS \vee (E_\PPH \wedge (\cdot, z, \tau) \in \Pi)\;.
\end{align*}
It follows that $\Pr\left[
    E_\DSS \vee (E_\PPH \wedge (\cdot, z, \tau) \in \Pi)
  \right] > 1/q(\lambda)$. This leads to a contradiction, since both $E_{\DSS}$ and $E_{\PPH} \wedge (\cdot, z, \tau) \in \Pi$ should occur with negligible probability due to the the strong unforgeability of $\DSS$, and the direct-access robustness of $\PPH$ respectively. Therefore, the robust signature scheme for $\Phi$ presented in Figure~\ref{alg:rds-construct} is strongly unforgeable.

\paragraph{Unique signatures. } Suppose, for the sake of contradiction, that there exists an adversary $\mathcal{A}$ and a polynomial $q : \mathbb{N} \to \mathbb{N}$ such that for  $(\pk,\sk) \leftarrow \RGen(1^\lambda)$ and $(\zeta, y, y') \leftarrow \cA^{\RSign_{\sk}}(\pk)$\ ,

\begin{align*}
    \Pr\left[ 
    \begin{array}{rl}
    \;
    \exists\sigma \text{ such that } (y, \sigma) \in \Pi \text{ and } (y', \sigma) \in \Pi \\
    \wedge \; \Phi(\zeta, y) = 1 \; \wedge \; \Phi(\zeta, y') = 0
     \end{array}
     \right] > 1/q(\lambda)\;,
\end{align*}
Then $\mathcal{A}$ finds $y$ and $y'$ such that $h(y) = h(y')$, where $h$ is the PPH hash function, with probability at least $1/q(\lambda)$, because their signatures are identical. However, this implies that either the pair $(\zeta, y)$ or $(\zeta, y')$ violates the robustness of $\PPH$. Specifically, the evaluation algorithm $\Eval$ cannot distinguish $h(y)$ from $h(y')$, and hence $\PPH.\Eval(h, h(y), h(\zeta)) \text{ and }  \PPH.\Eval(h, h(y'), h(\zeta))$ will have the same output distribution, even though $\Phi(\zeta, y) \neq \Phi(\zeta, y')$. Therefore, at least one of these pairs must cause $\Eval$ to output the incorrect predicate value with inverse-polynomial probability, contradicting the robustness of $\PPH$.
\end{proof}

\subsection{Recoverable signatures from difference-recovering PPH}
\label{sec:recoverability}
We show that our generic construction from~\refThm{thm:rds-appdx-main} is recoverable, provided the underlying $\PPH$ is difference recovering. Let $(\RGen, \RSign, \Recover)$ denote the scheme from Figure~\ref{alg:rds-construct}, where the verification algorithm $\RVer$ is replaced by the recovery algorithm $\Recover$ presented in Algorithm~\ref{alg:rds-rec}. The guarantees of this scheme are formally stated in~\refCor{cor:ham-pph-to-recoverable}.
\begin{algorithm}[H]
    \caption{$\Recover$}
    \label{alg:rds-rec}
    \begin{algorithmic}[1]
        \STATE  \textbf{input}: $(\pk =(h, \DSS.\pk), \zeta, \sigma' = (z, \tau))$

        \IF{$\RVer(\pk, \zeta, \sigma) = 0$}
            \STATE return $\perp$
        \ENDIF
        \STATE $\Delta \leftarrow \PPH.\Recover(h, z, \zeta)$
        \STATE return $\zeta \oplus \Delta$
        \end{algorithmic}
\end{algorithm}

\begin{corollary}[Recoverable signatures]
\label{cor:ham-pph-to-recoverable}
Let $\Phi$ be any predicate. If $\PPH[\Sample, \Hash, \Recover]$ is difference recovering w.r.t. $\Phi$ (as in Definition \ref{def:PPH-recoverability}) and $\DSS[\Gen, \Sign, \Vrfy]$ is a strongly unforgeable digital signature scheme, then the resulting $(\RGen, \RSign, \Recover)$ is a recoverable digital signature scheme w.r.t.~$\Phi$.
\end{corollary}

\begin{remark}[Recoverable signatures without PPH]
Note that Corollary \ref{cor:ham-pph-to-recoverable} can be instantiated using \emph{any} difference recovering hash family, not necessarily a PPH. Although $\RVer$, invoked by $\Rec$, leverages PPH evaluation $\Eval(h, h(\zeta), z)$, this call can instead be replaced by computing $\Phi(\zeta, \PPH.\Rec(h, z, \zeta) \oplus \zeta)$. We present the construction using $\Eval$ only because PPHs are widely studied.
\end{remark}

\begin{proof}
Let $(\pk = (h, \DSS.\pk), \sk) \leftarrow \RGen(1^\lambda)$ and let 
$\zeta \leftarrow \cA^{\Sign_{\sk}}(\pk)$ be the adversary's output. 

\paragraph{Robust recovery.}  
Suppose there exists $y$ such that $(y, \sigma) \in \Pi$ and $\Phi(\zeta, y) = 1$. By robustness of $\RDS$, $\RVer(\pk, \zeta, \sigma) = 1$ with probability at least $1 - \negl(\lambda)$.
Conditioned on this, $\sigma = (h(y), \cdot)$, and the difference-recovery of $\PPH$ satisfies 
$\PPH.\Recover(h, h(y), \zeta) = y \oplus \zeta$ with probability at least $1 - \negl(\lambda)$, so the final output is $y$ with overwhelming probability.

\paragraph{Unforgeable recovery.}  
Suppose $\Rec(\pk, \zeta, \sigma) = y_{\rec}$, with $y_{\rec} \neq \perp$. Then by definition of $\Rec$, $\RVer(\pk, \zeta, \sigma) = 1$.  
By strong unforgeability of $\RDS$, with probability at least $1-\negl(\lambda)$ there exists $y$ such that $(y, \sigma) \in \Pi$ and $\Phi(\zeta, y) = 1$, meaning $\sigma = (h(y), \cdot)$ corresponds to a valid signature of $y$.  

Conditioned on the above, difference-recovery of $\PPH$ yields 
$\PPH.\Recover(h, h(y),$ $\zeta)$  $= y \oplus \zeta$ with probability at least $1 - \negl(\lambda)$, meaning
\[
    y_{\rec} = \zeta \oplus \PPH.\Recover(h, h(y), \zeta) = \zeta \oplus (y \oplus \zeta) = y,
\] 
which implies that $y_{\rec}$ satisfies $(y_{\rec}, \sigma) \in \Pi$ and $\Phi(\zeta, y_{\rec}) = 1$. Since the intersection of both events occurs with overwhelming probability, the claim follows.
\end{proof}

\subsection{Instantiating Hamming signatures}
\label{sec:rds-instantiation}
We instantiate a recoverable digital signature scheme for Hamming predicates, providing concrete trade-offs between the signature size $k$ (relative to the message length $n$) and the tolerable Hamming error $r$. Since the regime of practical interest in watermarking corresponds to a constant error rate $r/n = \Omega(1)$, we focus on the rates achievable in this setting.

Here, we refer to the signature-to-message rate, $k/n$ of a RDS with signature size $k$, message size $n$. 

\begin{corollary}[Recoverable digital signatures with constant error-rate]
\label{cor:recoverable-signature-concrete}
Let $\lambda \in \NN$ be the security parameter, and suppose there exists a CRHF with output size $\ell = \ell(\lambda)$ and a strongly unforgeable digital signature scheme with signature size $k = k(\lambda)$. Let $n = \poly(\lambda)$ such that $(\ell + k)/n = o(1)$ be the message size. For any constants $\delta \in (0, 1/2)$ and $\rho > H_2^{\mathsf{BZ}}(\delta)$, where $H_2^{\mathsf{BZ}}(\cdot)$ denotes the Blokh--Zyablov bound, there exists a recoverable digital signature scheme $(\Gen, \Sign, \Recover)$ for the predicate $\Ham_{\delta}$ over $\{0,1\}^n$ with signature-to-message rate $\rho + o(1)$. Moreover, the scheme has unique signatures with respect to $\Ham_\delta$.
\end{corollary}

\begin{proof}
We instantiate Corollary \ref{cor:ham-pph-to-recoverable} using the difference recovering hash family from~\refThm{thm:sharp-sketch-hamming} for the predicate $\Ham_{\delta}$. The hash size is $\rho n + \ell$, and signing the hash incurs an additional $k$ bits. Thus, the total signature size is $\rho n + \ell + k$, and the resulting recoverable signature-to-message rate is $(\rho n + \ell + k)/n = \rho + o(1)$.
\end{proof}

\begin{remark}
We instantiate our construction using a sharp sketch (Theorem~\ref{thm:sharp-sketch-hamming}) rather than a PPH. Note that we could alternatively use the PPH from Theorem~\ref{thm:pph-hamming}, with recoverability following from Proposition~\ref{prop:pph-recover-app}. However, using a secure sketch offers the advantage that we can achieve any constant $\rho > H_2^{\mathsf{BZ}}(\delta)$ (rather than $\rho > H_3^{\mathsf{BZ}}(\delta)\log_23$), and that we can rely only on CRHFs instead of homomorphic CRHFs.\end{remark}

We can also instantiate signatures that achieve a subconstant error rate, by using a sharp sketch and analysis along the lines of~\cite[Corollary 6.8]{holmgren2022nearly}.  
\begin{corollary}[Recoverable digital signatures with subconstant error-rate]
\label{cor:RDS-subconstant}
Let $\lambda \in \NN$ be the security parameter, and suppose there exists a CRHF with output size $\ell = \ell(\lambda)$ and a strongly unforgeable digital signature scheme with signature size $k = k(\lambda)$. Let $n = \poly(\lambda)$ be the message size. For any $\delta = o(1)$, there exists a recoverable digital signature scheme $(\Gen, \Sign, \Recover)$ for the predicate $\Ham_{\delta}$ with signature size $O(\delta n \log n) + \ell + k$. Moreover, the scheme has unique signatures with respect to $\Ham_\delta$.
\end{corollary}

%% file: sections/watermarking-formal.tex
\section{Recoverable, Robust, and Unforgeable Watermarking}

In Section~\ref{sec:watermark-main}, we present our main construction, which combines block steganography with robust signatures to obtain a robust and unforgeable watermarking scheme. We then demonstrate, in Section~\ref{sec:watermark-chain}, that our verification algorithm can be generalized by chaining across longer substrings, yielding fine-grained control over the tradeoff in substring size between robustness and unforgeability. In Section~\ref{sec:watermark-recoverable}, we prove that if the underlying signature scheme is recoverable, we can build a \emph{recoverable} watermarking scheme. Finally, in Section~\ref{sec:watermark-instantiation}, we instantiate both secret and public key versions of our construction using existing block steganography constructions \cite{christ2024pseudorandom,alrabiah2024ideal,fairoze2023publicly} together with our robust and recoverable digital signature instantiations (Corollaries~\ref{cor:recoverable-signature-concrete} and~\ref{cor:RDS-subconstant}).
\subsection{Robust and unforgeable watermarking}
\label{sec:watermark-main}
Theorem~\ref{thm:watermark-main} summarizes the guarantees of our watermarking scheme. It has the following properties
\begin{itemize}
    \item \((\EBC{\Phi}{n}, 2n)\)-robustness: an (adversarially chosen) $\zeta$ verifies if it contains a substring $\zeta^* \preceq \zeta$ of size $2n$ whose two blocks $(\zeta_1^*, \zeta_2^*)$ are $\Phi$-close to two consecutive response blocks $(y_t, y_{t+1})$.
    \item \((\EBC{\Phi}{n}, n)\)-unforgeability: whenever $\zeta$ verifies, there exists an $n$-sized substring $\zeta^* \preceq \zeta$ that is $\Phi$-close to some response block $y_t$.
\end{itemize}
The parameters \(\Phi\) and \(n\) are inherited from the underlying block steganography and robust signature schemes. Whether robustness is secret key or public key depends on the block steganography scheme. In contrast, unforgeability is inherited entirely from the robust signature and is therefore always public key. Compared to the basic watermarking scheme obtained directly from block steganography (Fact~\ref{fact:full-wm-from-block}), our construction yields slightly weaker robustness---requiring a \(2n\)-sized close substring to verify rather than an \(n\)-sized substring---but additionally provides unforgeability (and, as we show in Section \ref{sec:watermark-recoverable}, recoverability).
\begin{theorem}[Robust and unforgeable 2-block watermarking]
\label{thm:watermark-main}
Let $\lambda \in \mathbb{N}$ be the security parameter, and let $k, n = \poly(\lambda)$ be any message and block size, respectively. Let $\Phi$ be any closeness predicate, and $Q$ be any language model. Suppose that:
\begin{itemize}
    \item $\BWM[\Gen, \Embed, \Dec]$ is a message size $k$, block size $n$ block steganography scheme for $Q$ with (secret or public key) $\Phi$-robustness (Definition \ref{def:block-steg}).
    \item $\RDS[\Gen, \Sign, \Ver]$ is a message size $n$, signature size $k$ $\Phi$-robust digital signature scheme (Definition \ref{def:rds-formal}).
\end{itemize}
Then the construction $(\Gen, \Wat, \Ver)$ in Figure~\ref{alg:wm-construct} is a watermarking scheme for $Q$ with (secret or public key) $(\EBC{\Phi}{n}, 2n)$-robustness, and public key $(\EBC{\Phi}{n}, n)$-unforgeability, where $\EBC{\Phi}{n}$ is defined as in Definition \ref{def:every-block-close}.
\end{theorem}
We present the full details of our watermarking scheme $(\Gen, \Wat, \Ver)$, then prove the theorem. The construction is the same whether $\BWM$ is public or secret key, but its guarantees change. The $\Gen$ algorithm generates the watermarking and verification keys. 
\begin{algorithm}[H]
\caption{$\Gen$}
\label{alg:setup}
    \begin{algorithmic}[1]
        \STATE  \textbf{input}: $(1^\lambda)$
        \STATE $\RDS.\pk, \RDS.\sk \leftarrow \RDS.\Gen(1^\lambda)$
        \STATE $\BWM.\vk, \BWM.\sk \leftarrow \BWM.\Gen(1^\lambda)$
        \STATE $\vk = (\RDS.\pk, \BWM.\vk),\; \sk = (\RDS.\sk, \BWM.\sk)$
        \STATE return $(\vk, \sk)$
        \end{algorithmic}
\end{algorithm}
Given a prompt $x \in \{0, 1\}^*$, the watermarking algorithm $\Wat$ generates a response by sampling text blocks of size $n=n(\lambda)$ sequentially and embedding the robust signature associated with each block into the next. 
\begin{algorithm}[H]
    \caption{$\Wat$}
    \label{alg:wat}
    \begin{algorithmic}[1]
        \STATE  \textbf{input}: ($1^\lambda$, $\sk$, $x$)
        \STATE $k \gets k(\lambda)$ \hfill \texttt{// compute signature size k}
        \STATE $\sigma_1 \gets \{0,1\}^k$
        \STATE $y_1 \leftarrow \BWM.\Embed(\BWM.\sk, x, \sigma_1)$ \hfill \texttt{// sample initial block with random message}
        \STATE $t \gets 1$
        \WHILE{$\done \not \preceq y_t$} 
            \STATE $t \gets t + 1$
            \STATE $\sigma_{t-1} \leftarrow \RDS.\Sign(\RDS.\sk, y_{t-1})$ \hfill \texttt{// generate signature of previous block}
            \STATE $y_t \leftarrow \BWM.\Embed(\BWM.\sk, (x,y_1, \ldots, y_{t-1}), \sigma_{t-1})$ \hfill \texttt{// embed signature into $y_t$}
            \ENDWHILE
        \STATE return $(y_1, \ldots, y_{t-1}, y_t[<\done])$ \hfill \texttt{// return the generated $y$ up to the $\done$ token}
        \end{algorithmic}
\end{algorithm}
The verification algorithm $\Ver$ operates on arbitrary sized strings, by using $\BWM$ to recover embedded signatures and checking whether \emph{any} $2n$-sized substring in the input corresponds to a valid message-signature pair. 
\begin{algorithm}[H]
    \caption{$\Ver$}\label{alg:ver}
    \begin{algorithmic}[1]
        \STATE  \textbf{input}: $(1^\lambda, \vk, \zeta)$
        \STATE $n \gets n(\lambda)$  \hfill \texttt{// compute block size n}
        \FOR{$i \in [|\zeta|  - 2n + 1]$}
            \STATE $\sigma_i \gets \BWM.\Dec(\BWM.\vk, \zeta[i+n:i+2n])$ \hfill \texttt{// attempt to decode the signature}
            \IF{$\RDS.\Ver(\RDS.\pk, \zeta[i: i+n], \sigma_{i}) = 1$} 
                \STATE return $1$ \hfill \texttt{// output $1$ if the signature verifies with the first block}
            \ENDIF
        \ENDFOR
        \STATE return $0$
        \end{algorithmic}
\end{algorithm}
{\captionof{figure}{\label{alg:wm-construct} Pseudocode for watermarking scheme $(\Gen, \Wat, \Ver)$.}}
\begin{proof}
Undetectability follows directly from the undetectability of $\BWM$. For robustness and unforgeability, fix $\lambda \in \NN$ and a PPT adversary $\cA$. Let 
\[(\RDS.\pk, \RDS.\sk) \gets \RDS.\Gen(1^\lambda), \quad (\BWM.\vk, \BWM.\sk) \gets \BWM.\Gen(1^\lambda)\;.\] 
Denote $\vk = (\RDS.\pk,$ $\BWM.\vk)$ and $\sk = (\RDS.\sk, \BWM.\sk)$. Let $\zeta \gets \mathcal{A}^{\Wat_{\sk}}(\vk)$.  
\paragraph{Robustness.}  
We prove the public key version; the secret-key version follows by an analogous argument. Suppose there exists a substring $\zeta^* \preceq \zeta$ with $|\zeta^*| = 2n$ and a prompt-response pair $(x, y) \in \Trans(\mathcal{A}^{\Wat_{\sk}})$ such  that $\zeta^*$ is $\EBC{\Phi}{n}$-close to $y$.  We will show $\Ver(\vk, \zeta) = 1$ with probability at least $1 - \negl(\lambda)$. Denote $(\zeta^*_1, \zeta^*_2) = \Blocks(\zeta^*)$. The closeness condition says that there exists a block index $t \in \NN$ such that
\[\Phi(\zeta^*_1, y_t) = 1 \text{ and } \Phi(\zeta^*_{2}, y_{t+1}) = 1\;.\]
Moreover, because $y_{t+1}$ is an $n$-sized block returned by $\Wat_\sk$, and it isn't the first block $y_1$, it has the signature of the previous block $y_t$ embedded in it. That is,
\[\Big(\big(\pi = (x, y_1, \ldots, y_t), m =\sigma_t\big), y_{t+1}\Big) \in \Trans(\cA^{\BWM.\Embed_{\BWM.\sk}})\]
where $\sigma_t$ is the $\RDS$ signature of $y_t$ computed during an oracle call to $\Wat_{\sk}(x)$.

The closeness of the second half $\Phi(\zeta^*_{2}, y_{t+1}) = 1$, together with the block robustness of $\BWM$, ensures that $\BWM.\Dec(\BWM.\vk, \zeta^*_2) = \sigma_t$ with probability at least $1-\negl(\lambda)$. Conditioned on this event, the closeness of the first half $\Phi(\zeta^*_1, y_t) = 1$, together with the robustness of $\RDS$, ensures that $\RDS.\Ver(\vk, \zeta^*_1, \sigma_t) = 1$ with probability at least $1 - \negl(\lambda)$.

If both events occur, then $\Ver(\vk, \zeta) = 1$, since it accepts if \emph{any} $2n$-substring of $\zeta$ verifies under this two-part check. Because both events occur with all but negligible probability, successful verification also occurs with all but negligible probability.
\paragraph{Unforgeability.}  
Let $\Pi$ denote the transcript $\Trans(\cA^{\Wat_{\sk}})$. Suppose, towards contradiction, that there exists a polynomial $q : \NN \rightarrow \NN$ such that
\begin{align*}
    \Pr\left[
    \begin{array}{rl}
    \Ver(\vk, \zeta) = 1 \\
    \wedge \; \neg \big(\exists \zeta^* \preceq \zeta, (x, y) \in \Pi \text{ such that } \\
    \EBC{\Phi}{n}(\zeta^*, y) = 1\;
    \wedge \; |\zeta^*| = n \big)
    \end{array}
    \right] \ge 1/q(\lambda)\;,
\end{align*}
When the above event occurs, then for every substring $\zeta^* \preceq \zeta$ of size $n$, and for every $(x, y) \in \Trans(\mathcal{A}^{\Wat_\sk})$, $\zeta^*$ is $\EBC{\Phi}{n}$-far from $y$, which ensures that for every block index $t \in [\lfloor |y|/n \rfloor]$, $\Phi(\zeta^*, y_t) = 0$.

Note that the set of all such blocks $y_t$ subsumes the set of message for which $\cA$ observed robust signatures, because the only messages that $\cA$ sees signatures for during interaction with $\Wat_{\sk}$ are response blocks $y_t$. This means that for all $\zeta^* \preceq \zeta$ such that $|\zeta^*| = n$, and for all $(y, \cdot) \in \Trans(\cA^{\RDS.\Sign_{\RDS.\sk}}),\; \Phi(\zeta^*, y) = 0.$

However, under the same event, $\Ver(\vk, \zeta) = 1$, meaning $\cA$ found some $\zeta^* \preceq \zeta$ of size $n$ and some $\sigma \in \{0, 1\}^k$ such that $\RDS.\Ver(\RDS.\pk, \zeta^*, \sigma) = 1$ with probability at least $1/q(\lambda)$. These two facts combined contradict the unforgeability of $\RDS$.

\end{proof}

\subsection{Generalizing verification to chains}
\label{sec:watermark-chain}
The verification algorithm for our watermarking scheme (Figure~\ref{alg:wm-construct}) applies a sliding window to $\zeta$, checking whether any substring of size $2n$ contains a valid message–signature pair. As a result, Theorem \ref{sec:watermark-main} guarantees unforgeability with respect to substrings of size $n$ (the size of the signed message).

We can generalize this approach by using a sliding window of size $rn$ for any $r \in \NN$ with $r \ge 2$. Within each $rn$-bit window, we can verify that the \emph{entire chain} of $r$, $n$-bit blocks verifies back-to-back. This yields a stronger unforgeability guarantee: successful verification implies the existence of an $(r-1)n$-bit substring in $\zeta$ that is close to a consecutive \emph{chain} of blocks, $(y_t, \ldots, y_{t+r-2})$, present in a single observed response. We formally define a family of general verification algorithms, $\{\Ver_r\}_{r \in \NN}$, in Figure \ref{alg:chain-ver}.
\begin{algorithm}[H]
    \caption{$\Ver_r$}\label{alg:ver_r}
    \begin{algorithmic}[1]
        \STATE  \textbf{input}: $(\vk, \zeta)$
        \STATE $n \gets n(\lambda)$  \hfill \texttt{// compute block size n}
        \FOR{$i \in [|\zeta|  - rn + 1]$}
            \STATE $b \leftarrow \ChainVer_r(\vk, \zeta[i:i+rn])$ \hfill \texttt{// verify the entire chain}
            \IF{$b = 1$}
                \STATE return $1$
            \ENDIF
        \ENDFOR
        \STATE return $0$
        \end{algorithmic}
\end{algorithm}
Chain verification takes $r$ blocks of size $n$ as input and verifies that for each consecutive pair of blocks $(\zeta_j^*, \zeta_{j+1}^*)$, the signature recovered from the second block successfully verifies with the message given by the first. If any block fails verification, the entire chain is rejected.
\begin{algorithm}[H]
    \caption{$\ChainVer_r$}\label{alg:chainver_r}
    \begin{algorithmic}[1]
        \STATE  \textbf{input}: $(\vk, \zeta^* = (\zeta_1^*, \ldots, \zeta_r^*) \in \{0, 1\}^{rn})$
        \STATE $\mathsf{flag} \gets 1$ 
        \FOR{$j \in [r-1]$}
            \STATE $\sigma_{j} \gets \BWM.\Dec(\BWM.\vk, \zeta^*_{j+1})$ \hfill \texttt{// attempt to decode signature}
            \IF{$\RDS.\Ver(\RDS.\pk, \zeta^*_j, \sigma_{j}) = 0$}
                \STATE $\mathsf{flag} \gets 0$ \hfill \texttt{// if any pair fails verification, set $\mathsf{flag}$ to $0$}
            \ENDIF
        \ENDFOR
        \STATE return $\mathsf{flag}$ \hfill \texttt{// indicates whether or not \emph{every} block verified}
        \end{algorithmic}
\end{algorithm}
{\captionof{figure}{\label{alg:chain-ver} Pseudocode for chain verification.}}
\vspace{3.0mm}
Theorem~\ref{thm:r-block-watermark-full} states the guarantees obtained when instead using $\Ver_r$ as the verifier for our construction. Its guarantees differ from those obtained using 2-block verification in two respects. First, chain verification requires that $\RDS$ has \emph{unique signatures} ( Definition \ref{def:unique-signatures}), though, as noted previously, this is not particularly restrictive. Second, unlike in the 2-block case, whether unforgeability holds in the public or secret key setting additionally depends on the robustness guarantees of the underlying block steganography scheme $\BWM$.
\begin{theorem}[Robust and unforgeable $r$-block watermarking]
\label{thm:r-block-watermark-full}
Let $\lambda\in \NN$ be the security parameter, let $r = \poly(\lambda)$ satisfying $r \geq 2$ be any chain size, and let $k, n = \poly(\lambda)$ be any message and block size, respectively. Let $\Phi$ be any closeness predicate, and $Q$ be any language model. Suppose that:
\begin{itemize}
    \item $\BWM[\Gen, \Embed, \Dec]$ is a message size $k$, block size $n$ block steganography scheme for $Q$ with (secret or public key) $\Phi$-robustness (Definition \ref{def:block-steg}).
    \item $\RDS[\Gen, \Sign, \Ver]$ is a message size $n$, signature size $k$ $\Phi$-robust digital signature scheme (Definition \ref{def:rds-formal}) that has unique signatures w.r.t. $\Phi$ (Definition \ref{def:unique-signatures}).
\end{itemize}
Then the construction $(\Gen, \Wat, \Ver_r)$ (Figures \ref{alg:wm-construct} and \ref{alg:chain-ver}) is a watermarking scheme for $Q$ with (secret or public key) $(\EBC{\Phi}{n}, rn)$-robustness, and (secret or public key) $(\EBC{\Phi}{n}, (r-1)n)$-unforgeability. 
\end{theorem}

\begin{proof} Undetectability and robustness are straightforward from the robustness and undetectability of $\BWM$ and the robustness of $\RDS$. Therefore, we will focus on proving the unforgeability guarantee. We will prove secret key unforgeability; public key follows analogously.\\\\ 
Fix $\lambda \in \NN$ and a PPT adversary $\cA$. Let 
\[
(\RDS.\pk, \RDS.\sk) \gets \RDS.\Gen(1^\lambda), \quad (\BWM.\vk, \BWM.\sk) \gets \BWM.\Gen(1^\lambda)\;.
\] 
Denote $\vk = (\RDS.\pk, \BWM.\vk)$ and $\sk = (\RDS.\sk, \BWM.\sk)$. Let $\Pi$ denote the transcript $\Trans(\cA^{\Wat_{\sk}})$ and let $\zeta \gets \mathcal{A}^{\Wat_{\sk}, (\Ver_r)_\vk}(1^\lambda)$ 
\\\\
We begin by showing that the event that a string $\zeta$ is $\EBC{\Phi}{n}$-far from every response $y$ implies a more convenient event. For a predicate $\Phi$ and an index $i \in \NN$, define

\begin{equation}
\label{def:suffix-forms-chain}
\SuffixFormsChain_{\Phi,i}(\zeta,\Pi)
:= 
\left(
\begin{aligned}
&\exists\ell \in \NN \text{ such that } |\zeta| = \ell n \text{ and } \\
&\Big(i > \ell \text{ or }\exists(x,y) \in \Pi \text{ such that }\\
&\EBC{\Phi}{n}\big((\zeta_i,\ldots,\zeta_\ell),\, y\big) = 1\big)
\Big)\end{aligned}
\right).
\end{equation}
Given a sequence $\zeta = (\zeta_1,\ldots,\zeta_\ell)$ of $\ell$ blocks of size $n$, the event asserts that the suffix starting at the $i$th block, $(\zeta_i,\ldots,\zeta_\ell)$ is $\EBC{\Phi}{n}$-close to some response $y$ appearing in the transcript~$\Pi$. When $i > \ell$, we interpret $\SuffixFormsChain_{\Phi, i}(\zeta,$ $\Pi) = 1$.\\\\
Now notice that if a string $\zeta = (\zeta_1, \dots, \zeta_\ell)$ is
$\EBC{\Phi}{n}$-\emph{far} from all responses $y$, then there must be a point
at which the blocks of $\zeta$ transition from being close to a response chain to not. Formally, there
exists an index $i \in [\ell]$ such that the suffix $(\zeta_{i+1}, \dots, \zeta_\ell)$
is close to some response, but $(\zeta_i, \dots, \zeta_\ell)$ is not:
\begin{align*}
    &\Big(\forall (x, y) \in \Pi, \EBC{\Phi}{n}(\zeta, y) = 0\Big)
    \\ &\qquad \subseteq 
    \bigvee_{i \in [\ell]}
    \Big(
        \SuffixFormsChain_{\Phi, i+1}(\zeta, \Pi)
        \;\wedge\;
        \neg \SuffixFormsChain_{\Phi, i}(\zeta, \Pi)
    \Big)\;.
\end{align*}
This also covers the case when \emph{no} suffix of
$\zeta$ is close to a response chain. In particular, if 
no suffix of $\zeta$ is $\EBC{\Phi}{n}$-close to a response $y$, then for
all $i \in [\ell]$ we have
$\SuffixFormsChain_{\Phi,i}(\zeta,\Pi) = 0$, so
$\neg \SuffixFormsChain_{\Phi,i}(\zeta,\Pi) = 1$. Hence, the
right-hand side still evaluates to $1$ due to the $i=\ell$ clause, where
$\SuffixFormsChain_{\Phi,\ell+1}(\zeta,\Pi) = 1$ is trivially true.\\\\
Now consider $\Forgery_i$ for $i \in \NN$, the event that there exists a substring $\zeta^* \preceq \zeta$ such that the closeness chain of all-but-the-last block of $\zeta^*$ breaks exactly at index $i$, yet $\ChainVer_r$ still accepts it. Formally,
\begin{equation}
\label{event:forgery}
\Forgery_i(\zeta, \Pi, \vk) :=\left(
\begin{aligned}
&\exists\, \zeta^* \preceq \zeta \text{ of length } rn \text{ such that}\\
    &\ChainVer_r(\vk, \zeta^*) = 1, \\
    & \text{and }
    \SuffixFormsChain_{\Phi, i+1}((\zeta_1^*, \ldots, \zeta_{r-1}^*), \Pi), \\
    & \text{and } 
    \neg \SuffixFormsChain_{\Phi, i}((\zeta_1^*, \ldots, \zeta_{r-1}^*), \Pi).
\end{aligned}
\right).
\end{equation}
Lemma~\ref{lemma:bad-uf} shows that for any fixed $i \in [r-1]$, this event occurs with probability at most $\negl(\lambda)$.  
Applying a union bound over all $i \in [r-1]$ shows that the event $\bigvee_{i \in [r-1]}\Forgery_i(\zeta, \Pi, \vk)$ also occurs with negligible probability.\\\\
We conclude by showing that this suffices to prove unforgeability. Consider the event that an adversary successfully forges the watermark:
\[
\Ver_r(\vk, \zeta) = 1
\;\wedge\;
\Big(\forall\, \zeta' \preceq \zeta, (x, y) \in \Pi, \; |\zeta'| \neq (r-1)n \; \vee \; \EBC{\Phi}{n}(\zeta', y) = 0\Big).
\]
By definition of $\Ver_r$, if $\Ver_r(\vk, \zeta) = 1$, then there exists a
substring $\zeta^* \preceq \zeta$ of length $rn$ such that $\ChainVer_r(\vk, \zeta^*) = 1$.
Moreover, if every $(r-1)n$-sized
substring $\zeta' \preceq \zeta$ is $\EBC{\Phi}{n}$-far from all
responses, then the specific substring $(\zeta_1^*, \dots, \zeta_{r-1}^*) \preceq \zeta$, must be far from all responses. By the previous chain-breaking argument, this implies that there exists an index $i \in [r-1]$ where the chain breaks. Hence, the event of a successful forgery implies the event $\bigvee_{i \in [r-1]}\Forgery_i(\zeta, \Pi, \vk)
$, which we showed happens with negligible probability.
\end{proof}
\begin{lemma}
    \label{lemma:bad-uf}
    Let $\lambda \in \NN$ be a security parameter, let $r=\poly(\lambda)$ satisfying $r \geq 2$ be any chain size, and let $i \in [r-1]$ be any block index. Let $(\Gen, \Wat, \Ver_r)$ be the watermarking scheme presented in Figures \ref{alg:wm-construct} and \ref{alg:chain-ver}. Under the same assumptions as Theorem \ref{thm:r-block-watermark-full}, for any PPT adversary $\cA$, 
    \begin{align*}
    \Pr\left[
    \begin{array}{rl}
        (\pk, \sk) &\leftarrow \Gen(1^\lambda) \\ 
        \zeta &\leftarrow \cA^{\Wat_{\sk}, (\Ver_r)_\vk}
    \end{array} 
    :
    \begin{array}{rl}
       \Forgery_i(\zeta, \Pi, \vk)
    \end{array}
    \right] \leq \negl(\lambda)\;,
\end{align*}
where $\Pi$ is the transcript $\Trans(\cA^{\Wat_{\sk}})$ and $\Forgery_i(\zeta,\Pi, \vk)
$ is defined as in \ref{event:forgery}.    
\end{lemma}
\begin{proof}
Fix $\lambda \in \NN$. The case $i = r-1$ follows similarly to the proof of Theorem~\ref{thm:watermark-main}, therefore, we will prove the lemma for a fixed $i \in [r-2]$. Assume, toward contradiction, that there exists a polynomial $q : \mathbb{N} \to \mathbb{N}$ such that
\[
  \Pr[\Forgery_i(\zeta, \Pi, \vk)] \ge 1/q(\lambda)
\]
Notice that the event $\Forgery_i(\zeta, \Pi, \vk)$ implies the following event, specifically for the substring $(\zeta_i^*, \zeta_{i+1}^*)$:
\begin{equation*}
\left(
\begin{aligned}
&\exists\, \zeta^* \preceq \zeta \text{ of size } 2n \text{ such that }\RDS.\Ver\Big(\RDS.\pk, \zeta_1^*, \BWM.\Dec(\BWM.\vk, \zeta_{2}^*)\Big) = 1, \\
    & \text{and }
    \exists (x, y) \in \Pi, t \in \NN \text{ such that }\Phi(\zeta_2^*, y_t) = 1 \text{ and } 
   \big(\Phi(\zeta_1^*, y_{t-1}) = 0 \text{ or } t = 1\big)
\end{aligned}
\right).
\end{equation*}
Consequently, the event above occurs with probability at least $1/q(\lambda)$. Moreover, when this event occurs, the fact that $\Phi(\zeta_2^*, y_t) = 1$ and the block robustness of $\BWM$ tell us that $\BWM.\Dec(\BWM.\vk, \zeta_{2}^*)$ decodes to the message $\sigma_{t-1}$ embedded in $y_t$ with all but negligible probability. Therefore, with inverse-polynomial probability, $\cA$ finds a block $\zeta_1^* \in \{0, 1\}^n$ and a message $\sigma_{t-1} \in \{0, 1\}^k$ such that 
\[(\sigma_{t-1}, y_t) \in \Trans(\cA^{\BWM.\Embed_{\BWM.\sk}}) \text{ and } \RDS.\Ver(\RDS.\pk, \zeta_1^*, \sigma_{t-1}) = 1\;.\] 
By definition of $\Wat_{\sk}$, $\sigma_{t-1}$ is either a uniformly random string when $t = 1$, or the robust signature of the previous block $y_{t-1}$ when $t > 1$. We analyze these two cases separately.
\begin{enumerate}
    \item $\sigma_{t-1} \leftarrow \{0, 1\}^k$ is a uniformly random string, independent of $\zeta_1^*$.  
    We have that
    \[
      \RDS.\Ver(\RDS.\pk,\, \zeta_1^*,\, \sigma_{t-1}) = 1
    \]
    with inverse-polynomial probability. But then, there exists an adversary $\mathcal{B}^{\RDS.\Sign_{\sk}}(\RDS.\vk)$ against the strong unforgeability of $\RDS$ that simulates $\cA$ and outputs $\zeta_1^*$, along with a uniformly random string $\sigma' \leftarrow \{0, 1\}^k$. This verifies with inverse-polynomial probability, even though $\sigma'$ is likely not any of the signatures seen by $\mathcal{B}$.
    \item $\sigma_{t-1} \gets \RDS.\Sign(\RDS.\sk, y_{t-1})$ is a robust signature of the previous block.
    In this case, 
    \begin{equation}\label{eq:event1uf}
        (y_{t-1},\sigma_{t-1}) \in \Trans(\cA^{\RDS.\Sign_{\RDS.\pk}}).
    \end{equation}
    However, we have that $\Phi(\zeta_1^*, y_{t-1}) = 0$. Either $(\zeta_1^*, \sigma_{t-1})$ contradicts the unforgeability of $\RDS$, or there exists a different message $m \neq y_{t-1}$ such that 
    \[(m, \sigma_{t-1}) \in \Trans(\cA^{\RDS.\Sign_{\RDS.\pk}}) \text{ and } \Phi(\zeta_1^*, m) = 1\;.\]
    This, however, breaks the unique signature property of $\RDS$, because the adversary finds two different messages $m \ne y_{t-1}$ with the same
      signature $\sigma_{t-1}$, where $\zeta_1^*$ is $\Phi$-far from $y_{t-1}$ but
      $\Phi$-close to $m$.  Thus, in the case, $\cA$ either breaks unforgeability or unique signatures of $\RDS$.
\end{enumerate}
In both cases, we obtain a contradiction to the security of $\RDS$.  
Therefore, $\Forgery_i(\zeta, \Pi, \vk)$ occurs only with negligible probability.
\end{proof}

\subsection{Recoverable watermarking}
\label{sec:watermark-recoverable}
If the robust signature scheme 
$\RDS$ is recoverable, meaning it includes a recovery algorithm $\RDS.\Recover$ that can reconstruct the original message from its signature and a perturbed version of that message, then we can recover the original prefix of the errored text using $\Recover$. Notably, $\Recover$ does not rely on the underlying language model $Q$ or generation prompt $x$. Consequently, it enables users to recover original content using only the verification key and modified text. We obtain the following guarantees:
\begin{itemize}
    \item \((\EBC{\Phi}{n}, 2n, n)\)-robust recoverability: if an (adversarially chosen) $\zeta$ contains a substring $\zeta^* \preceq \zeta$ of size $rn$ for some $r \geq 2$, where the blocks $(\zeta^*_1, \ldots, \zeta_r^*)$ are blockwise $\Phi$-close to $r$ consecutive response blocks $(y_t, \ldots, y_{t+r-1})$, then $\Rec(\vk, \zeta)$ will recover a substring of that response, $\xi \preceq y$ of size $(r-1)n$. In particular, in our construction, $\xi = (y_t, \ldots, y_{t+r-2})$.
\item \((\EBC{\Phi}{n}, n)\)-unforgeable recoverability: if an (adversarially chosen) $\zeta$ leads to a recovery $\xi \in \Rec(\vk, \zeta)$ then $\xi$ must satisfy:
\begin{enumerate}
    \item \(\xi\) has size \(rn\) for some \(r \geq 1\),  
    \item \(\xi\) is a substring of some close response \(y\), i.e., \(\xi \preceq y\), and there exists \(\zeta^* \preceq \zeta\) of size $rn$ that is \(\EBC{\Phi}{n}\)-close to \(y\).  
\end{enumerate}
In particular, in our construction, \(\zeta^*\) is \(\EBC{\Phi}{n}\)-close to \(\xi\) and \(\xi\) forms chain of blocks in \(y\).
\end{itemize}
Our recovery algorithm leverages the chaining idea from the previous section  to not only reconstruct individual $n$-bit blocks using $\RDS.\Rec$, but to also recover multi-block prefixes when the adversary has approximately copied a \emph{chain} of consecutive watermarked blocks. The idea is to apply chain
recoverability---analogous to chain verification---with
\emph{every} window size $rn$, ensuring recovery from any long enough substring
$\zeta^* \preceq \zeta$ with $|\zeta^*| \ge 2n$.
\begin{algorithm}[H]
    \caption{$\Recover$}\label{alg:rec}
    \begin{algorithmic}[1]
        \STATE  \textbf{input}: $(\vk, \zeta)$
        \STATE $n \gets n(\lambda)$  \hfill \texttt{// compute block size n}
        \STATE $\mathcal{L} \gets \emptyset$
        \FOR{$r \in \{2, \ldots ,\floor{|\zeta|/n}\}$  \hfill \texttt{// for every chain size $r$} }
            \FOR{$i \in [|\zeta|  - rn + 1]$  \hfill \texttt{// for every $rn$-bit substring in the input}} 
                \STATE $\xi \leftarrow \ChainRec_r(\vk, \zeta[i:i + rn])$ \hfill \texttt{// attempt to recover the first $r-1$ blocks}
                \IF{$\perp \not \preceq \xi$ and $|\xi| = (r-1)n$}
                    \STATE $\mathcal{L} \leftarrow \mathcal{L} \cup \{\xi\}$ \hfill \texttt{// only keep $\xi$ if all $r-1$ recoveries were successful}
                \ENDIF
            \ENDFOR
        \ENDFOR
        \STATE return $\mathcal{L}$ \hfill 
        \end{algorithmic}
\end{algorithm}
\noindent The subroutine $\ChainRec_r$ takes in an $r$-block chain and attempts to recover the first $r-1$ blocks.
\begin{algorithm}[H]
    \caption{$\ChainRec_r$}\label{alg:recr}
    \begin{algorithmic}[1]
        \STATE  \textbf{input}: $(\vk, \zeta^*= (\zeta_1^*, \ldots, \zeta_r^*))$
        \STATE $n \gets n(\lambda)$  \hfill \texttt{// compute block size n}
        \FOR{$i \in [r-1]$}
            \STATE $\sigma_j \gets \BWM.\Dec(\BWM.\vk, \zeta^*_{j+1})$
            \STATE $\xi_{j} \gets \RDS.\Rec(\RDS.\pk, \zeta^*_j, \sigma_j)$
        \ENDFOR
        \STATE return $\xi = (\xi_1, \ldots, \xi_{r-1})$ \hfill \texttt{// may contain $\perp$ symbols}
        \end{algorithmic}
\end{algorithm}
{\captionof{figure}{\label{alg:wm-recovery} Pseudocode for the recovery algorithm $\Rec$, which calls upon $\ChainRec_r$ as a subroutine.}}
\begin{corollary}[Recoverable watermarking]
\label{cor:watermark-recoverable}
    Let $\lambda\in \NN$ be the security parameter and let $k, n = \poly(\lambda)$ be any message and block size, respectively. Let $\Phi$ be any closeness predicate, and $Q$ be any language model. Suppose that:
\begin{itemize}
    \item $\BWM[\Gen, \Embed, \Dec]$ is a message size $k$, block size $n$ block steganography scheme for $Q$ with (secret or public key) $\Phi$-robustness (Definition \ref{def:block-steg}).
    \item $\RDS[\Gen, \Sign, \Ver]$ is a message size $n$, signature size $k$ $\Phi$-recoverable digital signature scheme (Definition \ref{def:rds-recoverability-formal}) that has unique signatures w.r.t. $\Phi$ (Definition \ref{def:unique-signatures}).
\end{itemize}
Then the construction $(\Gen, \Wat, \Rec)$ in Figures~\ref{alg:wm-construct} and \ref{alg:wm-recovery} is a \emph{recoverable} watermarking scheme with (secret or public key) $(\EBC{\Phi}{n}, 2n, n)$-robustness and (secret or public key) $(\EBC{\Phi}{n}, n)$-unforgeability.
\end{corollary}
\begin{proof}
Fix $\lambda \in \NN$ and a PPT adversary $\cA$. Let 
\[
(\RDS.\pk, \RDS.\sk) \gets \RDS.\Gen(1^\lambda), \quad (\BWM.\vk, \BWM.\sk) \gets \BWM.\Gen(1^\lambda)\;.
\] 
Denote $\vk = (\RDS.\pk, \BWM.\vk)$ and $\sk = (\RDS.\sk, \BWM.\sk)$. Let $\zeta \gets \mathcal{A}^{\Wat_{\sk}}(\vk)$. We prove the public key version for both properties; the secret key version follow analogously. 
\paragraph{Robust recoverability.}  
Suppose there exists a substring $\zeta^* \preceq \zeta$ with $|\zeta^*| \ge 2n$ and a pair $(x, y) \in \Trans(\mathcal{A}^{\Wat_{\sk}})$ such that $\zeta^*$ and $y$ are close under $\EBC{\Phi}{n}$. Then there exist $r, t \in \mathbb{N}$, with $r \geq 2$, such that $|\zeta^*| = r n$ and
\[\Phi(\zeta^*_j, y_{t+j-1}) = 1, \; \forall j \in [r]\;.\]
Fix $j \in [r-1]$. As in the proof of robustness,
\[\sigma_{j} \leftarrow \BWM.\Dec(\BWM.\vk, \zeta_{j+1}^*)\]
is the signature of $y_{t+j-1}$ with probability at least $1 - \negl(\lambda)$. Since $\Phi(\zeta^*_j, y_{t+j-1}) = 1$, the recovery algorithm of $\RDS$ satisfies $\RDS.\Rec(\RDS.\pk, \zeta_j^*, \sigma_j) = y_{t+j-1}$ with probability at least $1 - \negl(\lambda)$ . Applying a union bound over all $j \in [r-1]$, we conclude that
\[
\ChainRec_r(\vk, \zeta^*) = (y_{t}, \ldots, y_{t+r-2}),
\]
which is a substring of $y$ of size $|\zeta^*| - n$, with all but negligible probability. Since $r \le \lfloor |\zeta|/n \rfloor$ (otherwise $r n > |\zeta|$ and $\zeta^*$ could not be a substring of $\zeta$), $\Rec(\vk, \zeta)$ calls $\ChainRec_r(\vk, \zeta^*)$ during its execution, so the chain $(y_{t}, \ldots, y_{t+r-2})$, which is an $(|\zeta^*| - n)$-sized substring of $y$, will be added to the final output list with all but negligible probability.

\paragraph{Unforgeable recoverability.} The proof follows similarly to the proof of unforgeable chain verification. Let $\Pi$ be the transcript $\Trans(\cA^{\Wat_{\sk}})$. For a predicate $\Phi$ and a block index $i \in \NN$, define 
\[
\SuffixFormsExactChain_{\Phi,i}(\xi, \Pi)
:= 
\left(
\begin{aligned}
&\exists(x,y) \in \Pi,t \in \NN \text{ such that } \\&
\xi_{j} = y_{t+j-i}, \; \forall j \in \{i, \ldots, \ell\}
\end{aligned}
\right).
\]
the event that an input string $\xi \in \{0, 1\}^*$ forms an \emph{exact} chain of blocks in some response $y$ (in contrast with $\SuffixFormsChain$ \ref{def:suffix-forms-chain} from the chain verification proof, which only required the input was \emph{close} to a chain of blocks in some response $y$). Moreover, define 
\[
\SuffixRecoversChain_{\Phi,i}(\zeta, \xi, \Pi)
:= 
\left(
\begin{aligned}
&\exists\, \ell \in \NN \text{ such that } |\zeta| = |\xi| =  \ell n \text{ and}\\
&\Big(i > \ell \text{ or}\\
&\big(\EBC{\Phi}{n}\big((\zeta_i,\ldots,\zeta_\ell),\, (\xi_i, \ldots, \xi_\ell)\big) = 1\\
&\text{and } \SuffixFormsExactChain_{\Phi, i}(\xi, \Pi)\big)\Big)
\end{aligned}
\right).
\]
Given two sequences $\zeta, \xi \in \{0, 1\}^{\ell n}$ of $\ell$ blocks of size $n$, this event asserts that the suffix starting from the $i$th block $(\zeta_i, \ldots, \zeta_\ell)$ is blockwise $\Phi$-close to the blocks $(\xi_i, \ldots, \xi_\ell)$, and that $\xi$ forms of chain of blocks in some response $y$. When $i=1$, this implies that $\xi \preceq y$ for some response $y$ and that $\zeta$ is $\EBC{\Phi}{n}$ close to $y$.\\\\
Now notice that if a string $\zeta = (\zeta_1, \ldots, \zeta_\ell)$ and recovery $\xi = (\xi_1, \ldots, \xi_\ell)$ do not form a recovery chain---that is, either $\xi$ does not form a chain of blocks in some response $y$, or $\zeta$ is $\EBC{\Phi}{n}$-far from $\xi$---then there must be some index $i \in [\ell]$ at which the chain transitions from being a valid recovery chain to not:
\begin{align*}
    &\Big(\EBC{\Phi}{n} (\zeta, \xi) = 0 \; \vee \; \neg \SuffixFormsExactChain_{\Phi, 1}(\xi, \Pi)\Big)
    \\ &\qquad\subseteq 
    \bigvee_{i \in [\ell]}
    \Big(
        \SuffixRecoversChain_{\Phi, i+1}(\zeta, \xi, \Pi)
        \;\wedge\;
        \neg \SuffixRecoversChain_{\Phi, i}(\zeta, \xi, \Pi)
    \Big)\;.
\end{align*}
Now, for $i, r \in \NN$, consider the event $\RecForgery_{i, r}$: there exists a substring $\zeta^* \preceq \zeta$ such that $\ChainRec_r(\vk, \zeta^*)$ successfully recovers some string $\xi$, but for which the recovery chain between $\zeta^*$ and $\xi$ breaks at index $i$. Formally, 
\begin{equation}
\label{event:rec-forgery}
\RecForgery_{i, r}(\zeta, \Pi, \vk) := \left(
\begin{aligned}
&\exists\, \zeta^* \preceq \zeta \text{ of length } rn \text{ such that }\\
    &\ChainRec_r(\vk, \zeta^*) = \xi \; \text{ with } \perp \not \preceq \xi, \\
    & \text{and }
    \SuffixRecoversChain_{\Phi, i+1}((\zeta_1^*, \ldots, \zeta_{r-1}^*), \xi, \Pi), \\
    & \text{and } 
    \neg \SuffixRecoversChain_{\Phi, i}((\zeta_1^*, \ldots, \zeta_{r-1}^*), \xi, \Pi)
\end{aligned}
\right).
\end{equation}
Lemma~\ref{lemma:bad-rec} shows that for any fixed $r=\poly(\lambda)$ with $r \geq 2$ and $i \in [r-1]$, this event occurs with probability at most $\negl(\lambda)$.  
Applying a union bound over all $r  \in \{2, \ldots, \floor{|\zeta|/n}\}, i \in [r-1]$ shows that the event $\bigvee_{i, r}\RecForgery_{i, r}(\zeta, \Pi, \vk)$ also occurs with negligible probability.\\\\
We conclude by showing that this suffices to prove unforgeable recoverability. Consider the event that an adversary successfully forges:
\[
\xi \in \Rec(\vk, \zeta)
\;\wedge\;
\Big(\forall\, \zeta' \preceq \zeta, (x, y) \in \Pi, \; |\zeta'| \neq |\xi| \;\vee\; \xi \not \preceq y \; \vee \; \EBC{\Phi}{n}(\zeta', y) = 0 \Big).
\]
By definition of $\Rec$, if $\xi \in \Rec(\vk, \zeta)$, then there exists an $r \in \{2, \ldots, \floor{|\zeta|/n}\}$ and substring $\zeta^* \preceq \zeta$ of size $rn$ where $\ChainVer_r(\vk, \zeta^*) = \xi$, $\perp \not \preceq \xi$, and $|\xi| = (r-1)n$. Moreover, if for every $|\xi|$-sized substring $\zeta' \preceq \zeta$ and every response $y$, either $\xi \not \preceq y$ or $\zeta'$ is $\EBC{\Phi}{n}$-far from $y$, then either $\zeta'$ is $\EBC{\Phi}{n}$-far from $\xi$ or $\xi$ does not form a chain in any response.

By the previous chain breaking argument, this means that the suffix recovery chain must break at some point for every substring $\zeta'$ of size $|\xi|$, including the specific substring $(\zeta_1^*, \ldots, \zeta_{r-1}^*)$. Hence, the event of a successful forgery implies the event $\bigvee_{i, r}\RecForgery_{i, r}(\zeta, \Pi, \vk)
$, which we showed happens with negligible probability.
\end{proof}

\begin{lemma}
    \label{lemma:bad-rec}
    Let $\lambda \in \NN$ be any security parameter, let $r = \poly(\lambda)$ satisfying $r \geq 2$ be any chain size, and let $i \in [r-1]$ be any block index. Let $(\Gen, \Wat, \Rec)$ be the watermarking scheme presented in Figures \ref{alg:wm-construct} and \ref{alg:wm-recovery}. Under the same assumptions as Corollary \ref{cor:watermark-recoverable}, for any PPT adversary $\cA$, 
    \begin{align*}
    \Pr\left[
    \begin{array}{rl}
        (\pk, \sk) &\leftarrow \Gen(1^\lambda) \\ 
        \zeta &\leftarrow \cA^{\Wat_{\sk}, \Rec_\vk}
    \end{array} 
    :
    \begin{array}{rl}
       \RecForgery_{i, r}(\zeta, \Pi, \vk)
    \end{array}
    \right] \leq \negl(\lambda)\;,
\end{align*}
where $\Pi$ is the transcript $\Trans(\cA^{\Wat_{\sk}})$ and $\RecForgery_{i, r}(\zeta,\Pi, \vk)
$ is defined as in \ref{event:rec-forgery}.    
\end{lemma}
\begin{proof}
Fix $\lambda \in \NN$. The case $i = r-1$ is similar to the proof of Theorem~\ref{thm:watermark-main}, but by using the unforgeable recoverability of $\RDS$ instead of standard unforgeability. Hence, we prove the lemma for a fixed $i \in [r-2]$. Assume, toward contradiction, that there exists a polynomial $q : \mathbb{N} \to \mathbb{N}$ such that
\[
  \Pr[\RecForgery_{i, r}(\zeta, \Pi, \vk)] \ge 1/q(\lambda)\;.
\]
Notice that this event implies the following event, specifically for the substring $(\zeta_i^*, \zeta_{i+1}^*)$:
\begin{equation}
\label{event:implies-rec}
\left(
\begin{aligned}
&\exists\, \zeta^* \preceq \zeta \text{ of size } 2n \text{ such that}\\
    &\RDS.\Rec\Big(\RDS.\pk, \zeta_1^*, \BWM.\Dec(\BWM.\vk, \zeta_{2}^*)\Big) = \xi \in \{0, 1\}^n, \\
    & \text{and }\exists (x, y) \in \Pi, t \in \NN \text{ such that}\\\ 
    &\Phi(\zeta_2^*, y_t) = 1 \text{ and }\Big(\Phi(\zeta_1^*, \xi) = 0  \text{ or } \xi \neq y_{t-1} \text{ or } t= 1\Big)
\end{aligned}
\right).
\end{equation}
Consequently, the event above occurs with probability at least $1/q(\lambda)$. Moreover, when this event occurs, the fact that $\Phi(\zeta_2^*, y_t) = 1$ and the block-robustness of $\BWM$ tell us that $\BWM.\Dec(\BWM.\vk, \zeta_{2}^*)$ decodes to the message encoded in $y_t$ with all but negligible probability. Therefore, with inverse-polynomial probability, $\cA$ finds a block $\zeta_1^* \in \{0, 1\}^n$ and a message $\sigma_{t-1} \in \{0, 1\}^k$ such that 
\[(\sigma_{t-1}, y_t) \in \Trans(\cA^{\BWM.\Embed_{\BWM.\sk}}) \text{ and } \RDS.\Rec(\RDS.\pk, \zeta_1^*, \sigma_{t-1}) = \xi \in \{0, 1\}^n\;.\] 
By definition of $\Wat_{\sk}$, the value $\sigma_{t-1}$ is either a uniformly random string when $t = 1$, or the robust signature of the previous block $y_{t-1}$ when $t > 1$. We analyze these two cases separately.
\begin{enumerate}
    \item $\sigma_{t-1} \leftarrow \{0,1\}^k$ is a uniformly random string. We have that
    \[
      \RDS.\Rec(\RDS.\pk,\, \zeta_1^*,\, \sigma_{t-1}) \neq \perp
    \]
    with inverse-polynomial probability. As in the proof of chain verification, there exists an adversary against the unforgeable recovery of $\RDS$.
    \item $\sigma_{t-1} \leftarrow \RDS.\Sign(\RDS.\pk, y_{t-1})$ is a robust signature of the previous block. In this case,
    \[(y_{t-1}, \sigma_{t-1}) \in \Trans(\cA^{\RDS.\Sign_{\RDS.\pk}})\;.\]
    Since $\RDS.\Rec(\RDS.\pk, \zeta_1^*, \sigma_{t-1}) = \xi \in \{0, 1\}^n$, it must be that
    \begin{equation}
        \label{event:break-uf-rec}
        (\xi, \sigma_{t-1}) \in \Trans(\cA^{\RDS.\Sign_{\RDS.\pk}}) \text{ and } \Phi(\zeta_1^*, \xi) = 1\;.
    \end{equation}
    with all but negligible probability; otherwise, $\cA$ would break the unforgeable recoverability of $\RDS$. Therefore, we may further condition on this event.
    
    Together event~\ref{event:implies-rec}, which guarantees that either
$\Phi(\zeta_1^*, y_{t-1}) = 0$ or $\xi \neq y_{t-1}$, and event \ref{event:break-uf-rec}, which guarantees that $\Phi(\zeta_1^*, \xi) = 1$, imply that $\xi \neq y_{t-1}$; otherwise, $\zeta_1^*$ would be simultaneously $\Phi$-close to and $\Phi$-far from $\xi$, a contradiction.

When $\xi \neq y_{t-1}$, we have two subcases:
\begin{enumerate}
    \item If $\Phi(\zeta_1^*, y_{t-1}) = 1$, then with inverse-polynomial probability, $\RDS.\Rec(\vk,$ $\zeta_1^*,\sigma_{t-1}) \neq y_{t-1}$ despite the fact that
    $(y_{t-1}, \sigma_{t-1}) \in \Trans(\cA^{\RDS.\Sign_{\RDS.\sk}})$ and $\Phi(\zeta_1^*, y_{t-1}) = 1$. This contradicts the robust recoverability of $\RDS$.
    \item If $\Phi(\zeta_1^*, y_{t-1}) = 0$, then $\cA$ produces
    $(\zeta_1^*, \xi, y_{t-1})$ that violates the unique signature property
    of $\RDS$.
\end{enumerate}
\end{enumerate}
In both cases, we arrive at a contradiction, proving that the scheme is in fact unforgeably recoverable.
\end{proof}
\subsection{Instantiating Hamming watermarks}
\label{sec:watermark-instantiation}
We instantiate our watermarking scheme for high-entropy language models. In particular, we combine our RDS construction (Corollary~\ref{cor:recoverable-signature-concrete}) with an existing secret key block steganography scheme (Theorem~\ref{thm:secret-key-block-wm}) to obtain a watermark that is secret key recoverable, robust, and unforgeable in the presence of a constant fraction of Hamming errors.
\begin{corollary}[Concrete secret-key watermarking]
\label{cor:sk-watermark-concrete}
Let $\lambda \in \NN$ be the security parameter, and suppose there exists
\begin{itemize}
    \item A CRHF with output size $m = m(\lambda)$.
    \item A strongly unforgeable digital signature with signature size $k = k(\lambda)$.
    \item An ideal PRC with with error rate $\delta_{\PRC} \in (0, \frac{1}{4})$ and information rate $\rho \in (0, 1)$\footnote{Recall that for any error-rate $\delta_{\PRC} \in (0,1/4)$, there exists an ideal PRC with constant rate $\rho = \Omega(1)$, assuming subexponential hardness of LPN~\cite[Theorem 3]{alrabiah2024ideal}.}. 
\end{itemize}
Let $n = \poly(\lambda)$ such that $n \geq \lambda$ and $(m + k)/n = o(1)$ be the block size. Let $\epsilon < \delta_{\PRC}$ be any constant and let $Q$ be any language model satisfying the following entropy condition,
    \[\forall \pi \in \{0, 1\}^*, y \leftarrow \Response_n(\pi) \text{ satisfies } H_\infty(y) \geq 4\Big(\sqrt{\frac{1}{2}-\epsilon}\Big)n\;.\]
Let $\delta > 0$ be any constant such that 
\begin{align*}
    \delta < \min\left\{ \delta_{\RDS}(\rho),\delta_{\PRC} - \epsilon\right\}\;,
\end{align*}
where $\delta_{\RDS}(\rho) := \sup\{ \delta \in (0,1/2) \mid \rho > H_2^{\mathsf{BZ}}(\delta) \}$.
Then, our construction in Corollary~\ref{cor:watermark-recoverable} yields a secret key 
\[\big(\EBC{\Ham_{\delta}}{n}, 2n\big)\text{-robust and } \big(\EBC{\Ham_{\delta}}{n}, n\big)\text{-unforgeable}\]
\emph{recoverable} watermarking scheme for $Q$. 
\end{corollary}
\begin{proof} 
From the ideal PRC assumption and Theorem~\ref{thm:secret-key-block-wm}, for any error rate $\delta \leq \delta_{\PRC} - \epsilon$, there exists a block steganography scheme for $Q$ with block size $n$, message size $\rho n$. From the CRHF and strongly unforgeable digital signature assumptions and Corollary \ref{cor:recoverable-signature-concrete}, for any $\delta \leq \delta_{\RDS}(\rho)$, there exists a $\Ham_{\delta}$-recoverable digital signature scheme with message size $n$ and signature size $\rho n + o(n)$. We apply Theorem \ref{thm:watermark-main} and Corollary \ref{cor:watermark-recoverable} to get the result. Note that the block steganography message size and the digital signature size differ by $o(n)$. Our proof ignores this difference, but it could be made exact by assuming an ideal PRC with a slightly better information rate $\rho' > \rho$.
\end{proof}
Moreover, we leverage an existing public-key block steganography scheme (Theorem~\ref{thm:public-key-block-wm}) to obtain a public-key construction that is robust to a subconstant fraction of Hamming errors, subject to the restriction that the first~$\ell$ bits of each block are error-free, as captured by the predicate~$\Phi_\delta$ defined in Theorem~\ref{thm:public-key-block-wm}.

\begin{corollary}[Concrete public-key watermarking]
\label{cor:pk-watermark-concrete}
Let $\lambda \in \NN$ be the security parameter, and suppose there exists
\begin{itemize}
    \item A CRHF with output size $m = m(\lambda)$.
    \item A strongly unforgeable digital signature with signature size $k = k(\lambda)$.
\end{itemize}
 Let $Q$ be any language model such that for some $\ell = \poly(\lambda)$,
    \[\forall \pi \in \{0, 1\}^*, y\leftarrow \Response_\ell(\pi) \text{ satisfies } H_\infty(y) \geq \omega (\log \lambda) \;.\]
Let $n = \poly(\lambda)$ be any block size satisfying $\ell(m + k)/n = o(1)$. There exists a subconstant $\delta_{\mathrm{max}} = \Theta(\frac{1}{\ell\log n})$ such that for any $\delta = o(1)$ satisfying $\delta(\lambda) \leq \delta_{\mathrm{max}}(\lambda) \; \forall \lambda \in \NN$, there exists a public key
\[\big(\EBC{\Phi_\delta}{n}, 2n\big)\text{-robust and } \big(\EBC{\Ham_{\delta}}{n}, n\big)\text{-unforgeable}\]
\emph{recoverable} watermarking scheme for $Q$ in the random oracle model, where $\Phi_\delta$ is defined the same as in Theorem~\ref{thm:public-key-block-wm}.
\end{corollary}

\begin{proof}
Fix $\lambda \in \NN$ sufficiently large. By Theorem~\ref{thm:public-key-block-wm}, in the random oracle model there exist constants
$\rho \in (0,1)$ and $\delta^* \in (0,1/2)$ such that for any $\delta \le \delta^*$, 
there exists a public key $\Phi_\delta$-robust block steganography scheme for $Q$ with block size $n$ and 
message size $\rho (n/\ell - 1)$.

Moreover, the CRHF and strongly unforgeable digital signature assumptions, Corollary~\ref{cor:RDS-subconstant} guarantees that for any $\delta = o(1)$,
there exists a $\Ham_\delta$-robust digital signature scheme with signature size
\[
O(\delta n \log n) + m + k \le C \delta n \log n + o(n/\ell)
\]
for some constant $C>0$.

To apply our main theorem, we require that the signature fits as a message for the block steganography scheme:
\[
C \delta n \log n + o(n/\ell) \le \rho (n/\ell - 1)\;.
\]
This condition holds whenever
\[\delta(\lambda) \leq \delta_{\mathrm{max}}(\lambda) \text{ where } \delta_{\mathrm{max}} = \frac{\rho(n/\ell) - \rho - o(n/\ell)}{C n\log n} = \Theta(\frac{1}{\ell \log n})\;\]

Therefore, there exists a $\Ham_{\delta(\lambda)}$-recoverable signature scheme with message size $n$ and signature size at most $\rho (n/\ell - 1)$. Note we can adjust $\delta_{\mathrm{max}}$ so that additionally, $\delta(\lambda) \le \delta^*$ for all $\lambda$, since $\delta^*$ is constant.  Then both a $\Phi_\delta$-robust block steganography scheme and a $\Ham_\delta$-recoverable signature scheme exist.  

Furthermore, for any strings $y, y' \in \{0,1\}^*$, we have
\[
\Phi_\delta(y, y') = 1 \implies \Ham_\delta(y, y') = 1,
\]
because if at most $\delta \big(\frac{n}{\ell}-1\big)$ blocks of size $\ell$ differ between inputs $y$ and $y'$, then at most $\delta (n-\ell) \le \delta n$ bits differ.  
It follows that the robust signature is also $\Phi_\delta$-robustly recoverable. Applying Theorem~\ref{thm:watermark-main} with these constructions then yields the desired result.
\end{proof}

\section*{Acknowledgements}
This work was supported in part by the Simons Collaboration on the Theory of Algorithmic Fairness. Part of this work was conducted while some of the authors were visiting the Simons Institute for the Theory of Computing.

%% file: sections/appendix.tex
\section{Strong Difference Recovery and Homomorphic PPH}
\label{sec:homomorphic-pph}
We notice that the PPH construction of Holmgren et al.~\cite{holmgren2022nearly}, when instantiated with a \emph{homomorphic} collision-resistant hash function (CRHF), satisfies two notable properties: it supports strong difference recovery and it is homomorphic. We formalize both properties and state the guarantees. 

Note, however, that our recoverable watermarking scheme does not rely on strong recovering or homomorphic PPHs, as the perturbed input $x'$ is available in the clear. Instead, the notion of \emph{sharp sketches}, introduced in~\refSc{sec:sharp-sketch}, more accurately captures the relevant functionality. We include strong difference recovery and homomorphic PPHs here primarily for completeness.

\begin{definition}[Strong difference recovering PPH]\label{def:strong-PPH-recoverability}
    A $\Phi$-PPH $(\Sample, \Hash, \Eval)$ is strong difference-recovering if there exists a PPT algorithm $\Recover$ such that for any PPT adversary $\mathcal{A}$,
\begin{align*}
    \Pr\left[
        \begin{array}{rl}
        h &\leftarrow \Sample(1^\lambda) \\ (x, x')&\leftarrow \cA(1^\lambda, h)\\
        \Delta &\leftarrow \Recover(h, h(x), h(x'))
        \end{array}
        :
        \begin{array}{rl}
        \big(\Delta \neq x \oplus x' \wedge \Phi(x, x') = 1\big) \\
        \vee\; \big(\Delta \neq \bot \wedge \Phi(x,x') = 0\big)
        \end{array}
    \right] \leq \negl(\lambda)
\end{align*}
\end{definition}
\begin{definition}[Homomorphic PPH]
\label{def:homomorphic-pph}
A PPH scheme $(\Sample, \Hash, \Eval)$ with input size $n$ and hash size $k$ is said to be \emph{homomorphic} if there exists a fixed abelian group $(\cG_n, +)$ and an injective map $\psi: \cG_n \to \{0,1\}^k$ such that for every hash function representation $h \in \cH$, there exists an efficiently computable\footnote{That is, computable using $\poly(n)$ many elementary algebraic operations.} homomorphism $\varphi: \ZZ^n \to \cG_n$ satisfying
\begin{align*}
    \Hash(h, x) = \psi(\varphi(x)) \quad \text{for all } x \in \{0,1\}^n\;.
\end{align*}
In other words, every hash function $h: \{0,1\}^n \to \{0,1\}^k$ in the family $\cH$ can be expressed as the composition of a homomorphism $\varphi: \ZZ^n \to \cG_n$ and a fixed injective embedding $\psi: \cG_n \to \{0,1\}^k$.
\end{definition}

\paragraph{Notation for homomorphic PPH.} If $\PPH[\Sample,\Hash,\Eval]$ is homomorphic, we abuse notation and write $h(x) - h(x') = h(x - x')$ for any $h \in \cH$ and $x, x' \in \{0,1\}^n$, where subtraction is interpreted over $\ZZ^n$ (so that $x-x' \in \{-1,0,1\}^n$) and $h(x)$ refers to the group element $\varphi(x) \in \cG_n$ (i.e., before applying the embedding $\psi$ into $\{0,1\}^k$).

\begin{definition}[Sparsity testing]
A homomorphic PPH scheme $(\Sample, \Hash, \Eval)$ is said to support \emph{sparsity testing} if there exists an algorithm $\mathsf{SparseEval}$ such that for all $h \in \cH$ and $x, x' \in \{0,1\}^n$,
\begin{align*}
    \Eval(h, h(x), h(x')) &= \mathsf{SparseEval}(h, h(x) - h(x')) \\
    &= \mathsf{SparseEval}(h,h(\Delta))\;,
\end{align*}
where $\Delta = x - x' \in \{-1, 0, 1\}^n$, and the final equality follows from the homomorphic property.
\end{definition}

In other words, if a homomorphic PPH supports sparsity testing, then its evaluation algorithm $\Eval$ depends only on the hash function $h$ and the difference $h(x)-h(x')=h(x-x')$, rather than on the individual hash values. Specifically, $\Eval$ returns $1$ if the difference $\Delta$ is sparse, and $0$ otherwise. 

Holmgren et al.'s homomorphic PPH goes further by supporting \emph{sparse recovery}: if $\Delta$ is sparse, then given only $(h, h(\Delta))$, their recovery algorithm $\Recover$ outputs $\Delta \in \{-1,0,1\}^n$ exactly. Note that this recovers only the difference vector, not the preimages $x$ or $x'$.

\begin{theorem}[{\cite[Corollary 6.7]{holmgren2022nearly}}]
\label{thm:homomorphic-pph-hamming}
Let $\lambda \in \NN$ be the security parameter, and let $n = \poly(\lambda)$ denote the input size. Suppose there exists a homomorphic CRHF with output size $\ell = \ell(\lambda)$ that is collision-resistant over $\{-1,0,1\}^n$~\cite[Definition 6.1]{holmgren2022nearly}. Then, for any constants $\delta \in (0,1/2)$ and $\rho > H_3^{\mathsf{BZ}}(\delta) \cdot \log_2 3$, where $H_3^{\mathsf{BZ}}(\cdot)$ refers to the Blokh-Zyablov bound~\cite[Fact 2.12]{holmgren2022nearly}, there exists a \emph{strongly recoverable} homomorphic PPH over the input domain $\{0,1\}^n$ for $\Ham_{\delta}$ with hash size $\rho \cdot n + \ell$.
\end{theorem}

\section{Sharp Sketches and Lower Bounds}
\label{sec:sharp-sketch}
A \emph{sharp sketch} is a cryptographic primitive that enables efficient recovery given auxiliary input $x'$ in the clear. The goal is to reconstruct an input $x$ from its sketch $g(x)$ and a candidate $x'$ that is sufficiently close to it. Its defining feature is \emph{sharpness}: $\Recover(g, g(x), x')$ must output $\bot$ if $x'$ is far from $x$.

This notion is inspired by the secure sketches of Dodis et al.~\cite{dodis2008fuzzy}, but differs in both purpose and security guarantees. Secure sketches require \emph{unconditional} hardness of inversion: it must be hard to recover $x$ from its sketch $g(x)$, even for an unbounded adversary.  This necessitates a non-trivial compression rate $\rho = k/n$, so secure sketches are compressive by definition. Their recoverability guarantee is also \emph{unconditional}, requiring recovery to succeed for \emph{every} $x'$ that is close to $x$.

In contrast, sharp sketches enforce a different notion of security, one not centered on inversion. They require that no \emph{computationally-bounded} adversary can produce input pairs $(x, x')$ that violate the sharp recoverability condition (see~\refDef{def:sharp-sketch}). Thus, compressiveness is not required by definition. Even the identity sketch $g(x) = x$ qualifies as a (trivial) sharp sketch, though it would not qualify as a secure sketch. On the other hand, any secure sketch can be converted into a sharp sketch by returning $\bot$ whenever $\hat{x} = \Recover(g, g(x), x')$ is either inconsistent $g(x) \neq g(\hat{x})$ or far from $x'$.

Sharp sketch constructions for Hamming predicates were given by Dodis et al.~\cite[Section 8.3]{dodis2008fuzzy}. We include it here for completeness and to enable precise reference in our watermarking scheme, where it plays a key role. Our main technical contribution is a generalization of the lower bound by Holmgren et al.~\cite{holmgren2022nearly}, extending it from Hamming predicates to any predicate $\Phi$ that satisfies an efficient sampleability condition (\refThm{thm:sharp-sketch-lower-bound}). When specialized to the Hamming case, our bound implies that the construction of Dodis et al.~achieves nearly optimal sketch size, since the gap between the Hamming bound $H(\delta)$ and the Blokh-Zyablov bound $H_2^{\mathsf{BZ}}(\delta)$ is small (see also Remark~\ref{rem:blokh-zyablov}). This shows that the recoverable digital signature in Corollary \ref{cor:recoverable-signature-concrete} is nearly optimal compared to constructions based on the generic PPH/sharp sketch transformation of Section \ref{sec:rds-construction}.

\begin{definition}[Sharp sketch]
\label{def:sharp-sketch}
Let $\lambda \in \NN$ be the security parameter, let $\Phi$ be any predicate, and let $n = n(\lambda)$ and $k = k(\lambda)$ be such that $k \leq n$. A sharp sketching scheme for $\Phi$ with input size $n$ and sketch size $k$ is a triple of PPT algorithms $(\Sample, \Sketch, \Recover)$ defined as follows:

\begin{itemize}
    \item $\Sample(1^\lambda)$: samples a sketching function $g$ from a family of functions $\cG_\lambda$.
    \item $\Sketch(g, x)$: takes in $x \in \{0,1\}^n$ and outputs its sketch $g(x) \in \{0,1\}^k$.
    \item $\Recover(g, z, x')$: takes in $z \in \{0,1\}^k$ and $x' \in \{0,1\}^n$, and returns an output in $\{0,1\}^n \cup \{\bot\}$.
\end{itemize}
The scheme satisfies the following security guarantee:

\medskip
\noindent
\emph{\textbf{Sharp recoverability:}} For any PPT adversary $\cA$,
\begin{align*}
    \Pr\left[
        \begin{array}{rl}
            g &\leftarrow \Sample(1^\lambda); \\
            (x, x') &\leftarrow \cA(1^\lambda, g)
        \end{array}
        :
        \begin{array}{rl}
            \big(\Recover(g, g(x), x') \neq x \;\wedge\; \Phi(x, x') = 1\big) \\
            \vee\; \big(\Recover(g, g(x), x') \neq \bot \;\wedge\; \Phi(x, x') = 0\big)
        \end{array}
    \right] \leq \negl(\lambda)\;.
\end{align*}
\end{definition}
We restate Theorem~\ref{thm:sharp-sketch-hamming}, which gave a difference recovery hash family for Hamming from CRHFs, in the language of sharp sketches.
\begin{theorem}[{Sharp sketch for Hamming, adapted from~\cite[Section 8.3]{dodis2008fuzzy}}]
Let $\lambda \in \NN$ be the security parameter, and let $n=n(\lambda)$ denote the input size. Suppose there exists a CRHF with output size $\ell=\ell(\lambda)$. For any constants $\delta \in (0,1/2)$ and $\rho > H_2^{\mathsf{BZ}}(\delta)$, where $H_2^{\mathsf{BZ}}(\cdot)$ refers to the Blokh-Zyablov bound, there exists a sharp sketching scheme $(\Sample, \Sketch, \Recover)$ over $\{0,1\}^n$ for $\Ham_{\delta}$ with sketch size $\rho \cdot n + \ell$.
\end{theorem}

\begin{remark}[PPH and sharp sketch]
    By Proposition~\ref{prop:pph-recover-app}, every Hamming PPH yields a sharp sketch for the same predicate. However, the converse does not hold. For instance, the sharp sketching scheme in~\refThm{thm:sharp-sketch-hamming} does not give rise to a Hamming PPH, as it crucially relies on receiving $x'$ in the clear, rather than just its sketch. Note, however, that our recoverable digital signature algorithm (Section~\ref{sec:recoverability}) \emph{does} receive one input in the clear, and therefore sharp sketches suffice.
\end{remark}

\begin{lemma} [\cite{dodis2008fuzzy}, Lemma 2.2]
\label{lem:cond-entropy} For random variables $X, Y, Z$, where $Y$ is supported over a set of size $k$, we have
\[H_\infty(X | Y, Z) \geq H_\infty(X |Z) - \log k\]
\end{lemma}

\begin{theorem}[Sharp sketch lower bound]
\label{thm:sharp-sketch-lower-bound}
Let $\Phi$ any predicate and let $A \subseteq \{(x,x') \in \{0,1\}^n \times \{0,1\}^n: \Phi(x,x') = 1\}$ be such that the uniform distribution $\cU$ over pairs $(x,x')$ conditioned on $A$ is efficiently sampleable. Then the sketch size of any sharp sketch for $\Phi$ is at least 
\begin{align*}
    n + \log \cU(A) \;.
\end{align*}
\end{theorem}
\begin{proof}
Let $X$ be uniformly distributed over $\{0,1\}^n$, and let $g: \{0,1\}^n \to \{0,1\}^k$ be a sharp sketch drawn independently of $X$. Then the min-entropy of $X$ conditioned on $h$ and $g(X)$, which we denote by $H_\infty(X\mid h,g(X))$, satisfies:
\begin{align*}
    H_\infty(X \mid g, g(X)) \ge H_\infty(X \mid g) - k \ge n - k\;.
\end{align*}

In particular, no (even unbounded) adversary can ``invert'' $(g, g(X))$ to recover the original input $X$ with probability greater than $2^{-(n - k)}$ when $X$ is drawn uniformly from $\{0,1\}^n$.

Now consider the following inversion strategy based on $\Recover(g,g(x),\cdot)$. Define a randomized algorithm $\Inv$ that, on input $(g, g(x))$, samples a fresh $x' \leftarrow \{0,1\}^n$ and outputs $\Recover(g, g(x), x')$. Let $\mu$ be the distribution of $(x, x')$ sampled uniformly from $\{0,1\}^{2n}$ conditioned on the event $A$. By assumption, $(x,x') \in A$ implies $\Phi(x,x') = 1$. Since $\mu$ is efficiently sampleable and $\Phi(x,x')=1$, guarantees of the sharp sketch $g$ yield
\begin{align*}
    \Pr\left[
    \begin{array}{rl}
        g &\leftarrow \Sample(1^\lambda) \\
        (x, x') &\leftarrow \mu
    \end{array}
    :
    \Recover(g, g(x), x') = x
    \right] \ge 1 - \negl(\lambda)\;.
\end{align*}

Using the above recovery guarantee, we lower bound the success probability of $\Inv$ as follows.
\begin{align*}
    \Pr&\bigg[
    \begin{array}{rl}
        g &\leftarrow \Sample(1^\lambda)  \\
        x &\leftarrow \{0,1\}^n 
    \end{array}
    :\;
    \Inv(g,g(x)) = x
    \bigg] \\
    &\ge \Pr\bigg[
    \begin{array}{rl}
        g &\leftarrow \Sample(1^\lambda)  \\
        (x, x') &\leftarrow \{0,1\}^{2n} 
    \end{array}
    :
    \begin{array}{rl}
        \Recover(g,g(x),x') = x \;\bigg\vert\; A
    \end{array}
    \bigg] \cdot \cU(A) \\
    &\ge (1-\negl(\lambda)) \cdot \cU(A)\;.
\end{align*}

But from the earlier information-theoretic bound, no function $g$ with output size $k$ can support inversion with success probability greater than $2^{-(n - k)}$. Hence, we conclude
\begin{align*}
    2^{-(n - k)} \ge \cU(A) \cdot (1 - \negl(\lambda))\;,
\end{align*}

For sufficiently large $\lambda$, we have $\negl(\lambda) < 1/2$, which implies
\begin{align*}
    k \ge n + \log\cU(A) \;.
\end{align*}
\end{proof}
\begin{claim}
Let $r = r(n)$ be any efficiently computable function. Then, the uniform distribution $\cU$ over pairs $(x,x') \in \{0,1\}^n \times \{0,1\}^n$ conditioned on $\|x-x'\|_0 = r$ is efficiently sampleable.
\end{claim}
\begin{proof}
First, sample $x \leftarrow \{0,1\}^n$ uniformly at random. Then sample an $r$-sparse perturbation vector $\Delta \in \{0,1\}^n$ independently of $x$. Finally, set $x' = x \oplus \Delta$.
\end{proof}

\begin{corollary}
\label{cor:hamming-sketch-lower-bound}
Let $r = r(n)$ be any efficiently computable function. Then, for all sufficiently large $n$, the sketch size of any sharp sketching scheme for $\Ham_{r/n}$ is at least $\log \binom{n}{r}$.
\end{corollary}